\documentclass[11pt]{article}
\usepackage[top=1in, bottom=1in, left=1in, right=1in]{geometry}

\usepackage[style=alphabetic,natbib=true,maxbibnames=99,backend=biber]{biblatex}

\usepackage[utf8]{inputenc}
\usepackage{mathrsfs}
\usepackage{graphicx}
\usepackage{savesym}
\savesymbol{Bbbk}
\usepackage{float}

\usepackage{xcolor}
\usepackage{thmtools}
\usepackage{thm-restate}

\usepackage{algorithm}
\usepackage{algpseudocode}[1]

\usepackage[normalem]{ulem}
\usepackage{amsthm,amsmath,amssymb,amsfonts}
\usepackage{hyperref, cleveref}
\usepackage{mathtools,pifont,enumitem}
\usepackage{dsfont}
\usepackage{subfigure}
\usepackage{nicefrac}
\usepackage{caption}

\newtheorem{theorem}{Theorem}[section]
\newtheorem{lemma}[theorem]{Lemma}
\newtheorem{definition}[theorem]{Definition}

\newtheorem{fact}[theorem]{Fact}

\newtheorem{remark}[theorem]{Remark}
\newtheorem{example}[theorem]{Example}
\newtheorem{problem}[theorem]{Problem}

\crefname{section}{Section}{Sections}

\crefname{theorem}{Theorem}{Theorems}
\crefname{assumption}{Assumption}{Assumptions}
\crefname{lemma}{Lemma}{Lemmas}
\crefname{definition}{Definition}{Definitions}
\crefname{conjecture}{Conjecture}{Conjectures}
\crefname{corollary}{Corollary}{Corollaries}
\crefname{construction}{Construction}{Constructions}
\crefname{claim}{Claim}{Claims}
\crefname{observation}{Observation}{Observations}
\crefname{proposition}{Proposition}{Propositions}
\crefname{fact}{Fact}{Facts}
\crefname{question}{Question}{Questions}
\crefname{problem}{Problem}{Problems}
\crefname{remark}{Remark}{Remarks}
\crefname{example}{Example}{Examples}
\crefname{equation}{Equation}{Equations}
\crefname{appendix}{Appendix}{Appendices}
\crefname{algorithm}{Algorithm}{Algorithms}
\crefname{model}{Model}{Models}

\newcommand{\yesnum}{\addtocounter{equation}{1}\tag{\theequation}}

\makeatletter
\newcommand{\customlabel}[2]{%
\protected@write \@auxout {}{\string \newlabel {#1}{{#2}{\thepage}{#2}{#1}{}} }%
\hypertarget{#1}{}
}
\makeatother

\newcommand{\white}[1]{\textcolor{white}{#1}}
\newcommand{\gray}[1]{\textcolor{lightgray}{#1}}

\newcommand{\N}{\mathbb{N}}

\newcommand{\R}{\mathbb{R}}
\newcommand{\Z}{\mathbb{Z}}

\newcommand{\cD}{\mathcal{D}}

\newcommand{\cM}{\mathcal{M}}
\newcommand{\cP}{\mathcal{P}}
\newcommand{\cR}{\mathcal{R}}
\newcommand{\cS}{\mathcal{S}}

\newcommand{\cU}{\mathcal{U}}

\newcommand{\evE}{\ensuremath{\mathscr{E}}}
\newcommand{\evF}{\ensuremath{\mathscr{F}}}
\newcommand{\evG}{\ensuremath{\mathscr{G}}}

\newcommand{\wt}{\widetilde}

\newcommand{\st}{\mathrm{s.t.}}

\newcommand{\eps}{\varepsilon}
\renewcommand{\epsilon}{\varepsilon}

\newcommand{\argmax}{\operatornamewithlimits{argmax}}
\newcommand{\Ex}{\operatornamewithlimits{\mathbb{E}}}

\newcommand{\poly}{\mathop{\mbox{\rm poly}}}

\def\abs#1{\left| #1 \right|}

\newcommand{\inparen}[1]{\left(#1\right)}
\newcommand{\inbrace}[1]{\left\{#1\right\}}
\newcommand{\insquare}[1]{\left[#1\right]}

\newcommand{\floor}[1]{\left\lfloor#1\right\rfloor}
\newcommand{\sfloor}[1]{\lfloor#1\rfloor}
\newcommand{\ceil}[1]{\left\lceil#1\right\rceil}

\newcommand{\zo}{\{0,1\}}

\newcommand{\Stackrel}[2]{\stackrel{\mathmakebox[\widthof{\ensuremath{#2}}]{#1}}{#2}}

\newcommand{\wh}[1]{\widehat{#1}}

\newcommand{\hD}{\widehat{D}}
\newcommand{\hM}{\widehat{M}}

\newcommand{\tD}{\wt{D}}

\newcommand{\np}{\textbf{NP}}
\newcommand{\negsp}{\hspace{-0.5mm}}

\newif\ifconf
\conftrue

\ifconf

\else

\fi

\ifconf

\else

\fi

\newcommand{\prog}[1]{Program~\eqref{#1}}

\newcommand{\gfviol}{\mathscr{G}_{\rm violation}}
\newcommand{\ifviol}{\mathscr{I}_{\rm violation}}
\newcommand{\util}{\mathscr{U}}

\title{Sampling Individually-Fair Rankings that are \textit{Always} Group Fair}

\author{Sruthi Gorantla\thanks{Equal contribution}\\ IISc, Bengaluru \and Anay Mehrotra\(^*\)\\ Yale University \and Amit Deshpande\(^\dag\)\\ MSR India \and Anand Louis\thanks{Equal contribution}\\ IISc, Bengaluru}

\begin{document}

\maketitle

\begin{abstract}
    Rankings on online platforms help their end-users find the relevant information---people, news, media, and products---quickly.
    Fair ranking tasks, which ask to rank a set of items to maximize utility subject to satisfying group-fairness constraints, have gained significant interest in the Algorithmic Fairness, Information Retrieval, and Machine Learning literature.
    Recent works, however, identify uncertainty in the utilities of items as a primary cause of unfairness and propose introducing randomness in the output.
    This randomness is carefully chosen to guarantee an adequate representation of each item (while accounting for the uncertainty).
    However, due to this randomness, the output rankings may violate group fairness constraints.
    We give an efficient algorithm that samples rankings from an individually-fair distribution while ensuring that \emph{every} output ranking is group fair.
    The expected utility of the output ranking is at least $\alpha$ times the utility of the optimal fair solution.
    Here, $\alpha$ depends on the utilities, position-discounts, and constraints---it approaches 1 as the range of utilities or the position-discounts shrinks, or when utilities satisfy distributional assumptions.
    Empirically, we observe that our algorithm achieves individual and group fairness and that Pareto dominates the state-of-the-art baselines.
\end{abstract}

\newpage
\tableofcontents
\newpage

\section{Introduction}\label{sec:intro}

    Rankings are ubiquitous on online platforms and have become a quintessential tool for users to find relevant information \cite{IRbook,liu2011learning,Epstein2015,Noble2018,googleLTR}.
    The core algorithmic problem in generating a ranking, given $m$ items (denoting individuals, products, or web pages) is to select and order a subset of $n$ items that are ``most'' relevant to the given query \cite{IRbook,liu2011learning,googleLTR}.
    This is a fundamental problem in Information Retrieval and has been extensively studied in the Machine Learning literature \cite{IRbook,liu2011learning,googleLTR}.

    Without any fairness considerations, rankings on online platforms have been observed to have skewed representations of certain demographic groups resulting in large-scale perpetuation and amplification of fairness-related harms \cite{KayMM15, Noble2018}.
    Skewed rankings can have adverse effects both at the group level--altering the end-users' perception of socially-salient groups \cite{KayMM15} and polarizing their opinions \cite{Epstein2015,polarizationWSJ2020}--and at an individual level--leading to a denial of economic opportunities to individuals (in later positions) \cite{hannak2017bias}.
    A reason for this is that the estimated relevance (or utilities) of items may be influenced by societal biases leading to skews affecting socially-salient, and often legally protected, groups such as women and people of color.
    Another reason for underrepresentation is that the utility estimates used to generate the ranking are bound to have some uncertainty, which leads to over-estimation or underestimation of utilities for different items--at an individual level.

    A large body of work designs algorithms to generate rankings that ensure sufficient representation \cite{causal2021yang,  policyLearningAshudeep, ReducingDisparateExposureZehlike, robustFairLTR2020, YangS17, celis2018ranking, BalancedRankingYang2019, linkedin_ranking_paper, GorantlaUnderranking21, fairExposureAshudeep, AmortizedFairnessBiega2018} (also see the surveys \cite{fair_ranking_survey1,fair_ranking_survey2,overviewFairRanking,criticalReviewFairRanking22}).
    A significant fraction of these works focus on group-level representation and have considered several types of group fairness constraints \cite{causal2021yang,  policyLearningAshudeep, ReducingDisparateExposureZehlike, robustFairLTR2020, YangS17, celis2018ranking, BalancedRankingYang2019, linkedin_ranking_paper, GorantlaUnderranking21, fairExposureAshudeep}.
    Two popular ones are equal representation and proportional representation.
    In the case of two groups $G_1$ and $G_2$, equal representation with a parameter $k$ requires that for every $j$ (roughly) $\frac{k}{2}$ items from each group appear between the $(kj+1)$-th and the $(kj+k)$-th positions \cite{GorantlaUnderranking21}.
    Here, for instance, $k$ could denote the number of items on each ``page'' of the ranking or the number of items in a user's browser window.
    Proportional representation requires that, for every $j$, $k\frac{\abs{G_\ell}}{m}$ items appear between the $(kj+1)$-th and the $(kj+k)$-th positions.
    Other constraints, that generalize equal representation and proportional representation, and notions of fairness from the perspective of other stakeholders (such as the end-user) have also been considered \cite{fair_ranking_survey1,fair_ranking_survey2,overviewFairRanking,criticalReviewFairRanking22} {(also see \cref{sec:model}).}
    Broadly speaking, all of these works, given group fairness constraints, output a ranking that has the maximum relevance or \textit{utility} subject to satisfying the specified constraints.

    Ensuring group-wise representation, via such group fairness constraints, can address underrepresentation across groups of items but may not address harms at an individual level:
    Across multiple output rankings, specific items (e.g., whose utilities have high uncertainty) may be systematically assigned lower positions.
    In other words, group fairness can mitigate under-representation at the group level but may not mitigate (or could even exacerbate) misrepresentation and denial of opportunities to individual items.

    \begin{example}[\textbf{Insufficiency of group-fairness constraints}]\label{ex:exampleIntro}
        As a concrete example consider an online hiring platform where recruiters search for relevant candidates and are presented with a ranked list of candidates; as is common in existing recruiting platforms \cite{linkedin_ranking_paper}.
        Suppose that this platform ensures proportional representation across individuals with, say, different skin tones.
        Here, it can be shown that, to maximize the ``utility,'' it is optimal to order individuals inside one group (those with the same skin tone) in decreasing order of their (estimated) utility.
        Consider two individuals $i_1$ and $i_2$ with the same skin tone and estimated utilities $\rho$ and $\rho-\eps$ (for some small constant $\eps>0$).
        Due to the difference in their utility, $i_2$ would always be ranked one or more positions below $i_1$.
        Since positions of individuals on recruiting platforms have been observed to affect their chances of being hired, $i_2$ has a systematically lower chance of being hired -- even though there is little difference in their utility \cite{hannak2017bias}.
        Moreover, this difference may be because of estimation errors that are bound to arise in any real-world setting and especially in the context of online recruiting where the utilities of \mbox{individuals are inherently uncertain and even change over time.}
    \end{example}

    \noindent Motivated by such examples, recent work on fair ranking has proposed various ways to define and incorporate fairness -- from the perspective of individuals -- in rankings \cite{fairExposureAshudeep,AmortizedFairnessBiega2018,AshudeepUncertainty2021}.
    Since opportunity, exposure, or attention received by individuals is ultimately linked to positions in the ranked order, a single deterministic ranking cannot avoid denial of opportunity when ranking multiple items with similar relevance, and hence, individually-fair rankings are inevitably stochastic in nature \cite{fairExposureAshudeep,AmortizedFairnessBiega2018,AshudeepUncertainty2021}.
    To gain intuition, observe that in the earlier example, any deterministic ranking must place either $i_1$ before $i_2$ or $i_2$ before $i_1$, due to which the ``exposure'' received by the item placed earlier is systematically higher than the other item irrespective of how small the difference in their utilities (i.e., $\eps$) is.
    While there are many notions of individual fairness with respect to items, their specification often boils down to specifying lower and/or upper bounds on the probability with which an individual or item must appear in a set of positions.
    For instance, an individual fairness constraint, specified by a matrix $C$, may require item $i$ to appear between the $(kj+1)$-th and the $(kj+k)$-th position with probability at least $C_{ij}$.
    Where, as before, $k$ could encode the number of items in one page of the ranking, in which case, the individual fairness constraint requires item $i$ to appear on page $j$ with at least a specified probability for every $i$ and $j$.
    While this stochasticity guarantees that individuals with ``similar'' utilities (as in the above example) receive similar average exposure, due to their stochastic nature, specific output rankings may violate group fairness requirements (as we empirically verify in \cref{sec:empirical}).

    Given the importance of both individual and group fairness in ranking, we study a dual-task in fair ranking wherein addition to individual fairness, we want to ensure that every output ranking is group fair.
    It is important to note that stochastic rankings that incorporate group fairness guarantee \emph{in expectation} may not output rankings that are \emph{always} group fair.
    This is particularly concerning in high stake contexts (such as online recruiting) where it may be legally required to ensure group fairness for each output ranking.
    Thus, the following question arises:
    \emph{Given the individual fairness constraints, the group fairness constraints, and item utilities, is there an algorithm that outputs samples rankings such that (1) individual fairness is satisfied, (2) each output ranking is group fair, and (2) the expected utility of the rankings is maximized?}

    \subsection{Our Contributions}
        We present an efficient approximation algorithm (\cref{alg:main}) for the above problem when the (socially salient) groups of items form a laminar set family (i.e., any two groups are either disjoint or related by containment) (\cref{sec:overview}).
        {This algorithm works for a general family of individual and group fairness constraints, which includes the aforementioned constraints and their generalizations (\cref{def:group_constraints,def:individual_constraints}).}
        For any given individual and group fairness constraints {from these families} along with the utilities of all items, our algorithm outputs rankings sampled from a distribution such that the specified individual fairness are satisfied and each output ranking is group fair (\cref{thm:algo_main}).
        The rankings output by our algorithm have an expected utility that is at least $\alpha$ times the optimal utility, where $\alpha$ is a constant that depends on the utilities, position discounts, and group fairness constraints--it approaches 1 when the ranges of the utilities or of the position-discounts shrink (\cref{eq:strongerInequality}).
        In particular, for the aforementioned constraints specified by a parameter $k$, %
        $\alpha\geq \frac{v_1+v_2+\dots+ v_k}{k\cdot v_1},$ where $v_j$ is a position-discount for each position (\cref{thm:algo_main}).
        With the standard DCG-discounts and equal representation constraints for two groups (for which it suffices to have $k=2$), this approximation guarantee becomes $\alpha \geq 0.81$ {(see \cref{sec:theoretical_results} for other common examples).}
        Further, in addition to these utility-independent bounds, we also derive additional (stronger) bounds on $\alpha$ when the item's utilities are generated via certain generative models (\cref{thm:approxStochasUtil}).

        Empirically, we evaluate our algorithm on synthetic and real-world data against standard group fairness metrics (such as equal representation) and the individual fairness constraints proposed by \citet{AshudeepUncertainty2021}.
        We compare the performance of our algorithm to key baselines \cite{fairExposureAshudeep,celis2018ranking} with both two and multiple protected groups.
        Unlike baselines, in all simulations, our algorithm outputs rankings that always satisfy the specified individual fairness constraint and group fairness constraint; at a small cost to utility (a maximum of 6\% loss compared to the baselines) (Figures~\ref{fig:fairness} and \ref{fig:utility}).

        To the best of our knowledge, there is no previously known algorithm that takes a stochastic fair ranking satisfying fairness constraints \emph{in expectation} and rounds it to output rankings that are \emph{always} group fair without much loss in the ranking utility. A key technical challenge in doing so is that the Birkhoff-von Neumann rounding of stochastic fair rankings (as used in \citet{AshudeepUncertainty2021}) can violate group fairness constraints significantly. Overcoming this challenge requires a generalization of the Birkhoff-von Neumann rounding from the polytope of all rankings (that has only integral vertices) to the polytope of group-fair rankings (that can have fractional vertices).
        Stochastic rankings that satisfy group-wise representation constraints in the top-$k$ positions \emph{in expectation}, typically have a standard deviation of about $\sqrt{k}$ (e.g., Theorem 4.1 by \citet{mehrotra2022fair}).\footnote{Concretely, consider a ranking $R$ sampled from some distribution such that $R$ satisfies the equal representation constraints in expectation for two groups. The best guarantee provided by state-of-the-art fair ranking algorithms that sample a ranking \cite{mehrotra2022fair} is that, with high probability, the output $R$ places at most $\frac{k}{2}+O(\sqrt{k})$ items from each group in the top-$k$ positions--thereby violating the constraint by up to an additive factor of $O(\sqrt{k})$.}
        As $k$ is small in practice (e.g., $k\approx 10$ on LinkedIn), a deviation of $\sqrt{k}$ in group-wise representation is impractical.

    \paragraph{Outline.}
        The rest of the paper is organized as follows.
        \cref{sec:relatedWork} presents a brief overview of relevance estimation in Information Retrieval and Machine Leaning literature and works (in Algorithmic Fairness, Information Retrieval, and Machine Learning) that design algorithms for fair ranking.
        \cref{sec:model} presents preliminaries, defines the families of fairness constraints we consider (\cref{def:group_constraints,def:individual_constraints}), and the technical challenges (\cref{sec:challenges}).
        \cref{sec:theoretical_results} presents our theoretical results (\cref{thm:algo_main,thm:approxStochasUtil}) and overviews our algorithm (\cref{sec:algorithm_overview}).
        \cref{sec:empirical} presents empirical evaluation of \cref{alg:main} and compares its performance to state-of-the-art fair ranking baselines.
        \cref{sec:theoretical_results} presents the key proof ideas in our main theoretical result.
        \cref{sec:conclusion} presents the main limitations and concludes.

    \section{Related Work}\label{sec:relatedWork}
        \paragraph{Relevance estimation for ranking.}
            There is a huge body of work studying relevance estimation for automated information retrieval  \cite{LiddyAutomatic05,cleverdon1991significance} (also see \citet{IRbook} and the references therein).
            This body of works develops methods to estimate the relevance (or utility) of items to specific queries in a variety of contexts (from web search \cite{bar2008random}, personalized feeds \cite{jeh2003scaling}, to e-commerce \cite{dave2003mining}) and modalities (from web pages \cite{kleinberg1999authoritative}, images and videos \cite{googleLTR}, to products \cite{dave2003mining}).
            In the last three decades, the Machine Learning literature has also made significant contributions to this body of works \cite{liu2011learning} -- by supplementing traditional IR methods (by, e.g., auto-tuning hard-to-tune parameters) \cite{taylor2006optimisation}, increasing the efficiency of IR methods (via clustering-based techniques) \cite{singitham2004efficiency,altingovde2008incremental}, and proposing novel neural-network-based methods to predict item relevance \cite{burges2010ranknet,burges2005learning,weston2010large,googleLTR}.
            That said, despite the numerous methods for relevance estimation, the relevance values output by any method is bound to have some uncertainty and have also been observed to propagate societal biases in their inputs \cite{KayMM15,linkedin_ranking_paper}.

        \paragraph{Fair ranking.}
        Below we summarize previous work on group fair and individually-fair rankings, various approaches to formulate and solve these problems, and their relation to our work. For a comprehensive survey of these topics, \mbox{we refer the reader to \cite{fair_ranking_survey1,fair_ranking_survey2,criticalReviewFairRanking22,overviewFairRanking, Castillo18}.}

        \smallskip\noindent\textit{Group fair ranking.}
        There is a long line of work on group fair rankings that can be divided into two broad categories: (1) those that incorporate group fairness in learning-to-rank (LTR) algorithms \cite{causal2021yang,policyLearningAshudeep,ReducingDisparateExposureZehlike,robustFairLTR2020,YangS17} and (2) (re-)ranking algorithms that modify a given output ranking to satisfy group fairness constraints \cite{celis2018ranking,BalancedRankingYang2019,linkedin_ranking_paper,GorantlaUnderranking21,fairExposureAshudeep,AmortizedFairnessBiega2018}.
        Furthermore, there are diverse approaches within each of the above categories.
            The group fair LTR works can be further subdivided as algorithms that (a) post-process the estimated utilities to ensure group fairness \cite{causal2021yang}, (b) add group fairness penalty in the LTR objective for training \cite{policyLearningAshudeep,ReducingDisparateExposureZehlike,robustFairLTR2020}, and (c) modify feature representation learned by up-stream systems so that the utilities learned from the modified representation satisfy group fairness \cite{YangS17}.
            The group-fair re-ranking works can be further subdivided based on whether they guarantee that (a) \emph{each} output ranking satisfies group fairness constraints \cite{celis2018ranking,BalancedRankingYang2019,linkedin_ranking_paper,GorantlaUnderranking21} or (b) the group fairness constraints are satisfied \emph{in aggregate} over multiple rankings \cite{fairExposureAshudeep,AmortizedFairnessBiega2018}.
        As highlighted before, rankings output by these works may lead to adverse effects on individuals due to uncertainties in utilities.

        \smallskip\noindent\textit{Individually fair ranking.}
        There are a number of notions of individual fairness in ranking.
        For instance, \citet{AmortizedFairnessBiega2018} define ``equity of attention'' as requiring that the cumulative attention garnered by an individual item across multiple rankings (corresponding to same or different queries) be proportional to its average relevance (across the corresponding queries).
        \citet{AmortizedFairnessBiega2018} propose an online algorithm to minimize the aggregate unfairness between attention and relevance for all items, amortized over multiple rankings, while maintaining the ranking utility (e.g., NDCG@k) above a given threshold; while they consider a notion of individual fairness, they do not consider group fairness.
        Other notions include fairness of exposure \cite{fairExposureAshudeep} and ``merit-based'' fairness--we discuss these below \cite{AshudeepUncertainty2021}.

        \smallskip\noindent\textit{Ranking under both individual and group fairness constraints:}
        Some of the aforementioned works offer frameworks that can be adapted to incorporate both individual and group fairness constraints \cite{policyLearningAshudeep,fairExposureAshudeep,AshudeepUncertainty2021}.
        \citet{fairExposureAshudeep} define fairness of exposure in stochastic rankings, which can be applied at both individual and group levels.
        They solve a linear programming relaxation over stochastic rankings to maximize the expected ranking utility subject to the fairness of exposure \emph{in expectation} \cite{fairExposureAshudeep}.
        This approach gets around the exponential search space of deterministic rankings (or permutations), and their final output is the Birkhoff-von Neumann rounding of the above stochastic ranking. %
        \citet{AshudeepUncertainty2021} define a notion of ``merit-based'' fairness when the merits (or utilities) are random variables.
        They take a similar linear-programming approach as \citet{fairExposureAshudeep} to formulate the fairness constraints and use the Birkhoff-von Neumann rounding to generate the output rankings. %
        However, unlike this work, these works either do not guarantee that the output rankings satisfy the fairness constraints or they use randomization and only guarantee that the output rankings satisfy group fairness constraints in aggregate (not \textit{always}).

        Some recent works step away from the paradigm of utility maximization to incorporate individual fairness and group fairness constraints \cite{Soriano2021Maxmin,GorantlaUnderranking21,bailey2021fair}.
        \citet{Soriano2021Maxmin} propose a polynomial time (re-)ranking algorithm to maximize the utility of the worst-off individual subject to group fairness constraints. They also show that probabilistic rankings give better max-min fairness than deterministic rankings.
        \citet{GorantlaUnderranking21} define individual fairness in terms of the worst-case ``underranking'' of any item compared to its true or deserved rank, and give efficient (re-)ranking algorithms for given group fairness and underranking constraints simultaneously.
        {In the special case of selection, where items only have to be selected and their order is not relevant, \cite{bailey2021fair} select subsets maximizing a specified individual-fairness metric subject to satisfying group-fairness constraints.}
        Unlike these works, we require the output ranking to maximize the utility subject to satisfying the specified (group and individual) fairness constraints.
        Beyond ranking {and selection,} there are also works that incorporate fairness constraints in matching problems (where multiple items can be matched to one position) \cite{Chierichetti0LV19,kempe2023matching,benabbou2018diversity}.
        Among these \citet{benabbou2018diversity} consider block-wise group-fairness constraints that are similar to the block-wise group fairness constraints we consider (\cref{def:group_constraints}) and design $\frac{1}{2}$-approximation algorithm for the resulting constrained matching task.
        However, unlike our work, \citet{benabbou2018diversity} do not consider individual fairness constraints and our algorithm provides better than $\frac{1}{2}$-approximation guarantee for common utility models (such as discounted cumulative gain \cite{DCG}) and block sizes.

        \paragraph{Fair decision-making with inaccuracies and uncertainty in inputs.}
            A growing number of works develop fair algorithms for decision-making that are robust to uncertainties and inaccuracies in their inputs \cite{LamyZ19,awasthi2020equalized,MozannarOS20,wang2020robust,wang2021label,Mehrotra2021MitigatingBI,celis2021fairclassification,celis2021adversarial,prob_fair_clustering,mehrotra2022fair,pln_2022,kempe2023matching}.
            Many of these works consider inaccuracies in protected attributes in decision-making tasks including ranking but extending beyond to subset selection, clustering, and classification \cite{LamyZ19,awasthi2020equalized,MozannarOS20,wang2020robust,Mehrotra2021MitigatingBI,celis2021fairclassification,celis2021adversarial,prob_fair_clustering,mehrotra2022fair}.
            A few recent works also consider uncertainty in other parts of the input \cite{wang2021label,pln_2022,kempe2023matching}.
            Among these, most relevant to our work, \citet{kempe2023matching} and \citet{pln_2022} study variants of the matching problem with uncertainty in utilities of items:
            \citet{kempe2023matching} adapt \citet{AshudeepUncertainty2021}'s notion of merit-based fairness to the matching task.
            \citet{pln_2022} consider both individual fairness and group fairness constraints, where the individual fairness constraints can capture the merit-based notion of \citet{kempe2023matching}.
            \citet{pln_2022} give an algorithm that samples a matching that satisfies the individual fairness constraints and satisfies the group fairness constraint (always).
            Interestingly, despite the differences between the ranking and the matching problem, we show a connection between our approach and a technical result in \cite{pln_2022} (see \cref{sec:overview}).

\section{Preliminaries and Model}\label{sec:model}
    \paragraph{Ranking problem.}
        In ranking problems, given $m$ items, the task is to select a subset $S$ of $n$ of these items and output the permutation of $S$ that is most valuable for the user.
        This permutation is called a \textit{ranking.}
        We consider a variant of the problem where the values or \textit{utilities} of the items are known.
        There is a vast literature on estimating item utilities (for specific queries) \cite{jeh2003scaling,dave2003mining,bar2008random,IRbook,liu2011learning} (see \cref{sec:relatedWork}).
        Abstracting this, we assume that for each item $i$ there is a utility $\rho_i\geq 0$ and for each position $j$ there is a discount factor $v_j>0$ such that placing the item $i$ at position $j$ generates value $\rho_i\cdot v_j$.
        The utility of a ranking is the sum of utilities generated by each item in its assigned position.
        The position discounts encode the fact that users pay higher attention to items earlier in the ranking.
        Various values of position discounts  have been considered in information retrieval literature. Perhaps the more prevalent one is discounted cumulative gain (DCG), which is specified by $v_j = (\log(1+j))^{-1}$ for each $j$ \cite{DCG}.
        Without loss of generality, we assume that item indices are ordered in non-increasing order of utilities, i.e., $\rho_1\geq \rho_2\geq \dots \geq \rho_m$.

        We denote a ranking by an assignment matrix $R\in \zo^{m\times n}$:
        $R_{ij}=1$ if item $i$ is placed in position $j$ and $R_{ij}=0$ otherwise.
        In this notation, the utility of a ranking $R$ is $$\rho^\top R v = \sum\nolimits_{i=1}^m\sum\nolimits_{j=1}^n \rho_i v_j R_{ij}.$$
        This variant of the vanilla ranking problem asks to solve:
        $$\max\nolimits_{R\in \cR} \rho^\top R v,$$ where $\cR$ is the set of all assignment matrices denoting a ranking:
        \begin{align*}
             \cR\coloneqq \inbrace{X\negsp{}\in\negsp{} \zo^{m\times n} \colon \forall_{i\in [m]},\ \sum\nolimits_j X_{ij}\leq 1,\ \forall_{j\in [n]},\ \sum\nolimits_i X_{ij} =1 }
            \yesnum\label{def:set_of_rankings}
        \end{align*}
        Here, for each $i$, the constraint $\sum\nolimits_j X_{ij}\leq 1$ encodes that item $i$ is placed in at most one position.
        For each $j$, the constraint $\sum\nolimits_i X_{ij}= 1$ encodes that there is exactly one item placed at position $j$.

    \paragraph{Fairness constraints.}
        Group fairness constraints are defined with respect to $p\geq 2$ socially-salient groups $G_1,G_2,\dots,G_p$ (e.g., the group of all women or the groups of all Asian or all black individuals).
        For simplicity, we state our results when groups $G_1,\dots,G_p$ are disjoint.
        In \cref{sec:overview}, we show that the same results hold when $G_1,\dots,G_p$ belong to a general family of overlapping sets, the laminar set family (see \cref{sec:overview}).
        There are many forms of group fairness constraints for ranking.
        We consider a class of fairness constraints that are placed over disjoint blocks of positions $B_1,B_2,\dots,B_q$.
        Blocks of positions can correspond to pages of search results or different windows in a scrollable feed.
        A basic example is where the $q=\frac{n}{k}$ blocks are disjoint sets of $1\leq k\leq n$ consecutive positions.
        Note, however, in general blocks can have different sizes.
        \begin{definition}[\textbf{Group fairness constraints}; \cite{GorantlaUnderranking21}]\label{def:group_constraints}
            Given matrices $L,U\in \Z^{q\times p}$ a ranking $R$ satisfies the $(L,U)$-group fairness constraints if for each $j\in [q]$ and $\ell\in [p]$
            \begin{align*}
                L_{j\ell}\leq \sum\nolimits_{i\in G_\ell}\sum\nolimits_{t\in B_j} R_{ij} \leq U_{j\ell}.
                \yesnum\label{eq:group_constraints}
            \end{align*}
        \end{definition}
        \noindent The above family of constraints can encapsulate a variety of group fairness notions.
        For instance, the equal representation constraint is captured by $L_{\ell j}=\floor{\frac{\abs{B_j}}{p}}$ and $U_{\ell j}=\ceil{\frac{\abs{B_j}}{p}}$ for each $\ell$ and $j$. (For readability, we omit the floor and ceiling operators henceforth.)
        To capture the Four-Fifths rule, it suffices to choose any constraints such that $L_{\ell j}\geq \frac{4}{5}\cdot U_{t j}$ for each $\ell,t\in [p]$ and $j\in [q]$.
        Existing works study related families of constraints \cite{fair_ranking_survey1,fair_ranking_survey2,overviewFairRanking}.
            We specifically consider \cref{def:group_constraints} as its block structure enables us to design efficient algorithms.
            In \cref{sec:additional_remarks}, we show that \cref{def:group_constraints} can ensure fairness with respect to the families of constraints from existing works \cite{fair_ranking_survey1,fair_ranking_survey2,overviewFairRanking}, hence, it also captures the corresponding notions of group fairness.

        That said, \cref{def:group_constraints} does not capture the adverse effects on specific items or individuals (henceforth, just items):
        highly-relevant items may get low visibility even though each protected group is sufficiently represented in every block (\cref{ex:exampleIntro}).
        To capture such underrepresentation, we consider the following family of individual fairness constraints.
        \begin{definition}[\textbf{Individual fairness constraints}]\label{def:individual_constraints}
            Given $A,C\in [0,1]^{m\times q}$, a distribution $\cD$ over the set $\cR$ of all rankings satisfies $(C,A)$-individual fairness constraints if for each $i\in [m]$ and $j\in [q]$
            \begin{align*}
                C_{ij} \leq \Pr\nolimits_{R\sim \cD}\insquare{R_{it} = 1 \text{ for some $t\in B_j$}} \leq A_{ij}.
                \yesnum\label{eq:equality_const_indv_fairness}
            \end{align*}
        \end{definition}
        \noindent By choosing $C_{ij}, A_{ij}$, one can lower and upper bound the probability that item $i$ appears in the block $B_j$ by a desired value.
        When item utilities are only probabilistically known, then a natural choice for the lower bounds is
        \[
            \hspace{-1.5mm}
            C_{ij} = Z \cdot \Pr_{\cU}[\text{$\exists t\in B_j$, $\rho_i$ is the $t$-th largest value in $\inbrace{\rho_1,\dots,\rho_m}$}].
            \yesnum\label{eq:example_indv_const}
        \]
        where $\cU$ is the joint distribution of utilities (see the discussion in \cite{AshudeepUncertainty2021}) {and $Z$ is a normalization constant that ensures that $\sum_{j=1}^m C_{ij}=1$ for all $1\leq i\leq m$.}
        Existing works have considered closely related families of individual fairness constraints and shown that those families capture many common notions of individual fairness \cite{AshudeepUncertainty2021,Patrick2022Expohedron}.
            Like group fairness constraints in \cref{def:group_constraints}, we consider the specific family in \cref{def:individual_constraints} as it enables efficient algorithms.
            In \cref{sec:additional_remarks}, we show that \cref{def:individual_constraints} can ensure fairness with respect to the constraints studied in  \cite{AshudeepUncertainty2021,Patrick2022Expohedron}, hence, \mbox{can also capture most common notions of individual fairness.}

        In the special case where $\cD$ is supported at just one ranking $R$, \cref{def:individual_constraints} specializes to the following: A ranking $R$ is $(C,A)$-individually fair if and only if the distribution supported on just $R$ is $(C,A)$-individually fair.
        Apart from very specific choices of the matrices $C$ and $A$, no ranking $R$ can be $(C,A)$-individually fair.
        For instance, this is true, whenever there is at least one $j$ such that $C_{ij}$ is positive for more than $\abs{B_j}$ choices of $i\in [n]$.
        Thus, in general, some of the output rankings must violate the individual fairness constraint.
        This has been recognized in the fair ranking literature \cite{AmortizedFairnessBiega2018,Oosterhuis21,fair_ranking_survey1,fair_ranking_survey2,overviewFairRanking}, and is one of the main reasons to consider randomized algorithms for ranking.
        In contrast, as mentioned in \cref{sec:intro}, it may be necessary (legally or otherwise) to ensure that each output ranking satisfies the group fairness constraints.
        Motivated by this, our goal is to solve the following problem.

        \begin{problem}[\textbf{Ranking problem with individual and group fairness constraints}]\label{prob:main}
            Given matrices $L,U, A, C$ and vectors $\rho,v$,
            find a distribution $\cD^\star$ over rankings maximizing the expected utility $\Pr_{R\sim \cD^\star}[\rho^\top R v]$ subject to satisfying  (i) $(C,A)$-individual fairness constraints and (ii) that each $R$ in the support of $\cD^\star$ satisfies $(L,U)$-group fairness constraints.
        \end{problem}
        \noindent A naive representation of $\cD^\star$ is to specify $\Pr_{S\sim \cD^\star}[S=R]$ for each ranking $R$.
        However, since the number of rankings is exponential in $n$ and $m$ (at least $n!$), even writing down this representation is intractable.
        Instead, like prior works \cite{fairExposureAshudeep,AshudeepUncertainty2021}, we encode $\cD^\star$ by the following $nm$ marginal probabilities.
        Given a distribution $\cD$, let $D\in [0,1]^{m\times n}$ encode the following marginals of $\cD$: $$D_{ij}\coloneqq \Pr_{R\sim \cD}[R_{ij}=1].$$
        In other words, in a ranking sampled from $\cD$, item $i$ appears in position $j$ with probability $D_{ij}$.

    \subsection{Challenges in Solving \cref{prob:main}}\label{sec:challenges}
        We first discuss the approach of a prior work \citet{fairExposureAshudeep,AshudeepUncertainty2021} and then discuss why it is challenging to use a similar approach to solve \cref{prob:main}.

        \paragraph{The approach of prior work.}
        \citet{AshudeepUncertainty2021} study a version of \cref{prob:main} where the blocks overlap and there are no group fairness constraints.
        Let $\wh{\cD}$ and $\hD$ be an optimal solution of their problem and its marginal respectively.
        Their algorithm has two parts: (1) solve \prog{prog:from_ashudeep} to compute $\hD$, (2) use the Birkhoff-von-Neumann (BvN)  algorithm \cite{brualdi_1982_birkhoff_von_neumann} to recover $\wh{\cD}$ from $\hD$.
        \begin{align*}
            &\argmax\nolimits_{D\in [0,1]^{m\times n}} \quad \rho^\top D v,
            \yesnum\label{prog:from_ashudeep}\\
            &\quad \st,\quad
            \forall i,\qquad\quad\ \ C_{ij}\leq \sum\nolimits_{t\in B_j} D_{it}\leq A_{ij}, \yesnum\label{eq:prog_from_ashudeep_eq_1}\\
            & \quad\qquad\ \ \forall j,\qquad\quad \ \
                \sum\nolimits_{i} D_{ij}=1
                        \quad\text{and}\quad
                    \forall i, \ \ \sum\nolimits_{j} D_{ij}\leq 1.\yesnum\label{eq:prog_from_ashudeep_eq_2}
        \end{align*}
        Consider any distribution $\cD$ and its marginal $D$, it can be shown that the objective $\rho^\top D v$ is equal to expected utility of $\cD$,  i.e., $\rho^\top D v = \Pr_{R\sim \cD}[\rho^\top R v]$, and that $D$ is feasible for \prog{prog:from_ashudeep} if and only if $\cD$ is $(C,A)$-individually fair.
        Using these, one can show that $\hD$ is an optimal solution of \prog{prog:from_ashudeep}.

        Since \prog{prog:from_ashudeep} is a linear program with $\poly(n,m)$ variables and constraints, it can be solved in polynomial time to get $\hD$.
        $\wh{\cD}$ can be recovered from $\hD$ using the BvN algorithm:
        Given $\hD$, the BvN algorithm outputs at most $nm$ rankings $R_1,\dots,R_{nm}$ and corresponding coefficients $\alpha_1,\dots,\alpha_{nm}$ such that
        $\wh{\cD}$ is the distribution that samples ranking $R_i$ with probability $\alpha_i$ for each $1\leq i\leq nm.$

        \newcommand{\rgf}{\ensuremath{\cR_{\rm GF}}}
        \newcommand{\mgf}{\ensuremath{\cM_{\rm GF}}}

        \paragraph{Challenges in solving \cref{prob:main}.}
            Let $\rgf$ be the set of all rankings that satisfy the $(L, U)$-group fairness constraints.
            Unlike \citet{AshudeepUncertainty2021}, we require the output distribution $\cD$ to be supported over $\rgf$.
            In other words, each $R$ sampled from $\cD$ should satisfy the $(L,U)$-group fairness constraints. %
            An obvious approach to solve \cref{prob:main} is to add ``group fairness constraints'' to \prog{prog:from_ashudeep} (to get \prog{prog:mod_of_ashudeep}) and generalize \cite{AshudeepUncertainty2021}'s algorithm as follows:
            \begin{enumerate} %
                \item Find a solution $\tD\in [0,1]^{m\times n}$ of \prog{prog:mod_of_ashudeep}
                \item Given $\tD$, output a distribution $\wt{\cD}$ such that $\tD$ is $\wt{\cD}$'s marginal and each $R$ sampled from $\wt{\cD}$ is in $\rgf$, i.e., $\Pr_{R\sim \wt{\cD}}[R\in \rgf{}]=1$
            \end{enumerate}
            \vspace{-3mm}

            \begin{align*}
                &\argmax\nolimits_{D\in [0,1]^{m\times n}} \ \ \rho^\top D v,
                \yesnum\label{prog:mod_of_ashudeep}\\
                &\quad \st,\qquad\qquad\quad\ \ D \text{ satisfies \cref{eq:prog_from_ashudeep_eq_1,eq:prog_from_ashudeep_eq_2}},
                \yesnum\label{eq:mod_of_ashudeep_eq_1}\\
                & \forall j\in [q],\ \forall \ell\in [p],\ \
                    L_{j\ell}\leq \sum\nolimits_{i\in G_\ell}\sum\nolimits_{t\in B_j} D_{ij} \leq U_{j\ell}.\yesnum\label{eq:mod_of_ashudeep_eq_2}
            \end{align*}
            \noindent Unfortunately, in general, the marginal of $\cD^\star$, say $D^\star$, is not  a  solution of \prog{prog:mod_of_ashudeep}, hence, in general, the output distribution $\wt{\cD}$ is different from the solution $\cD^\star$.
            In fact, it is possible that there is no distribution $\wt{\cD}$ supported over $\rgf$ such that $\tD$ is the marginal of $\wt{\cD}$--making it impossible to implement Step 2.
            One can explore different relaxations of Step 2.
            A relaxation is to output $\wt{\cD}$ that maximizes $\Pr_{R\sim \wt{\cD}}[R\in \rgf{}]=1$ subject to ensuring that $\tD$ is the marginal of $\wt{\cD}$.
            This, however, turns out to be \np-hard (\cref{thm:hardness_main}). %

            Let $\cS$ be the set of all matrices $D$ that are a marginal of some distribution $\cD$ that: (1) is $C$-individually fair and (2) is supported over rankings in $\rgf$.
            The key reason for these difficulties is that there are feasible solutions of  \prog{prog:mod_of_ashudeep} that are not in $\cS$.
            Using the definition of the marginal and $\cS$, one can show that $D^\star$ is  an  optimal solution to
            $$\argmax\nolimits_{D\in \cS}\rho^\top D v.$$
            However, it is unclear how to solve this program as it is not obvious how to even check if a matrix is in $\cS$.
            Thus, solving \cref{prob:main} requires new ideas.

\section{Theoretical Results}\label{sec:theoretical_results}
        In this section, we give our main algorithmic and hardness results.

        \subsection{Main Algorithmic Result}

        \paragraph{\bf Our approach.}
            The key idea is to consider a family of ``coarse rankings'' or \textit{matchings}: Each matching places $\abs{B_j}$ items in block $B_j$ (for each $1\leq j\leq m$), but it does not specify which items are placed at which positions inside $B_j$.
            We define natural analogs of the group fairness and the individual fairness constraints for matchings--leading to an analog of \cref{prob:main} for matchings.
             At a high level, our algorithm (\cref{alg:main}) first solves this analogue of \cref{prob:main} to get a distribution $\cD^{(\cM)}$ over matchings and then maps $\cD^{(\cM)}$ to a distribution $\cD=f(\cD^{(\cM)})$ over rankings; for an appropriate function $f$. %

            The fairness guarantee follows because of the facts that: (1) a matching $M$ is $(L,U)$-group-fair if and only if the corresponding ranking $R=f(M)$ is $(L, U)$-group-fair and (2) a distribution $\cD^{(\cM)}$ over matchings is $(C,A)$-individually-fair if and only if the corresponding distribution $\cD=f(\cD^{(\cM)})$ over rankings is $(C,A)$-individually-fair.
            This is where we use the fact that the blocks are disjoint.
            The utility guarantee follows because if $R=f(M)$ then the utility of $R$ is at least $\alpha$-times the utility of $M$ (see \cref{sec:algorithm_overview} for a definition of $f$). %

            Crucially, we are able to efficiently solve the analog of \cref{prob:main} for matchings because the linear inequalities capturing the group fairness constraints for matchings form a polytope such that all of its vertices are integral.
            The analogous statement is not true for rankings (see \cref{sec:fractional_vertex}); this is why all optimal solutions of \prog{prog:mod_of_ashudeep} can be different from $D^\star$.

        \paragraph{Main algorithmic result.}
        Next, we state our main algorithmic result, whose proof appears in \cref{sec:proofof:thm:algo_main}.
        This result holds for blocks of different sizes and any position discounts, but we also give a simpler expression for the utility when each block has size $k$ and the position discounts $v$ satisfy the following condition:
            \begin{align*}
                \text{
                    $\forall r\geq 0$,  $\frac{v_{t+r}}{v_t}$ is a non-decreasing in $1\leq t \leq n-r$.
                }
                \yesnum\label{eq:assumption_on_v}
            \end{align*}
            Standard position discounts such as those in DCG \cite{DCG} satisfy this assumption.

        \begin{restatable}[\textbf{Main algorithmic result}]{theorem}{thmAlgoMain}
            \label{thm:algo_main}
            There is a polynomial time randomized algorithm (\cref{alg:main}) that given matrices $L,U\in \Z^{q\times p}$, and $A,C\in \R^{m\times q}$, and vectors $\rho\in \R^m$ and $v\in \R^n$, outputs a ranking $R$ sampled from a distribution $\cD$ such that:
            \begin{itemize}[itemsep=-1pt]
                \item $\cD$ satisfies $(C,A)$-individual fairness constraint, and
                \item $R$ satisfies $(L,U)$-group fairness constraint.
            \end{itemize}
            The expected utility of $R$ is at least $\alpha$ times the expected utility of a ranking sampled from $\cD^\star$.
            If all blocks have size $k$ and \cref{eq:assumption_on_v} holds, then %
            \begin{align*}
                \alpha\geq \frac{v_1+v_2+\dots+v_k}{k \cdot v_1}.
                \yesnum\label{eq:lowerbound_on_alpha}
            \end{align*}
            Furthermore, regardless of block-sizes and \cref{eq:assumption_on_v}, it holds that $\alpha\geq \min_{1\leq j\leq q} \frac{\sum_{s\in B_j}  v_s}{\abs{B_j} \cdot v_{s(j)}},$ where $s(j)$ is the first position in $B_j$.
        \end{restatable}

            \noindent Thus, \cref{alg:main} is an $\alpha$-approximation algorithm for \cref{prob:main}.
            Here, $\alpha$ is a value that approaches 1 as the range of position discounts shrinks.
            In the worst case, when $v_2=v_3=\dots=v_k=0$, the RHS in \cref{eq:lowerbound_on_alpha} is $\frac{1}{k}$.
            For common position discounts the RHS of \cref{eq:lowerbound_on_alpha} is closer to 1: for instance, for DCG \cite{DCG} with $k=2,3,4$ it is at least $0.81,0.71,0.64$ respectively.
            These lower bounds are tight in some examples where a few items have a very large utility.
            If, however, items' utilities lie in a bounded interval, then this lower bound can be improved.
            To see concrete bounds, suppose $\frac{\max_i \rho_i}{\min_i \rho_i} \leq 1+\Delta$.
            One can show that
            \begin{align*}
                 \alpha
                &\geq \inparen{ 1 + \Delta}\inparen{  1+\frac{k  v_1\Delta}{  v_1+v_2+\dots+v_k}}^{-1}.
                \yesnum\label{eq:strongerInequality}
            \end{align*}
            Thus, $\alpha$ approaches 1 as $\Delta$ approaches 0, i.e., as the range of item utilities shrinks.
            One can show that, for any $\Delta\geq 0$, the above bound is at least as large than the RHS in \cref{eq:lowerbound_on_alpha} (\cref{sec:verify2}).
            The proof of \cref{eq:strongerInequality} appears in \cref{sec:verify}.
            We present further utility-dependent approximation guarantees in \cref{sec:utilityDependentGuarantee}.
            As for the running time, \cref{alg:main} solves a linear program in $O(nm)$ variables with $O(np+m)$ constraints and performs $O(n^2m(p+m))$ additional arithmetic operations (\cref{sec:complete_proof_algo_main}).
            \cref{thm:algo_main} also holds, without change, for any set of protected groups that form a laminar family, i.e., for any set of groups such that either $G_\ell\subseteq G_k$ or $G_k\subseteq G_\ell$ for each $1\leq \ell,k\leq p$  (\cref{sec:overview}).
            {Laminar groups can be relevant in contexts where (some notion of) group fairness for intersectional groups is desired: as a concrete example, if one defines (1) $G_1$ to be the group of all non-women, (2) $G_2$ to be the group of all women, and (3) $G_3$ to be the intersectional group of all Black women (within the group of all women), then \cref{alg:main} ensures that (the specified notion of) group fairness is also satisfied for the intersectional group of all Black women.}
            {Finally, one can verify that the more general bound in \cref{thm:algo_main} (i.e., $\alpha\geq \min_{1\leq j\leq q} \frac{\sum_{s\in B_j}  v_s}{\abs{B_j} \cdot v_{s(j)}}$) reduces to the one in \cref{eq:lowerbound_on_alpha} when all blocs have size $k$ and \cref{eq:assumption_on_v} holds} (see \cref{eq:verifying_upperbound}).

    \subsection{Better Approximation Guarantees With Distributional Assumptions}\label{sec:utilityDependentGuarantee}%
        Next, we present utility-dependent approximation guarantees of \cref{alg:main} under generative models where each item $i$'s utility $\rho_i$ has uncertainty and is only ``probabilistically known.''
        For the sake of concreteness, we begin with the generative model where, the ``true'' utility, $\rho_i$, of each item $1\leq i\leq m$ is drawn from the normal distribution $\mathcal{N}(\mu_i,\sigma_i^2)$ independent of all other items where $\mu_i\in \R$ and $\sigma_i\geq 0$ are parameters that are known to the algorithm.

        Here, we choose the normal distribution for the sake of simplicity: more generally, $\rho_i$ can be drawn from any (possibly nonsymmetric) sub-gaussian distribution with mean $\mu_i$ and variance $\sigma_i^2$.
        In particular, the specific sub-gaussian distribution can be different for different items.
        If we choose the normal distribution for each item, then the resulting utility model is identical to the implicit variance model of \citet{emelianov2020implicit}, who claim that such uncertainties in the utilities can arise in the real world.

        Uncertainties in utilities arise from various sources (from measurement errors, uncertainties in prediction, to errors in data) in practice and are one of the motivations to consider individual fairness constraints \cite{AshudeepUncertainty2021}.
        When utilities are only probabilistically known, a natural family of individual fairness constraints (which is also proposed by \citet{AshudeepUncertainty2021}) is in \cref{eq:example_indv_const}.
        Under these individual fairness constraints, when the parameters $\mu_1,\mu_2,\dots,\mu_m$ are i.i.d. from the uniform distribution on $[0, S]$ (for some constant $S>0$), we have the following approximation guarantee whose proof appears in \cref{sec:proofof:thm:approxStochasUtil}.

        \begin{restatable}[]{theorem}{thmApproxStochasUtil}
            \label{thm:approxStochasUtil}
            Suppose $\mu_1,\mu_2,\dots,\mu_m$ are i.i.d. from the uniform distribution on $[0,S]$, $\rho_i$ follow the above generative model, $C\in \R^{m\times q}$ is as specified in \cref{eq:example_indv_const},  $A=[1]_{m\times q}$, and $nm^{-1}$ is bounded away from 1.
            \cref{alg:main}, given means of the utilities $\mu\in \R^m$ and other parameters $(L,U,A,C,v)$,
            outputs a ranking $R$ sampled from a distribution $\cD$ that satisfies the fairness constraints in \cref{thm:algo_main} and has an expected utility at least $\alpha$ times the expected utility of a ranking sampled from $\cD^\star$, where
            \[
                \alpha \geq 1 - \wt{O}\inparen{\frac{\sigma_{\rm \max}}{S}\cdot \sqrt{\log{m}}} - O\inparen{m^{-\frac{1}{4}}}\quad \text{and}\quad \sigma_{\max}\coloneqq\max\inbrace{\sigma_1,\sigma_2,\dots,\sigma_m}.
            \]
        \end{restatable}

        \noindent
        Hence, the above theorem shows that if the variance of the uncertainty in items' utilities ($\sigma_{\max}^2$) is ``small'' compared to the range of their utilities ($S$), then with high probability \cref{alg:main} has a near-optimal approximation guarantee.
        As for the assumption about the distribution of the means $\mu_1, \mu_2,\dots, \mu_m$,
        note that if the utilities denote the percentiles of items, then one expects $\mu_1,\mu_2,\dots,\mu_m$ to be uniformly distributed in $[0,100]$ \cite{KleinbergR18}.
        In this case, $S=100$ and the approximation guarantee is of the order of $1-S^{-1}\cdot \sqrt{\log{m}}\geq 0.95,$ for any $m\leq 10^{10}$.

        The proof of the above result only uses the concentration property of the Gaussian distribution.
        This is why the result extends to (possibly non-symmetric) sub-Gaussian distribution (which can be different for different items).
        {Note that since $\rho_i$ is drawn from the normal distribution $\mathcal{N}(\mu_i,\sigma_i^2)$, it can take negative values.
        To avoid this, one can consider an appropriately truncated version of the normal distribution.
        Since any truncation of the normal distribution is sub-gaussian, a bound of the same form (with an appropriate constant) continues to hold for $\alpha$.}

        Moreover, if $\mu_1,\mu_2,\dots,\mu_m$ are arbitrary deterministic values, then the following approximation guarantee for \cref{alg:main} is implicit in the proof of \cref{thm:approxStochasUtil} (see \cref{eq:proofForArbitraryMu})
        \[
            \alpha
            \geq
                1
                - \wt{O}\inparen{
                        \frac{\sigma_{\max}}{\mu_n}
                        \cdot
                        \sqrt{\log{m}}
                    }
                - O\inparen{m^{-\frac{1}{4}}},
            \yesnum
        \]
        where  $\mu_{(n)}$ is the $n$-th largest value in $\mu_1, \mu_2,\dots,\mu_m$.

    \subsection{Overview of the Algorithm}\label{sec:algorithm_overview}
        \cref{alg:main} encodes a matching by an $m\times q$ matrix $M\in \inbrace{0,1}^{m\times q}$ where $M_{ij}=1$ if item $i$ is in the block $B_j$ and $M_{ij}=0$ otherwise.
        Let $\cM$ be the set of matrices encoding a matching.
        \cref{alg:main} uses two functions $f$ and $g$.
        For any ranking $R\in [0,1]^{m\times n}$, $g(R)\in [0,1]^{m\times q}$ is the matching such that $g(R)_{ij} \coloneqq \sum\nolimits_{t\in B_j} R_{it},$ for each $i\in [m]$ and $j\in [q]$.
        Intuitively, for a ranking $R$, $g(R)$ is the unique matching that matches item $i$ to block $B_j$ if and only if item $i$ appears in $B_j$ in $R$.
        For any matching $M$, $f(M)$ is the unique ranking that satisfies: (1) $g(f(M))=M$ and (2) for each $j$, items in block $B_j$ appear in non-increasing order of their utility in $f(M)$.
        Concretely, our algorithm is as follows.

        \begin{algorithm}[h!] %
            \caption{Pseudo-code for the algorithm in \cref{thm:algo_main}} \label{alg:main}
            \begin{algorithmic}[1]
                \Require Matrices $L,U\in  {\R}^{q\times p}$ and $A,C\in  {\R}^{m\times q}$, and vectors $\rho\in  \R^m$ and $v\in \R^n$, and sets $B_1,B_2,\dots,B_q\subseteq[n]$
                \Ensure A ranking $R\in \cR$
                \State \textit{(Solve)} Compute an optimal solution $\hD$ of Program~\eqref{prog:mod_of_ashudeep}
                \item[]\hfill $\triangleright$ {\gray{$\poly(n,m)$ time as Program~\eqref{prog:mod_of_ashudeep} has $\poly(n,m)$ variables and constraints}}
                \State \textit{(Project)} Compute the projection $\hM \coloneqq g(\hD)$ \Comment{\gray{$O(mn)$ time}}
                \State \textit{(Decompose)} Compute $M_1,M_2,\dots,M_T$  and $\alpha_1,\alpha_2,\dots,\alpha_T$ s.t. $$\hM=\sum\nolimits_{t\in [T]} \alpha_t M_t,$$
                \item[] where $T=O(n^2m^2)$
                \item[]\hfill $\triangleright$ {\gray{$(M_t,\alpha_t)_{t=1}^T$ can be computed in {$O(n^2m^2)$ time,} see Lemma~\ref{lem:decomposition}}}
                \State \textit{(Refine)} \textbf{For each} $t\in [T]$ \textbf{do:} Set $R_t\coloneqq f(M_t)$
                \Comment{\gray{$\wt{O}(Tmn)$ time}}
                \State \Return $R_t$ with probability $\propto \alpha_t$ for each $t\in [T]$
            \end{algorithmic}
        \end{algorithm}

\section{Empirical Results}\label{sec:empirical}
In this section, we show the performance of our algorithm on synthetic and real-world datasets. We explore two research questions: $(i)$
\textit{How likely is it for a fair ranking baseline to sample a ranking that violates group fairness constraints?}
$(ii)$ \textit{Does \cref{alg:main} achieve a similar utility as baselines?}
We start by describing our experimental setup before diving into the results of our experiments.

\subsection{Setup, Baselines, and Metrics}

    Recall the ranking problem we consider -- given $m$ items, the output should be an ordered list or ranking of $n$ items that maximize the utility.
    The utility generated by item $i$ in position $j$ is $\rho_i\cdot v_j$, where $\rho_i$ is an item-specific utility and $v_j$ is the position discount.
    The choices of $n$, $m$, $k$, and $\rho$ are data and application dependent; we specify our choices of these parameters in \cref{tab:hyperparams} and discuss the choice of all of $n$, $m$, $k$ and $\rho$ further with each dataset.
    Across all datasets, we set $v_j\coloneqq \frac{1}{\log(j+1)}$ for each $1\leq j\leq n$, corresponding to the popular discounted cumulative gain (DCG) measure \cite{DCG}.

\paragraph{{Fairness constraints.}}
    The choice of the right fairness constraints is context-dependent.
    For illustration, we choose generalizations of the equal representation constraint (which is, perhaps, the most common group fairness constraint considered in the literature \cite{fair_ranking_survey1,fair_ranking_survey2}).
    These generalizations are parameterized by $1\leq \phi\leq p$ are specified by blocks $B_1,B_2,\dots,B_q$ of equal size $k\coloneqq \frac{n}{2}$.
    Given a value of $\phi$, the constraint is specified by the upper bounds $U_{j\ell} \coloneqq \ceil{\frac{\phi k}{p}}$ for block $B_j$ and the protected group $G_\ell$ (for each $j$ and $\ell$); the lower bounds of the group fairness constraints are set to be vacuous, i.e., $L_{j\ell}=0$ for all $j$ and $\ell$.
    To gain some intuition about the relevant values of $\phi$, note that when $\phi=p$ the upper bounds are vacuous and when $\phi=1$ the upper bounds require the ranking to contain exactly $\frac{k}{p}$ items in each block.
    As for the individual fairness constraints, we consider a family of individual fairness constraints proposed by \citet{AshudeepUncertainty2021} which, in turn, are motivated by the uncertainties in the item utilities, as are bound to arise in the real world.
    Following the construction in \cite{AshudeepUncertainty2021},
    we assume that the true utility of item $i$ is $\rho_i=\wt{\rho}_i + X_i$, where $\wt{\rho}_i$ is an estimated utility and $X_i$ is a Gaussian random variable with mean $0$ and a data-dependent standard deviation $\sigma$ computed to be the smallest value such that there are at least $\frac{k}{2}$ items with estimated utility within $\wt{\rho}_i\pm \sigma$, on an average.
    Given these, the matrix $C$ specifying the individual fairness constraints is specified as $C_{ij} = \gamma\cdot \Pr[\text{$\exists t\in B_j$, $\rho_i$ is the $t$-th largest value in $\inbrace{\rho_1,\dots,\rho_m}$}]$
    where $0 \leq \gamma\leq 1$ is a relaxation factor.
    We set the upper bounds of the individual fairness constraints to be vacuous,i.e., $A_{ij}=1$ for all $i$ and $j$.
    Note that when $\gamma=1$, $C_{ij}$ is equal to the probability that item $i$ appears in the $j$-th block when items are ordered in decreasing order of true utility.
    Note that the ``strength'' of our group fairness constraint is specified by the parameter $1\leq \phi\leq p$ (where the closer $\phi$ is to 1 the closer the constraint is to equal representation) and the strength of our individual fairness constraint is specified by $0\leq \gamma\leq 1$ (where the closer $\gamma=1$ the ``stronger'' the individual fairness requirement is).

\paragraph{Baselines.} We compare our algorithm to both baselines that output a deterministic ranking and those that sample a ranking from a distribution.%
The following baselines output a deterministic ranking:
\begin{enumerate}
    \item \textbf{Unconstrained}, which is a baseline that outputs  a  ranking that maximizes the utility (without consideration for fairness constraints); and
    \item \textbf{CSV18 (Greedy)} \cite{celis2018ranking}, which is an algorithm that greedily ranks the item and is guaranteed to satisfy the specified group fairness constraints, but does not consider individual fairness constraints.
\end{enumerate}
We also consider baselines that are closer to \cref{alg:main}, in the sense that, they sample a ranking from an underlying distribution such that the output is guaranteed to satisfy the specified individual fairness constraints:

\begin{enumerate}
    \item \textbf{SJK21 (IF)}, which is the algorithm of \citet{AshudeepUncertainty2021} specialized to the individual fairness constraints considered in our simulations;\footnote{This algorithm first solves \prog{prog:from_ashudeep} to compute a marginal $D$, decomposes it as $D=\sum_t \alpha_t R_t$ (using Birkhoff von Neumann decomposition), and outputs $R_t$ with probability $\propto \alpha_t$.}
    \item \textbf{SJK21 (GF and IF)}, which is the algorithm of \citet{AshudeepUncertainty2021} specialized to satisfy both the individual fairness constraints and (in aggregate) the group fairness constraints considered in our simulations.\footnote{This algorithm first solve \cref{prog:mod_of_ashudeep} to compute a marginal $D$, decomposes it as $D=\sum_t \alpha_t R_t$ (using Birkhoff von Neumann decomposition), and outputs $R_t$ with probability $\propto \alpha_t$.}
\end{enumerate}

\paragraph{Metrics.}
    We evaluate the rankings output by each algorithm using three metrics:
    the probability with which the output ranking violates the group fairness constraints $\gfviol$,
    a measure of the amount of violation of the individual fairness constraints $\ifviol$ (see below), and
    the output ranking $R$'s normalized output utility $$\util =\frac{1}{\util_{\max}}\cdot \Ex[\rho^\top R v]$$ (where the expectation is over any randomness in $R$ and $\util_{\max}$ is a normalization constraint that ensures that $\util$ has range from 0 to 1).
    It remains to define $\ifviol$: let $P_{ij}$ be the probability with which item $i$ appears in block $B_j$ in the output ranking $R$, the individual fairness violation of the corresponding algorithm is defined as $$\ifviol \coloneqq \frac{1}{m}\sum_{i \in [m]}\frac{1}{q}\sum_{j \in [q]}\max\inbrace{1 - \frac{P_{ij}}{C_{ij}}, 0}.$$
    Note here that both $\gfviol$ and $\ifviol$ have a range from $0$ to 1, where a smaller value implies a smaller violation.

\begin{figure*}[t!]
    \centering
    \begin{subfigure}
        \centering
        \includegraphics[scale=0.5]{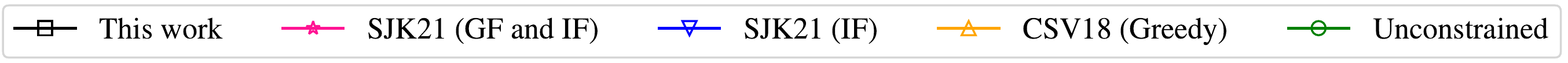}
    \end{subfigure}

    \begin{subfigure}
        \centering
        \includegraphics[scale=0.35]{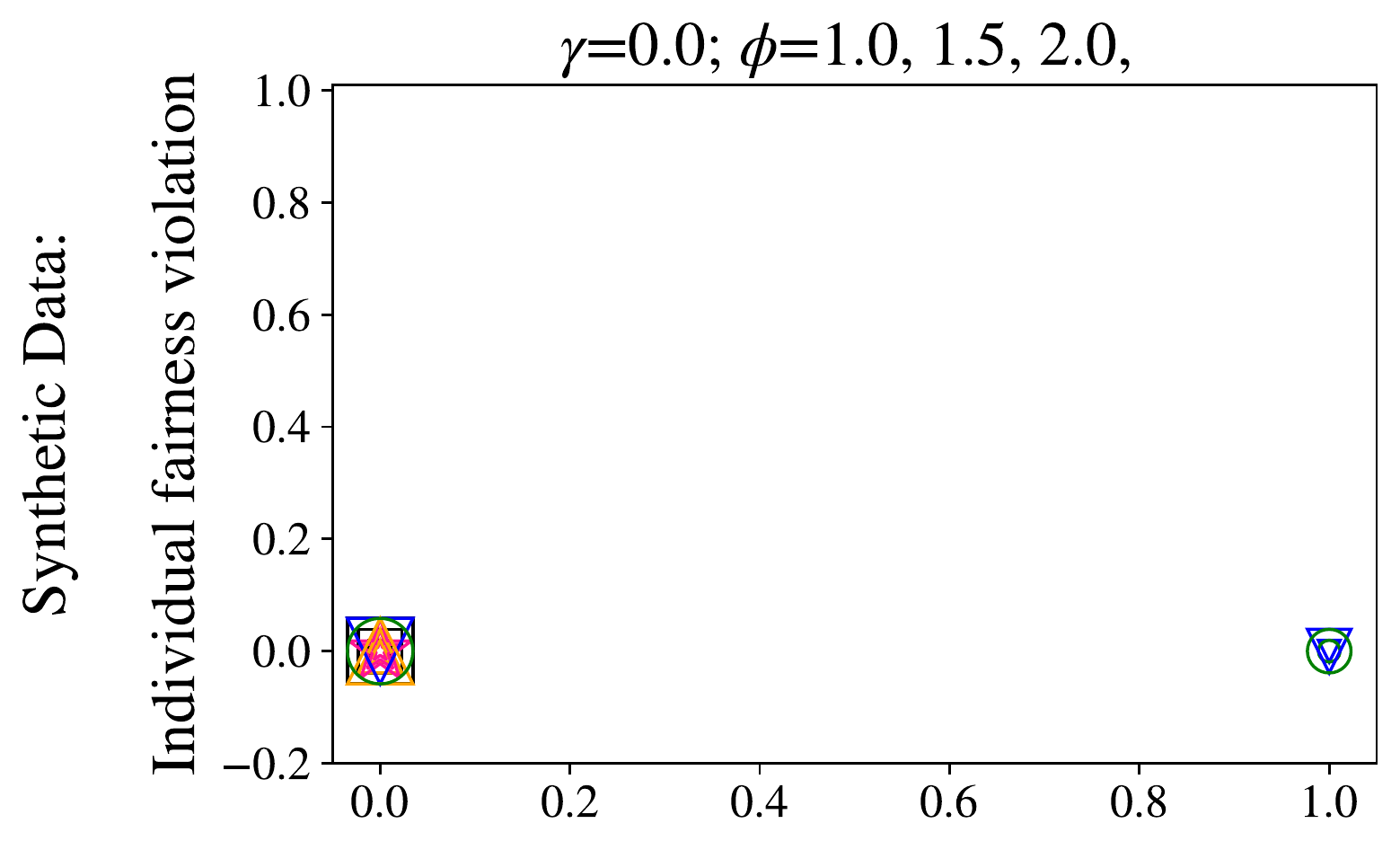}
    \end{subfigure}
    \begin{subfigure}
        \centering
        \includegraphics[scale=0.35]{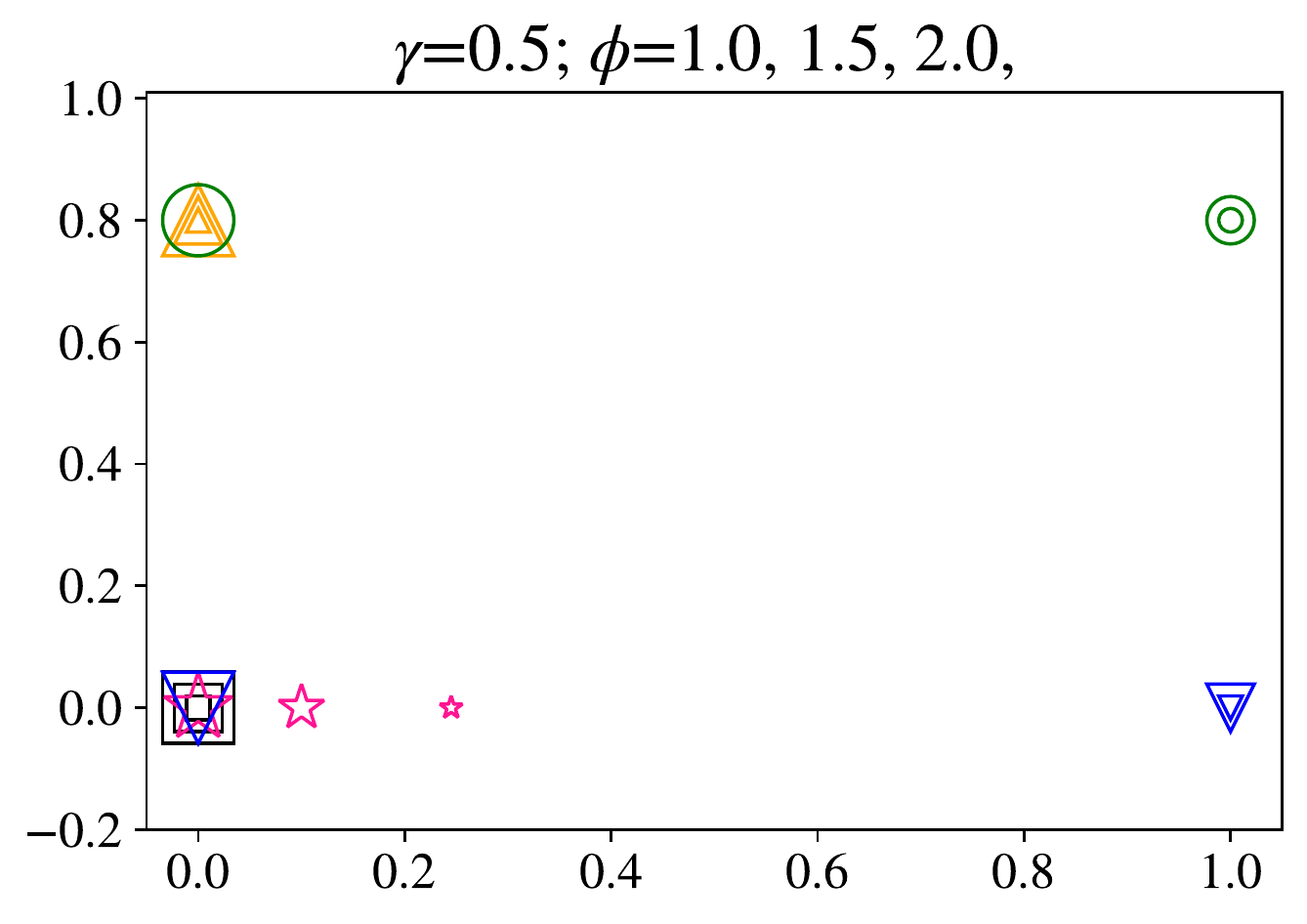}
    \end{subfigure}
    \begin{subfigure}
        \centering
        \includegraphics[scale=0.35]{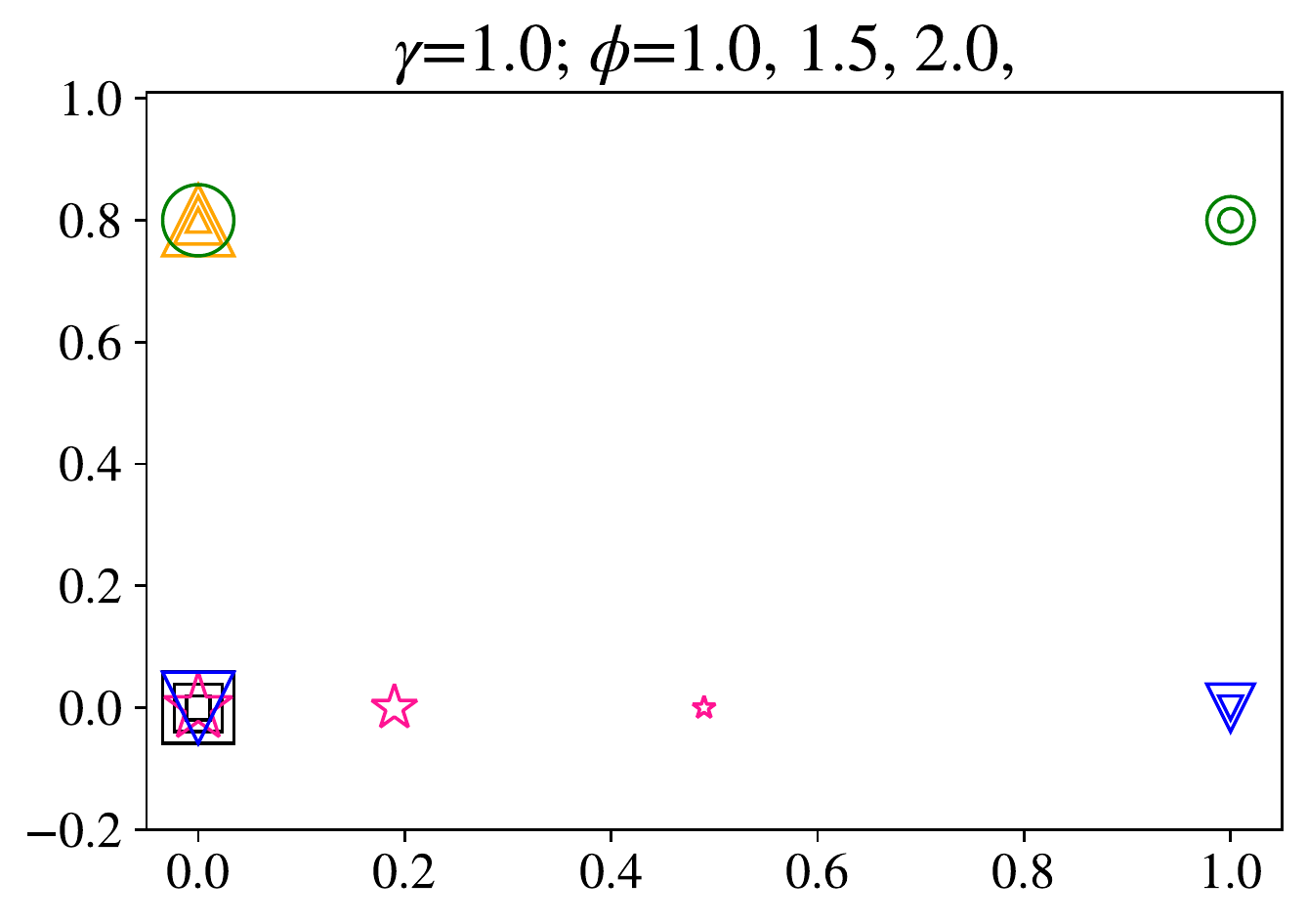}
    \end{subfigure}

    \begin{subfigure}
        \centering
        \includegraphics[scale=0.35]{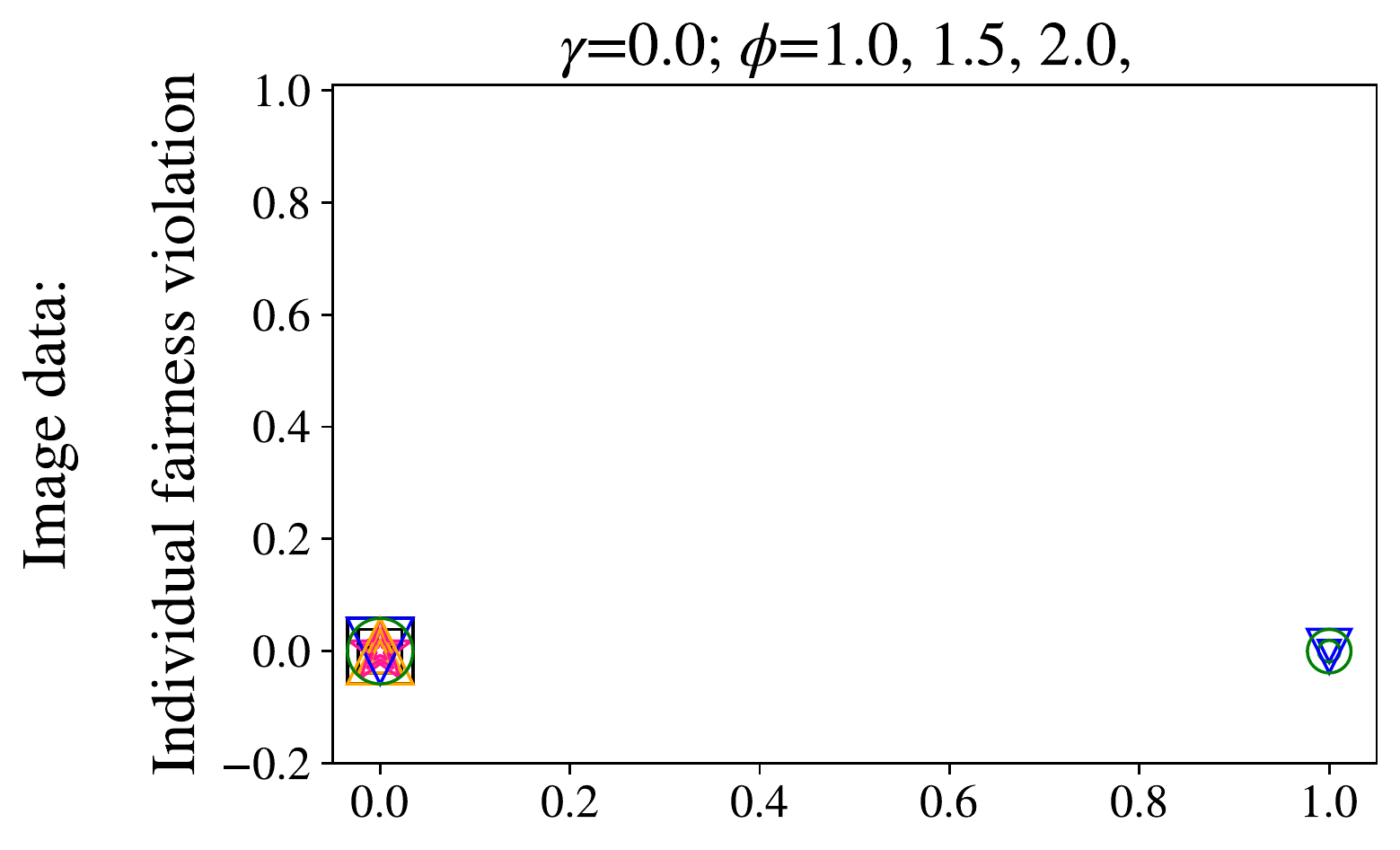}
    \end{subfigure}
    \begin{subfigure}
        \centering
        \includegraphics[scale=0.35]{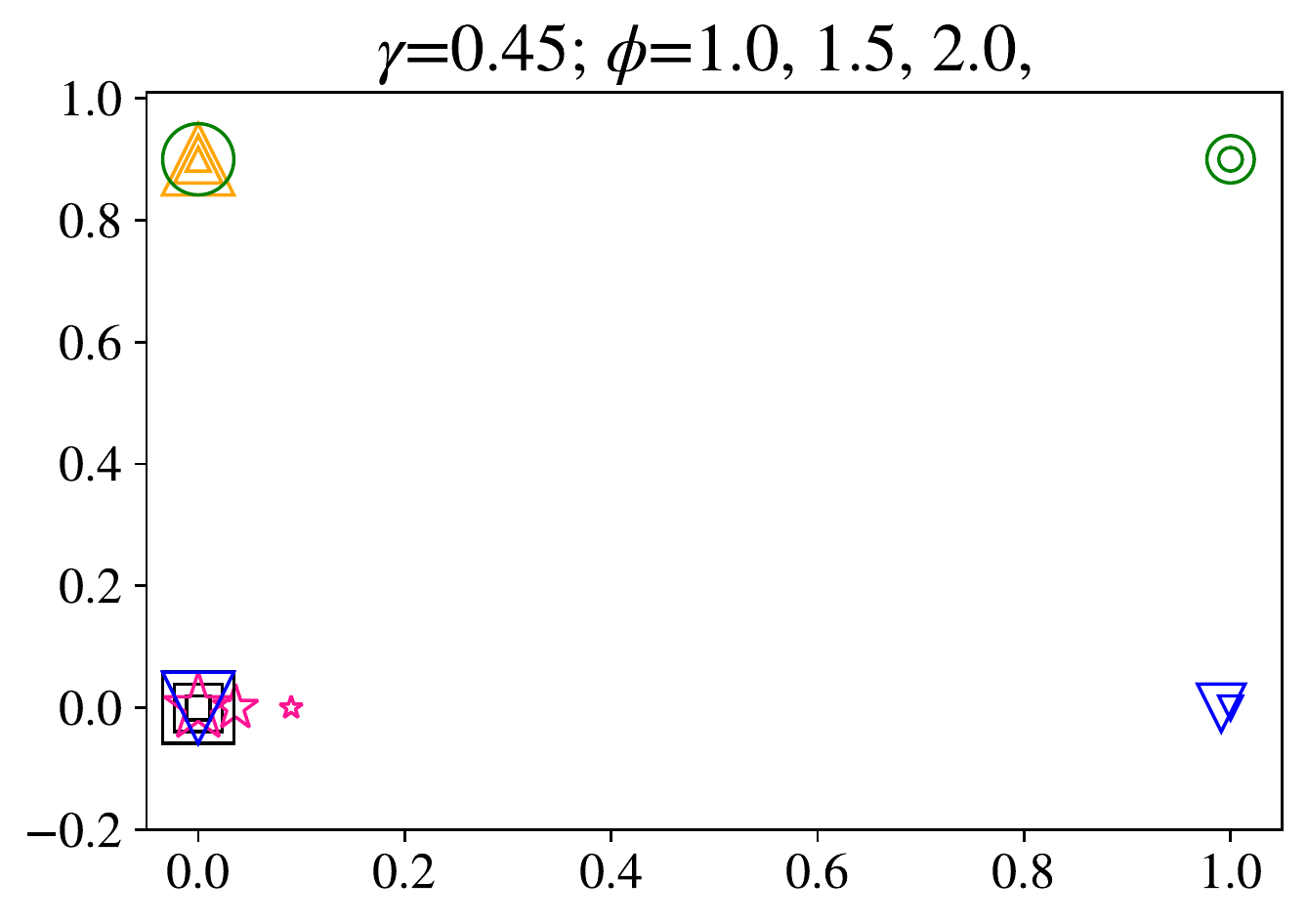}
    \end{subfigure}
    \begin{subfigure}
        \centering
        \includegraphics[scale=0.35]{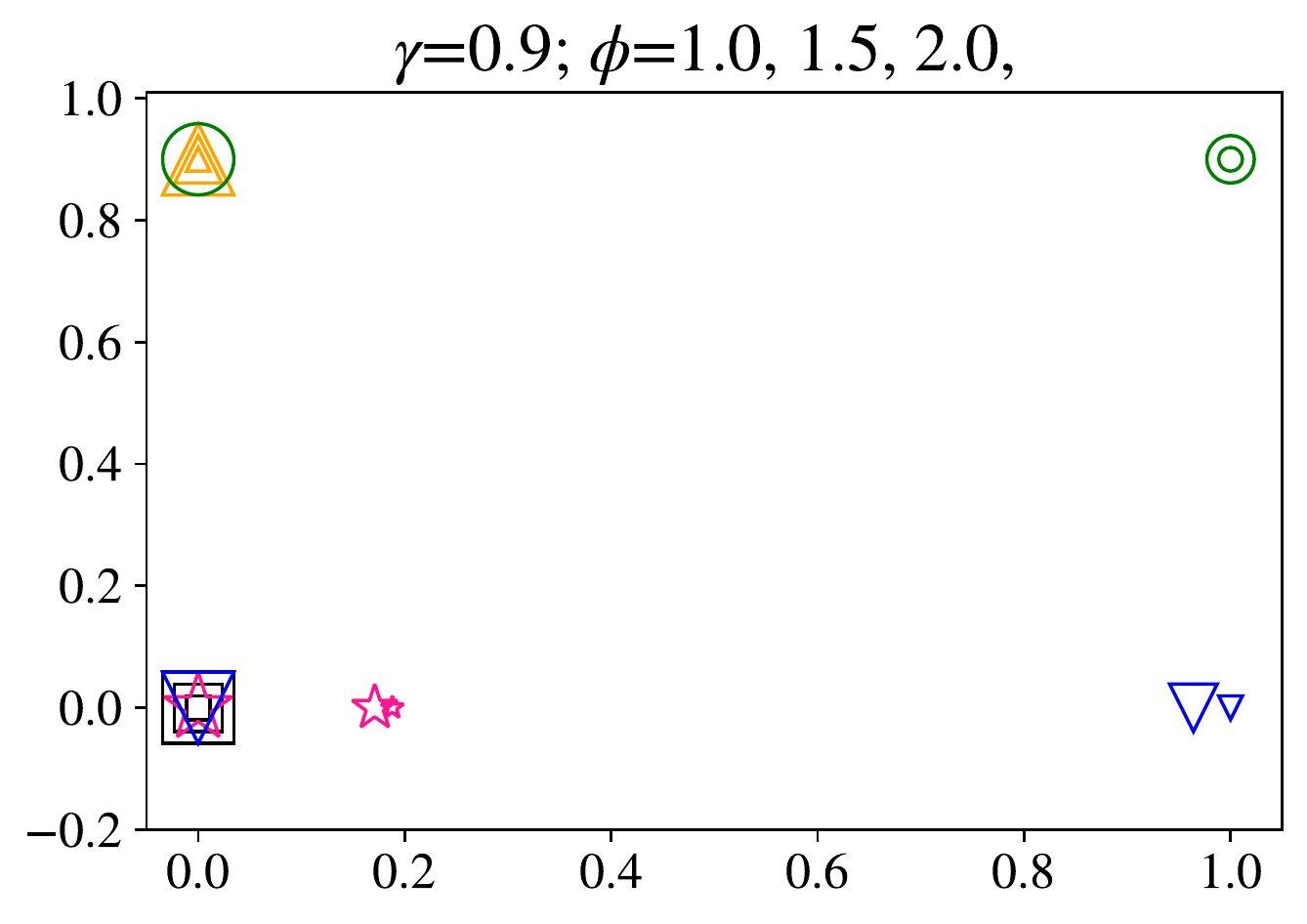}
    \end{subfigure}

    \begin{subfigure}
        \centering
        \includegraphics[scale=0.35]{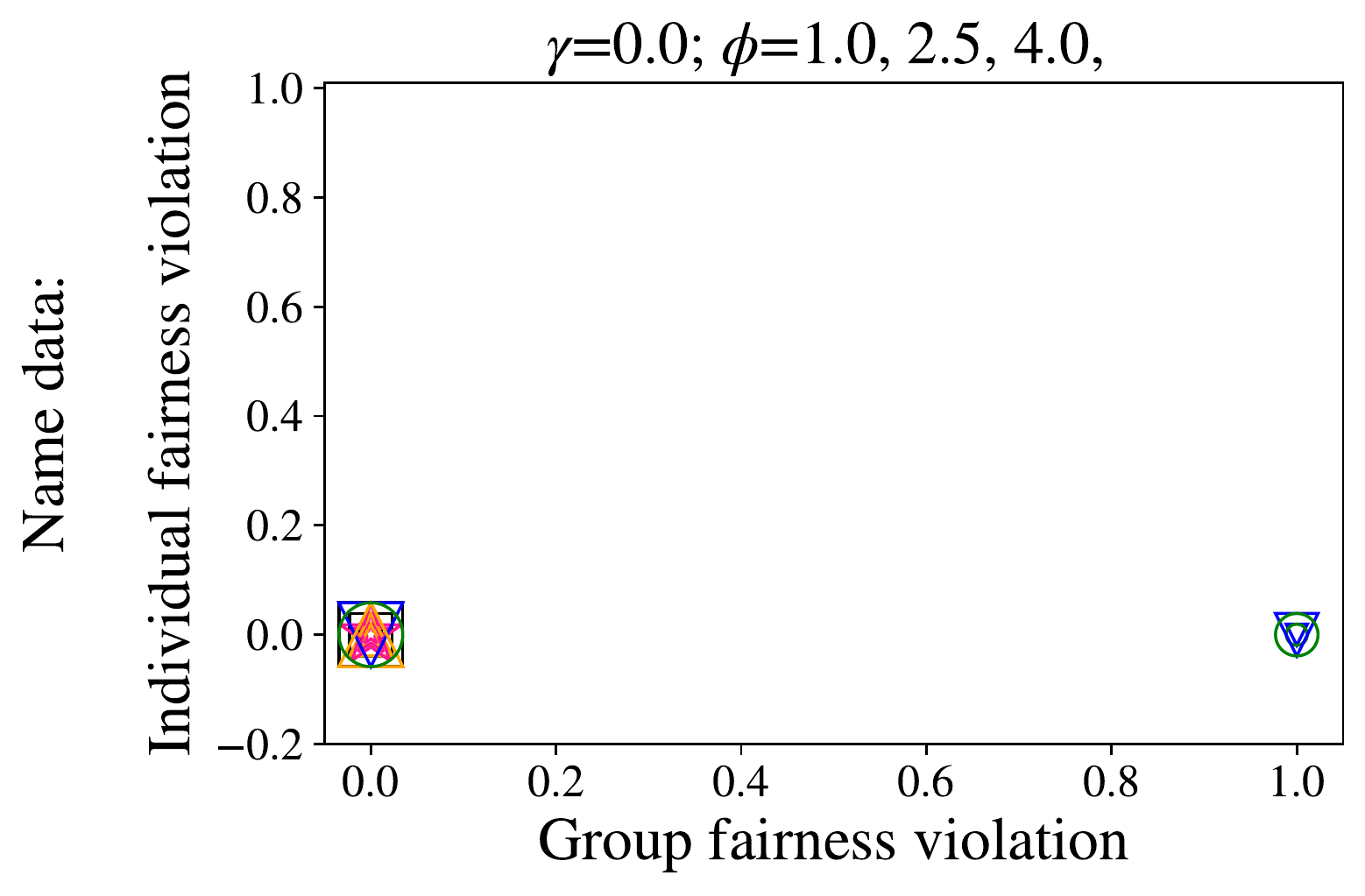}
    \end{subfigure}
    \begin{subfigure}
        \centering
        \includegraphics[scale=0.35]{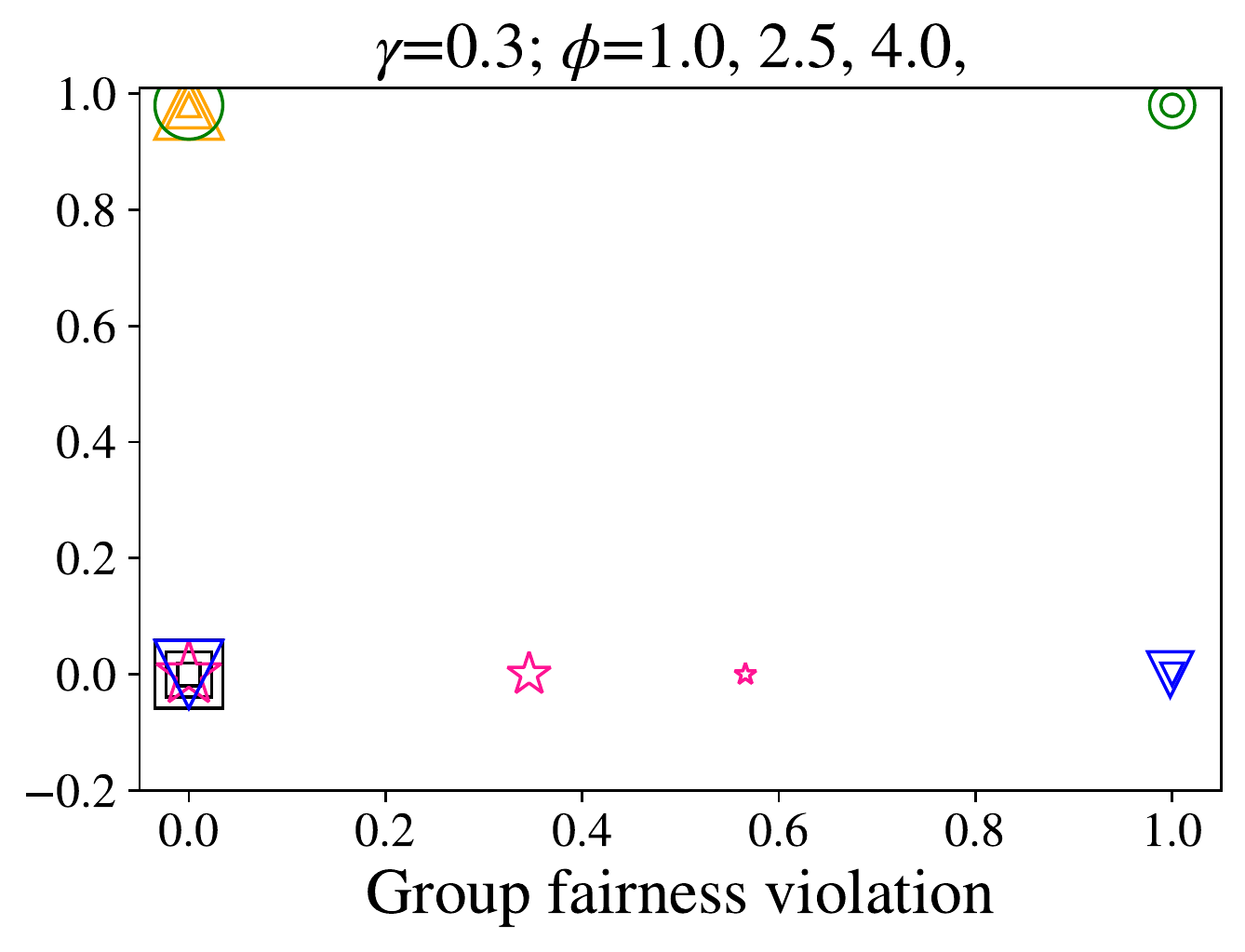}
    \end{subfigure}
    \begin{subfigure}
        \centering
        \includegraphics[scale=0.35]{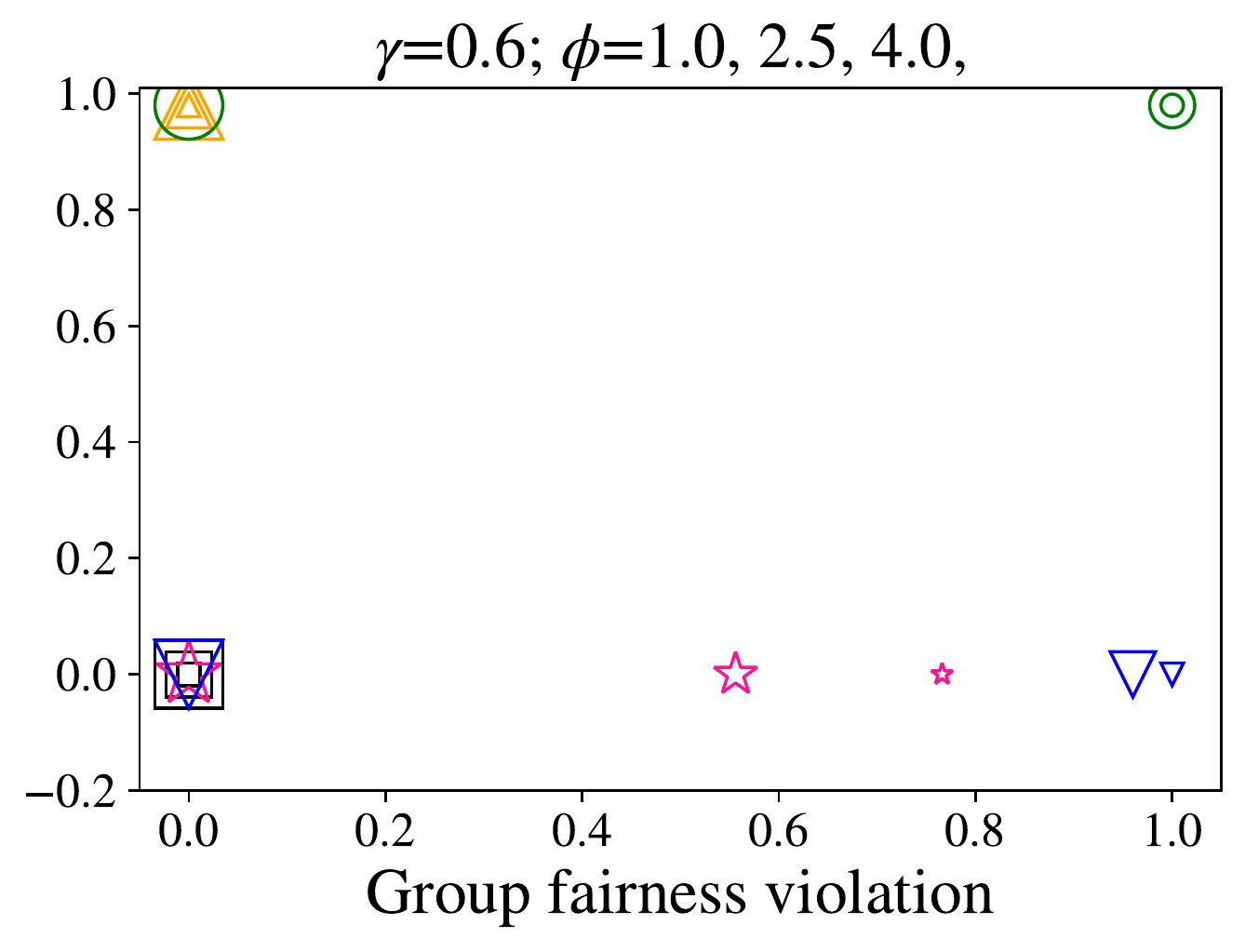}
    \end{subfigure}
    \caption{\textbf{Individual fairness violation vs. Group fairness violation:} In the plots, the parameter $\gamma$ controls individual fairness constraints and the parameter $\phi$ defines block-wise representation constraints. The size of the marker for each algorithm in each plot is proportional to the value of $\phi$. Lower the value of $\phi$, the stronger the group fairness constraints. In contrast, the lower the value of $\gamma$, the weaker the individual fairness constraints.}
    \label{fig:fairness}
\end{figure*}

    \begin{figure*}[t!]
        \centering
        \begin{subfigure}
            \centering
            \includegraphics[scale=0.5]{figures/legend.pdf}
        \end{subfigure}
        \begin{subfigure}
            \centering
            \includegraphics[scale=0.35]{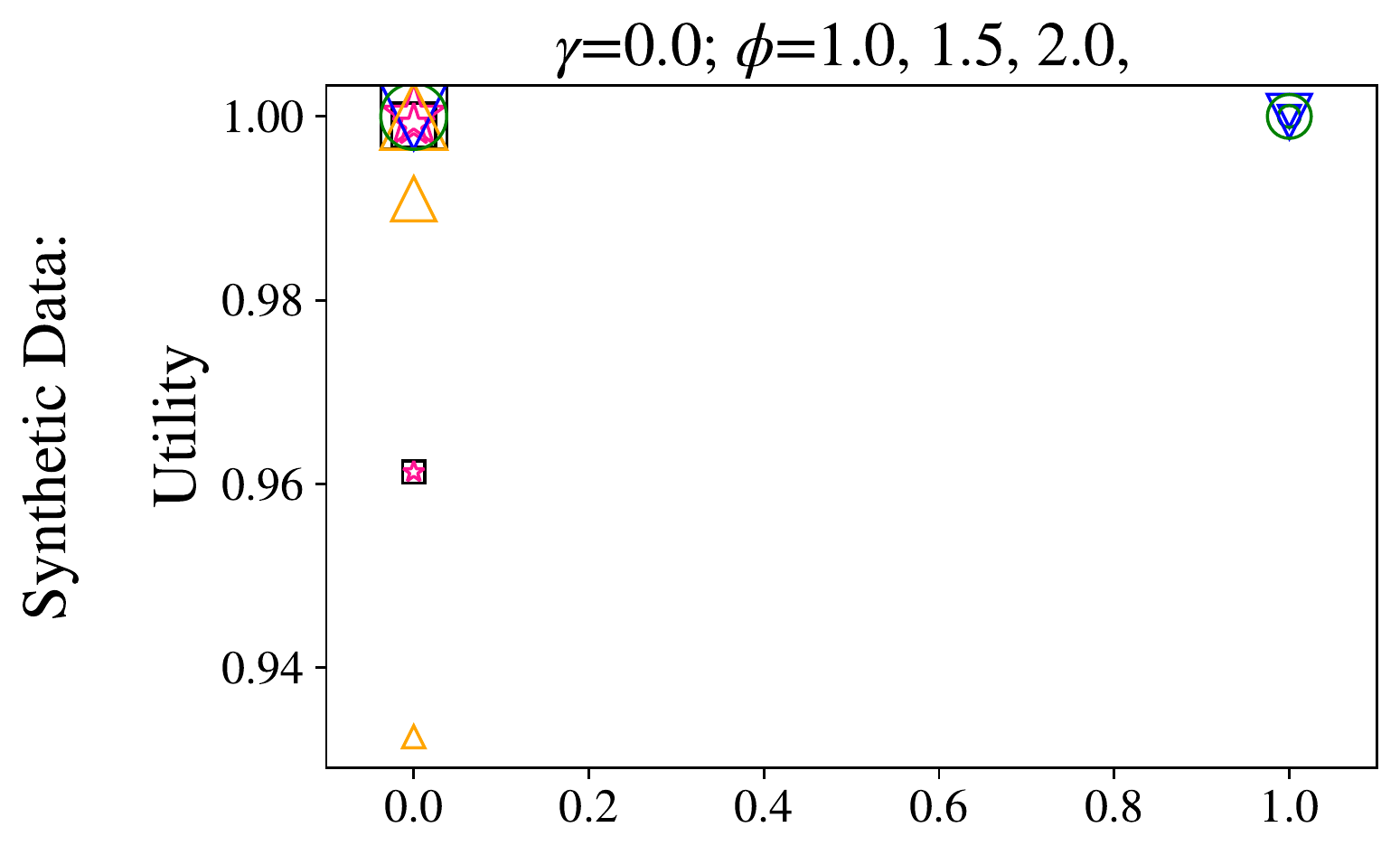}
        \end{subfigure}
        \begin{subfigure}
            \centering
            \includegraphics[scale=0.35]{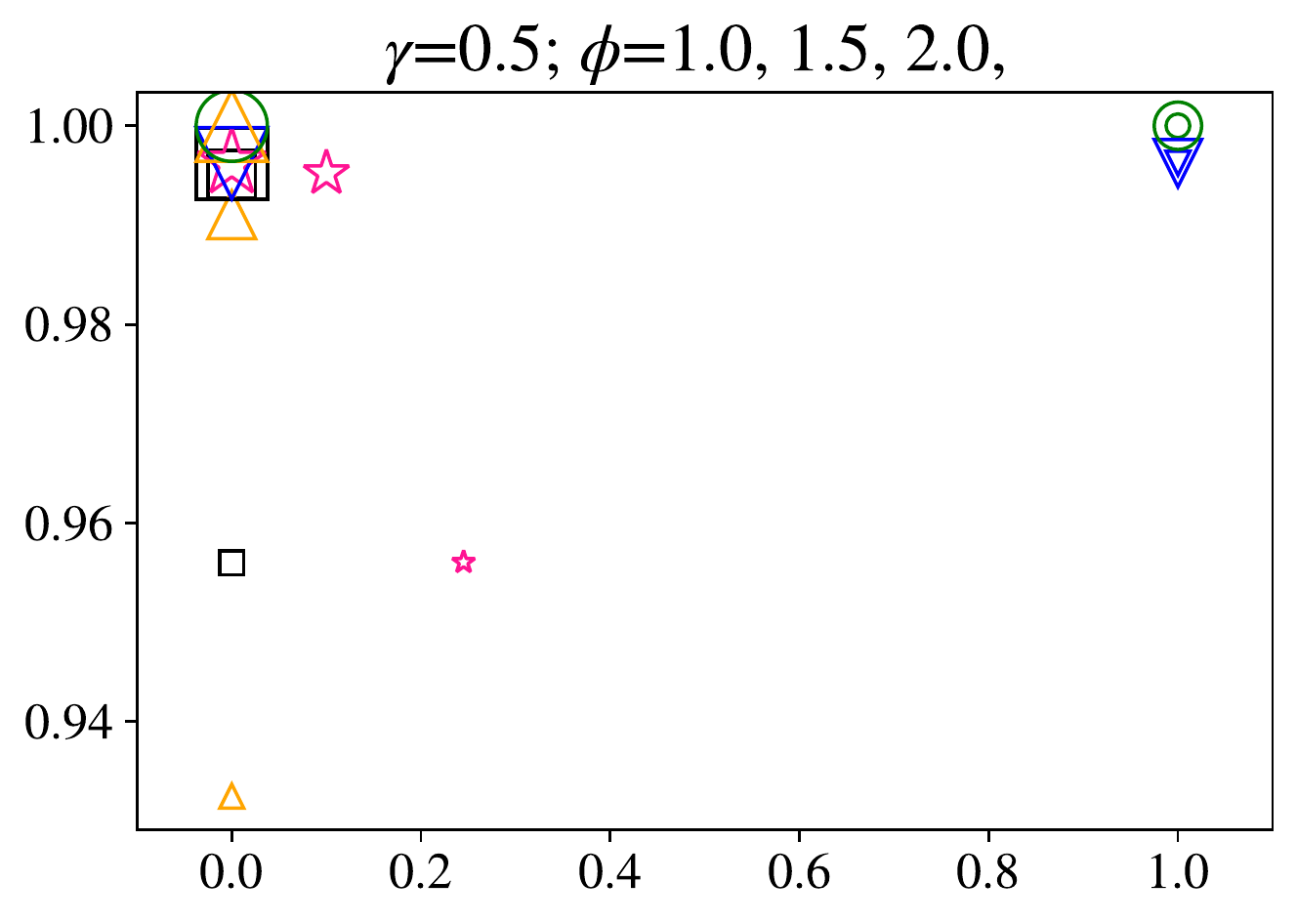}
        \end{subfigure}
        \begin{subfigure}
            \centering
            \includegraphics[scale=0.35]{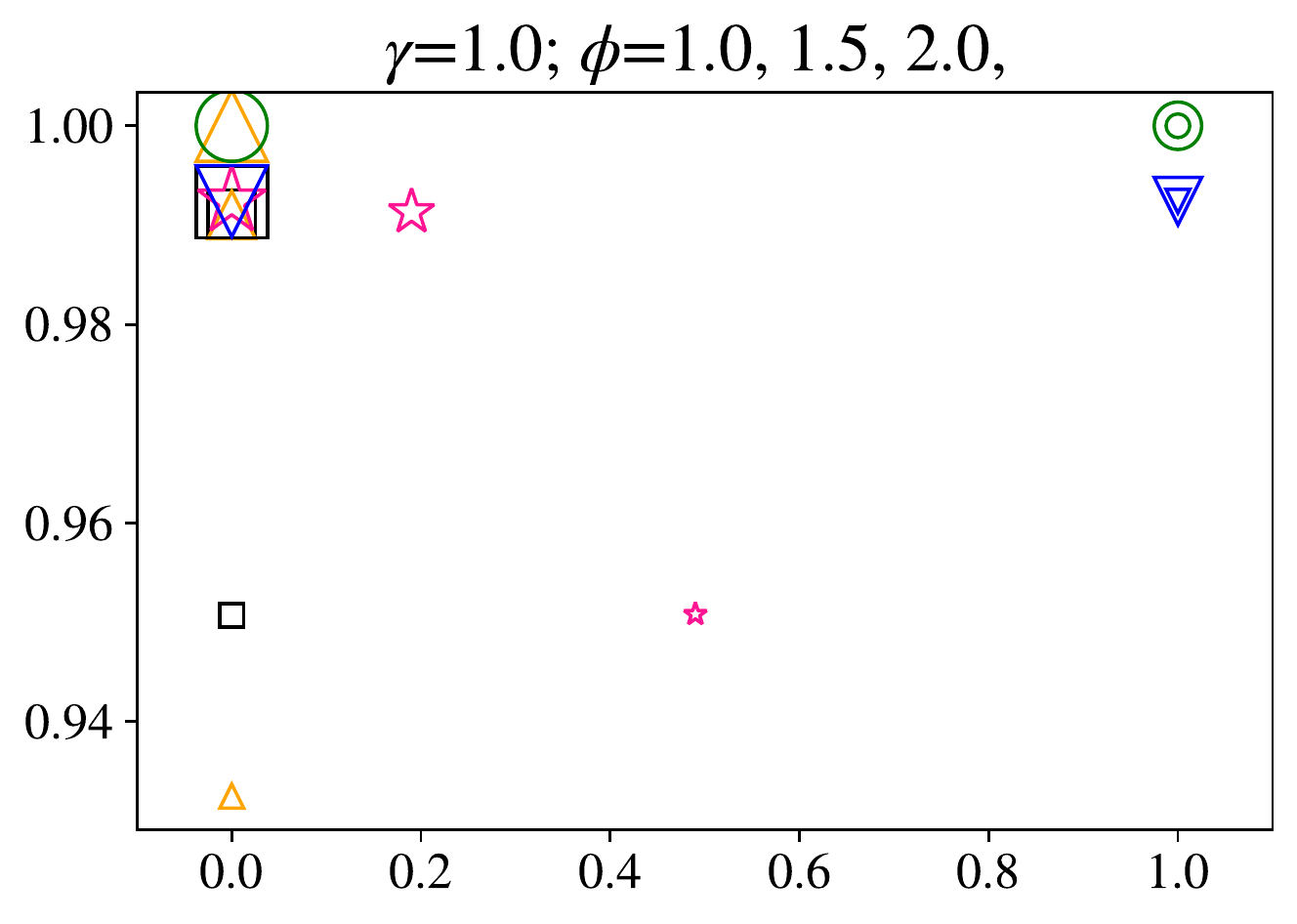}
        \end{subfigure}
        \begin{subfigure}
            \centering
            \includegraphics[scale=0.35]{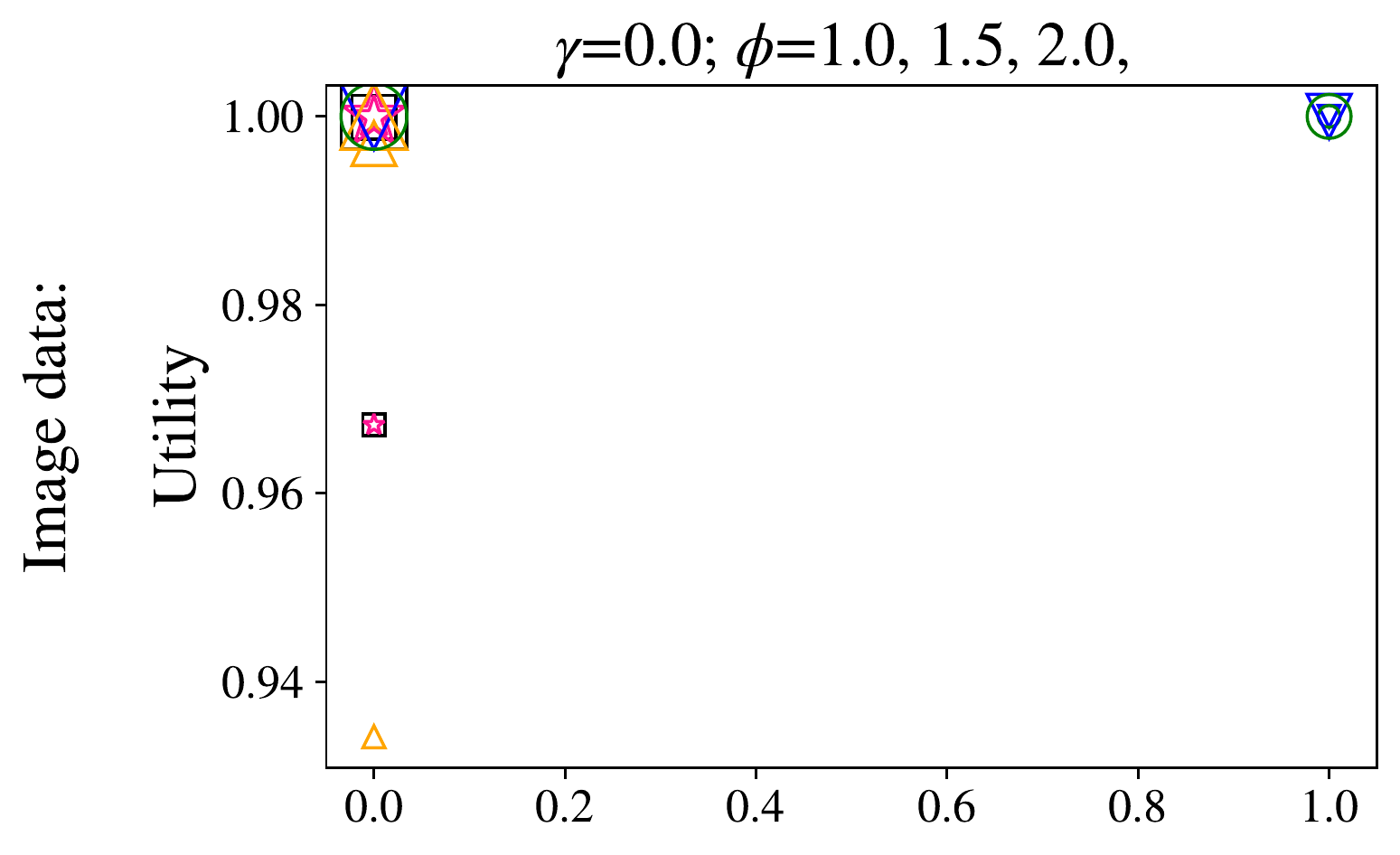}
        \end{subfigure}
        \begin{subfigure}
            \centering
            \includegraphics[scale=0.35]{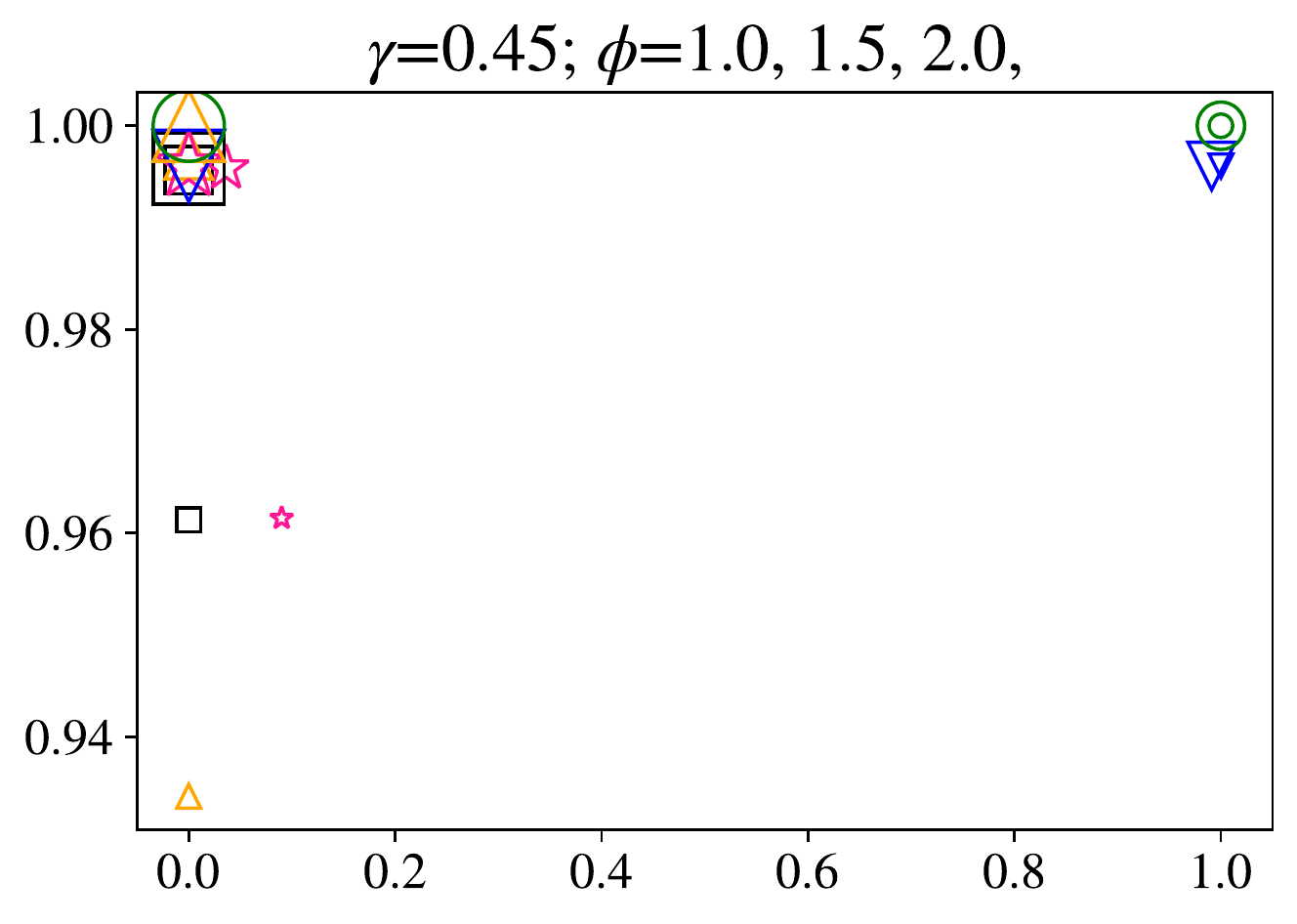}
        \end{subfigure}
        \begin{subfigure}
            \centering
            \includegraphics[scale=0.35]{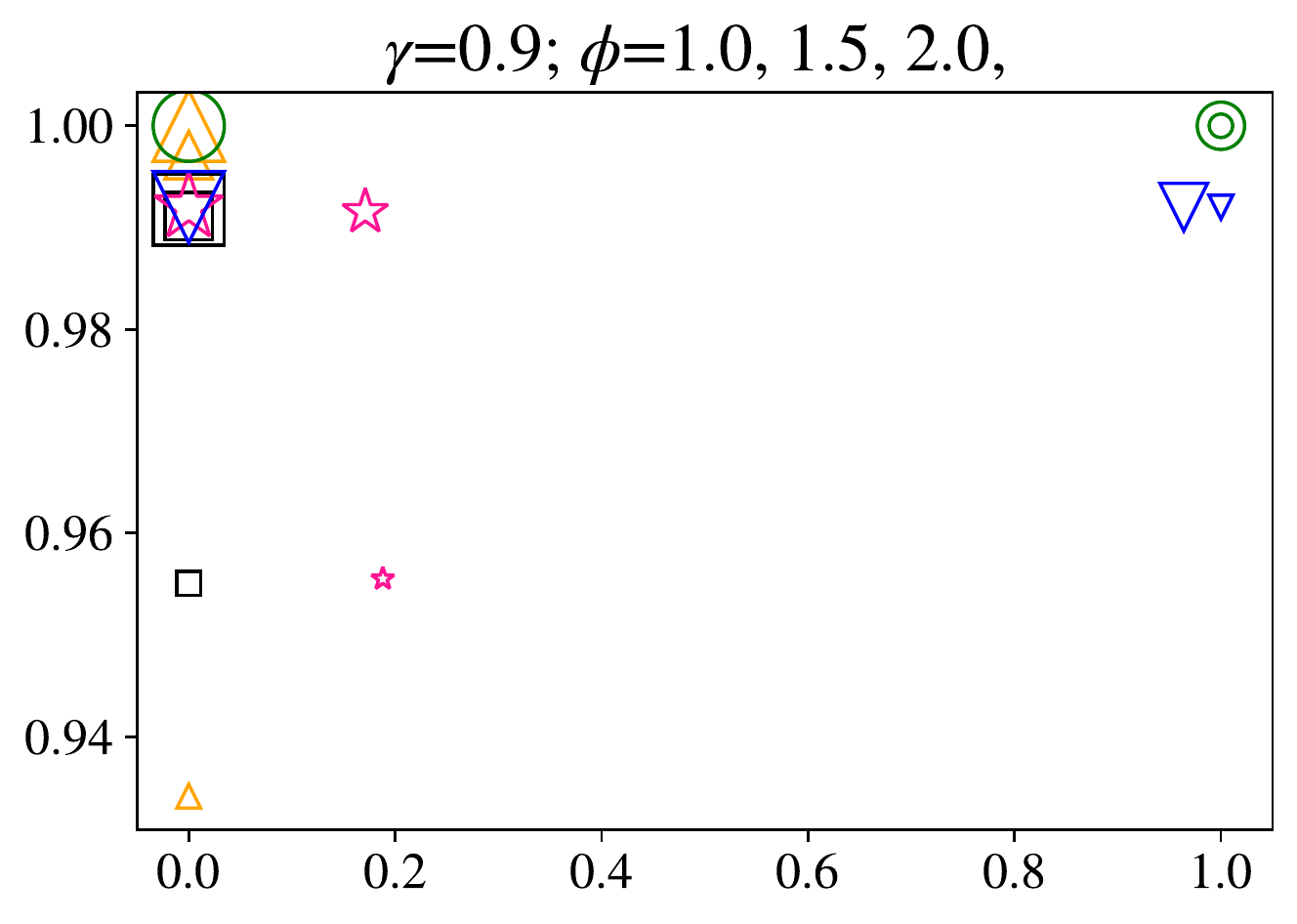}
        \end{subfigure}

        \begin{subfigure}
            \centering
            \includegraphics[scale=0.35]{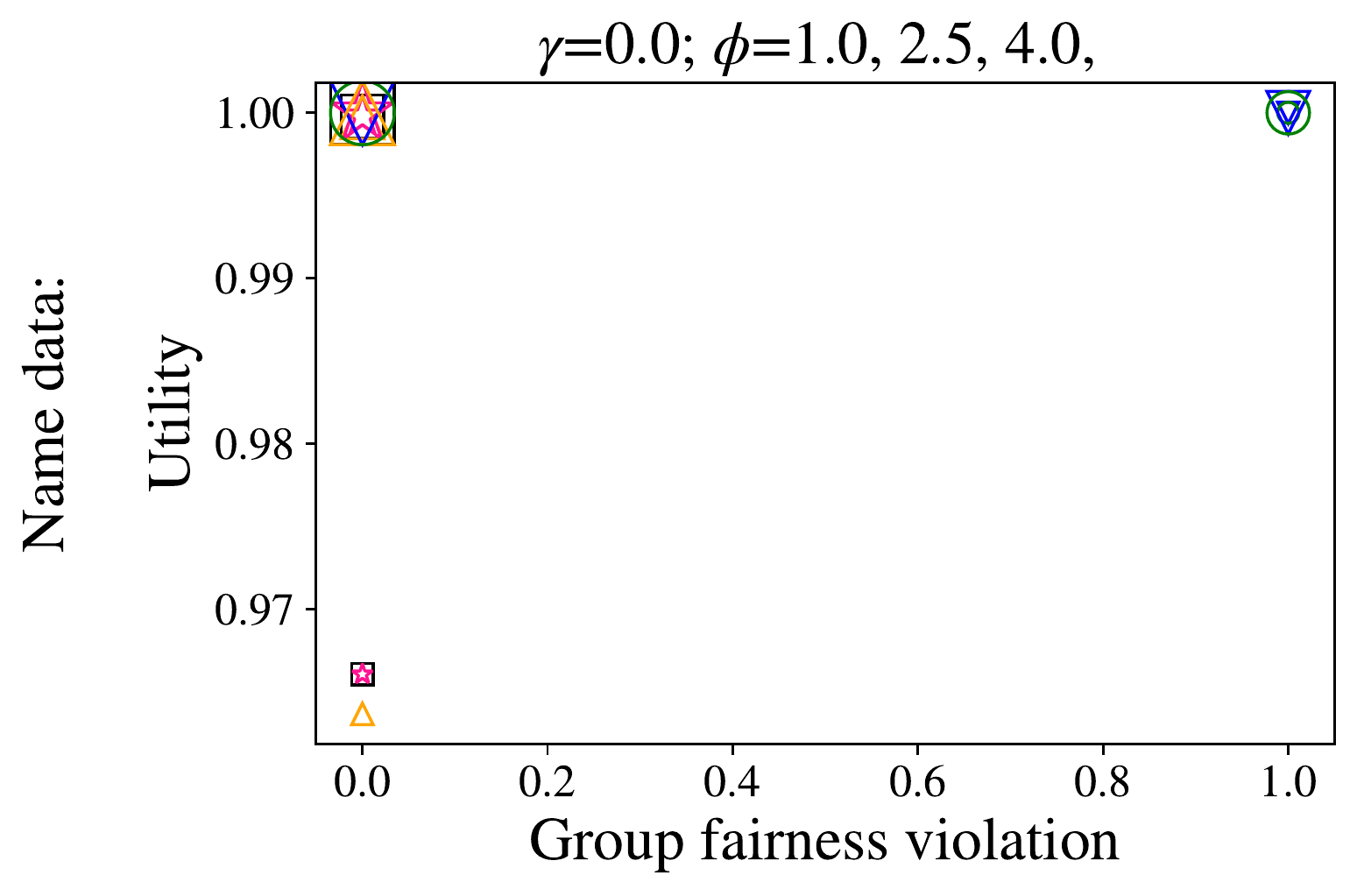}
        \end{subfigure}
        \begin{subfigure}
            \centering
            \includegraphics[scale=0.35]{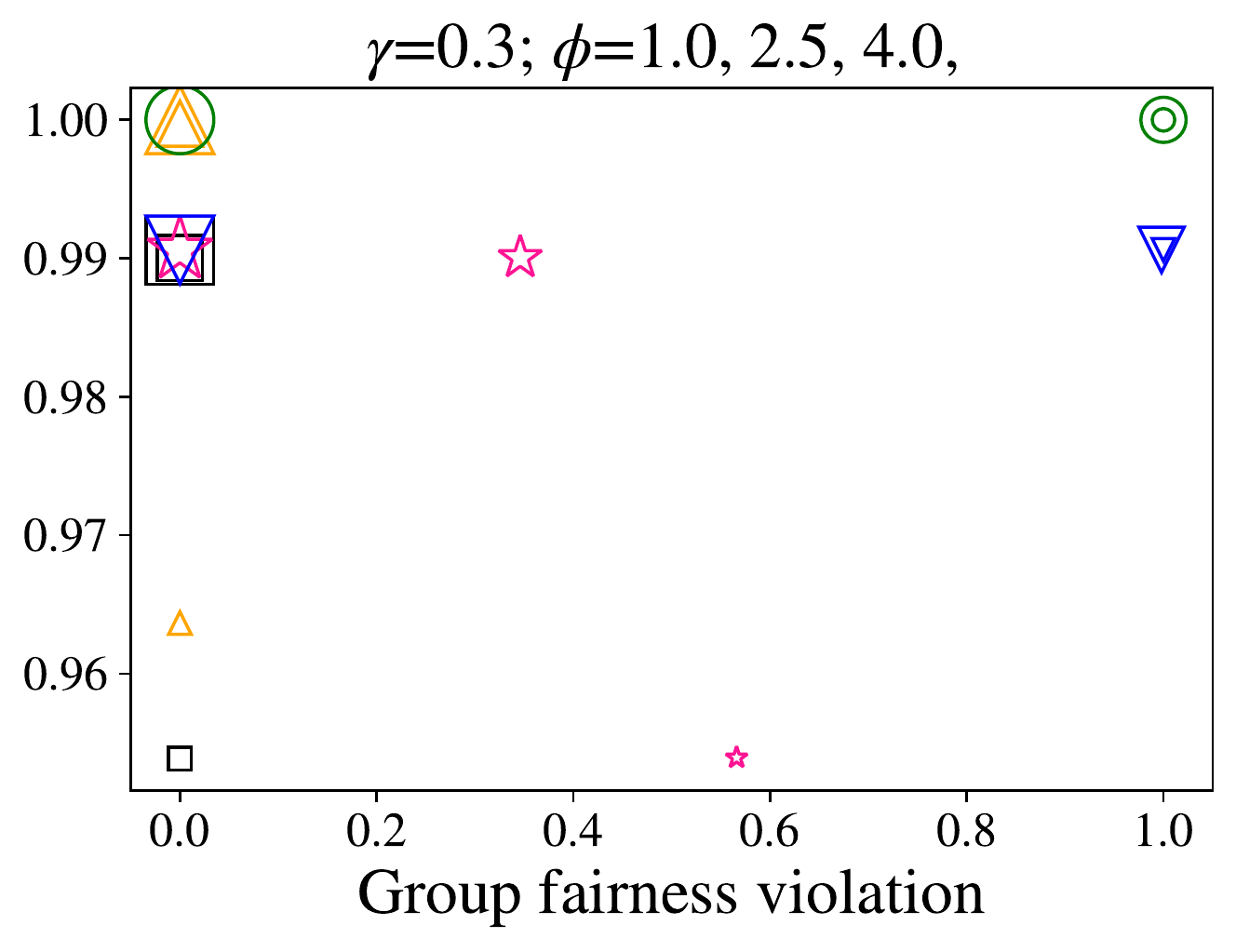}
        \end{subfigure}
        \begin{subfigure}
            \centering
            \includegraphics[scale=0.35]{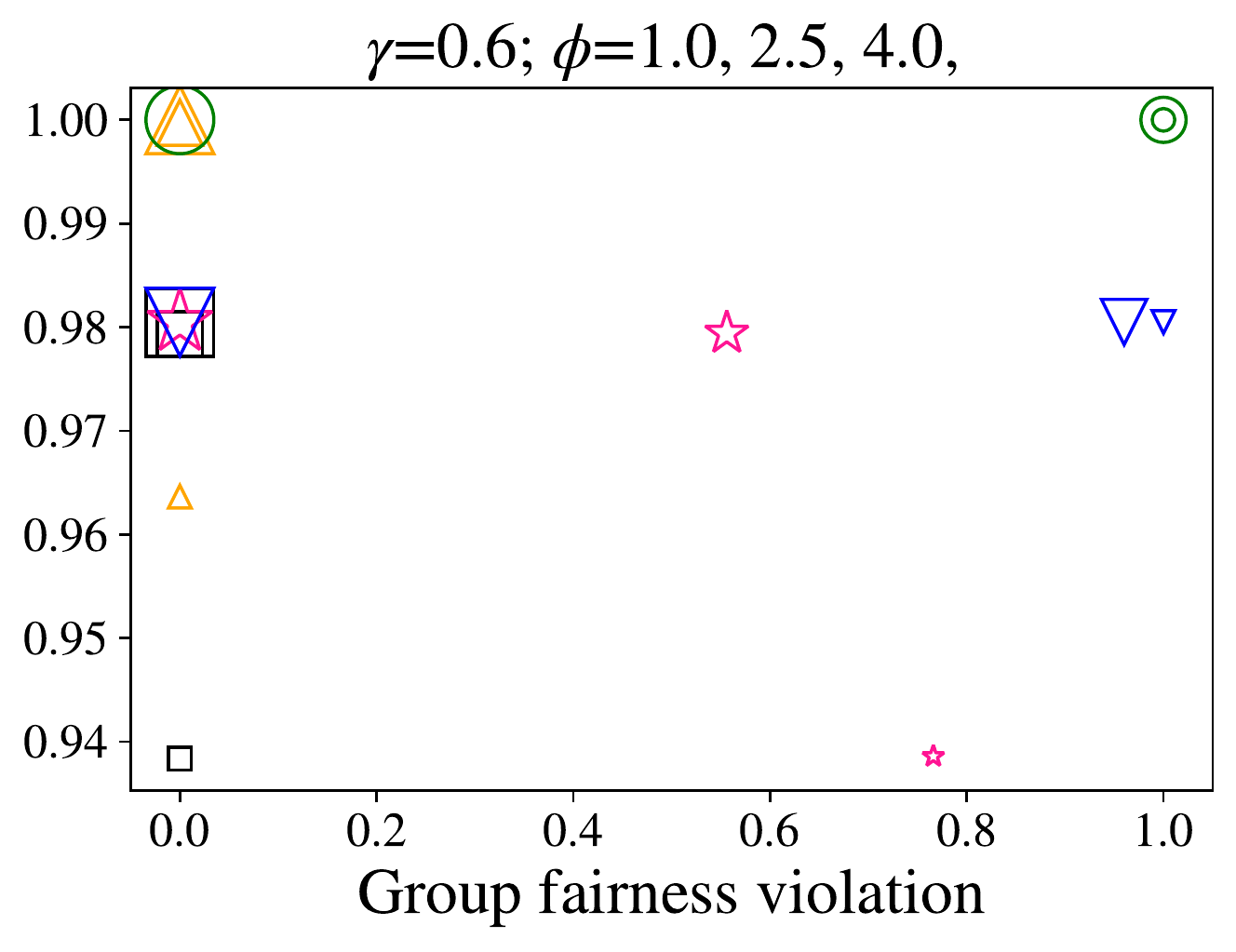}
        \end{subfigure}

        \caption{\textbf{Utility vs. Group fairness violation:} In the plots, the parameter $\gamma$ controls individual fairness constraints and the parameter $\phi$ defines block-wise representation constraints. The size of the marker for each algorithm in each plot is proportional to the value of $\phi$. Lower the value of $\phi$, the stronger the group fairness constraints. In contrast, the lower the value of $\gamma$, the weaker the individual fairness constraints.}
        \label{fig:utility}
    \end{figure*}

\begin{table}[t]
    \centering
    \begin{tabular}{|l|c|c|c|c|}
        \hline
        \textbf{Dataset} & \textbf{$m$} & \textbf{$n$} & \textbf{$k$} & \textbf{$p$} \\
        \hline
        Synthetic \cite{Mehrotra2021MitigatingBI} & 100 & 40 & 20 & 2\\
        Real-world Images \cite{mehrotra2022fair} & 100 & 20 & 10 & 2 \\
        Real-world Names \cite{celiskeswani20} & 400 & 16 & 8 & 4 \\
        \hline
    \end{tabular}
    \vspace{2mm}
    \caption{Parameter choices for each dataset. Experiments with additional parameter choices are presented in \Cref{fig:syn_supp}.}
    \label{tab:hyperparams}
\end{table}

\subsection{Datasets} We perform simulations with three datasets.

\paragraph{Synthetic dataset.}
    We use the synthetic dataset generated by the code provided by recent work on fair ranking \cite{mehrotra2022fair}:
    this dataset consists of two protected groups $G_1$ and $G_2$, where $G_1$ comprises $60\%$ of the total items, and the utilities of items in the minority group are systematically lower (with mean $0.35$) compared to the utilities of items in the majority group (with mean 0.7).

    \smallskip\paragraph{Real-world image dataset.} This dataset, also known as the Occupations dataset, consists of the top $100$ Google Image results for $96$ queries \cite{celiskeswani20}.
    For each image, the dataset provides the rank of the image in the search result and the (apparent) gender of the person in the picture (encoded as {binary} labels collected via MTurk) \cite{celiskeswani20}.
    We use the same preprocessing as \citet{mehrotra2022fair}:
    let an occupation be stereotypical if more than $80\%$ of the images in the corresponding search result are labeled to be of a specific gender.
    This results in $41/96$ stereotypical occupations, with 4,100 images.
    Each of the 4,100 images, is labeled as stereotypical if the image's gender label corresponds to the majority gender label in the corresponding occupation and is otherwise labeled as unstereotypical.
    We consider the set of stereotypical and unstereotypical images as protected groups and fix $\rho_i = \frac{1}{\log(1+r_i)},$ for all $1\leq i\leq m$ (as in \cite{mehrotra2022fair}).

\paragraph{Real-world names dataset.}
    This dataset, known as the chess ranking data, consists of the which consists of the FIDE rating of 3,251 chess players across the world \cite{ghosh2021uncertain}.
    For each player, the dataset consists of the players' self-identified gender (encoded as binary: male or female) and their self-identified race (encoded as Asian, Black, Hispanic, or White).
    For the simulation with this dataset, we consider the following four intersectional groups: White Male, White Female, Non-White Male, and Non-White Female.

    \paragraph{Data-specific parameters and setup.} \Cref{tab:hyperparams} lists the parameters $n$, $m$, $k$, and $p$ for each dataset.
    For the specified values of $n$ and $m$, in each simulation, we sample a subset of each dataset to select $m$ items from the data uniformly without replacement -- where $m$ is chosen to be the smallest value (up to a multiple of 100) so that in each draw there are at least $n/p$ items from each group (to ensure that the equal representation constraints are satisfiable).

\subsection{Observations and Discussion}\label{sec:results}

We now summarize our experimental observations and answer the research questions raised at the beginning of this section.

\paragraph{Pareto-optimality for individual and group fairness.}
\Cref{fig:fairness} presents the results that compare the individual and group fairness violations achieved by all the algorithms.
In \Cref{fig:fairness}, each row corresponds to results over the three respective datasets.
The sub-figures in the first column show group fairness violation against individual fairness violation by the algorithms.
Obviously, when no individual fairness constraints are enforced, i.e., $\gamma = 0$, all the algorithms achieve $0$ individual fairness violation.

We observe that our algorithm achieves Pareto-optimality with respect to individual $\ifviol{}$ and group-fairness $\gfviol{}$.
When the value of $\gamma$ is increased (sub-figures on the second and third columns), we see that the baseline \textbf{SJK21 (GF and IF)} violates group fairness constraints for smaller values of $\phi$ ($\phi\in\inbrace{1, 1.5}$ for the synthetic and image dataset and $\phi\in\inbrace{1, 2.5}$ for the Name dataset). This is caused because there are non-group fair rankings in the support of the distribution output by this algorithm. \textbf{SJK21 (IF)} has high group fairness violation ($\gfviol{}\approx 1$ for almost all non-trivial values of $\phi$ and $\gamma$) and no individual fairness violation as expected. Since \textbf{CSV18 (Greedy)} outputs a deterministic ranking, it has high individual fairness violation ($\ifviol{}\ge 0.8$ for all values of $\phi$ and $\gamma$) but has no group fairness violation. We note here that any deterministic ranking can have either $0$ or $1$ group fairness violation. Finally, \textbf{Unconstrained} does the worst of all the algorithms with $\gfviol{}=1$ and $\ifviol{}\ge 0.8$ on all the datasets.
In contrast, our algorithm does the best by achieving $0$ individual fairness violation and $0$ group fairness violation, thus indicating Pareto dominance over all other baselines.

\paragraph{Utility vs. Fairness}
\Cref{fig:utility} shows the utility vs.~group fairness violation plot for values of $\phi$ ranging from $1$ to $p$. Our algorithm achieves a very small decrease in the utility; that is, across all the datasets, our algorithm suffers only a maximum of 6\% loss in utility compared to other algorithms when all the algorithms are subject to the same fairness constraints. We also observe similar trends as \Cref{fig:fairness} on the synthetic dataset for other values of the parameters $m$ and $n$ (see \Cref{fig:syn_supp}).

\begin{figure*}[t!]
        \centering
        \begin{subfigure}
            \centering
            \includegraphics[scale=0.5]{figures/legend.pdf}
        \end{subfigure}

        \begin{subfigure}
            \centering
            \includegraphics[scale=0.35]{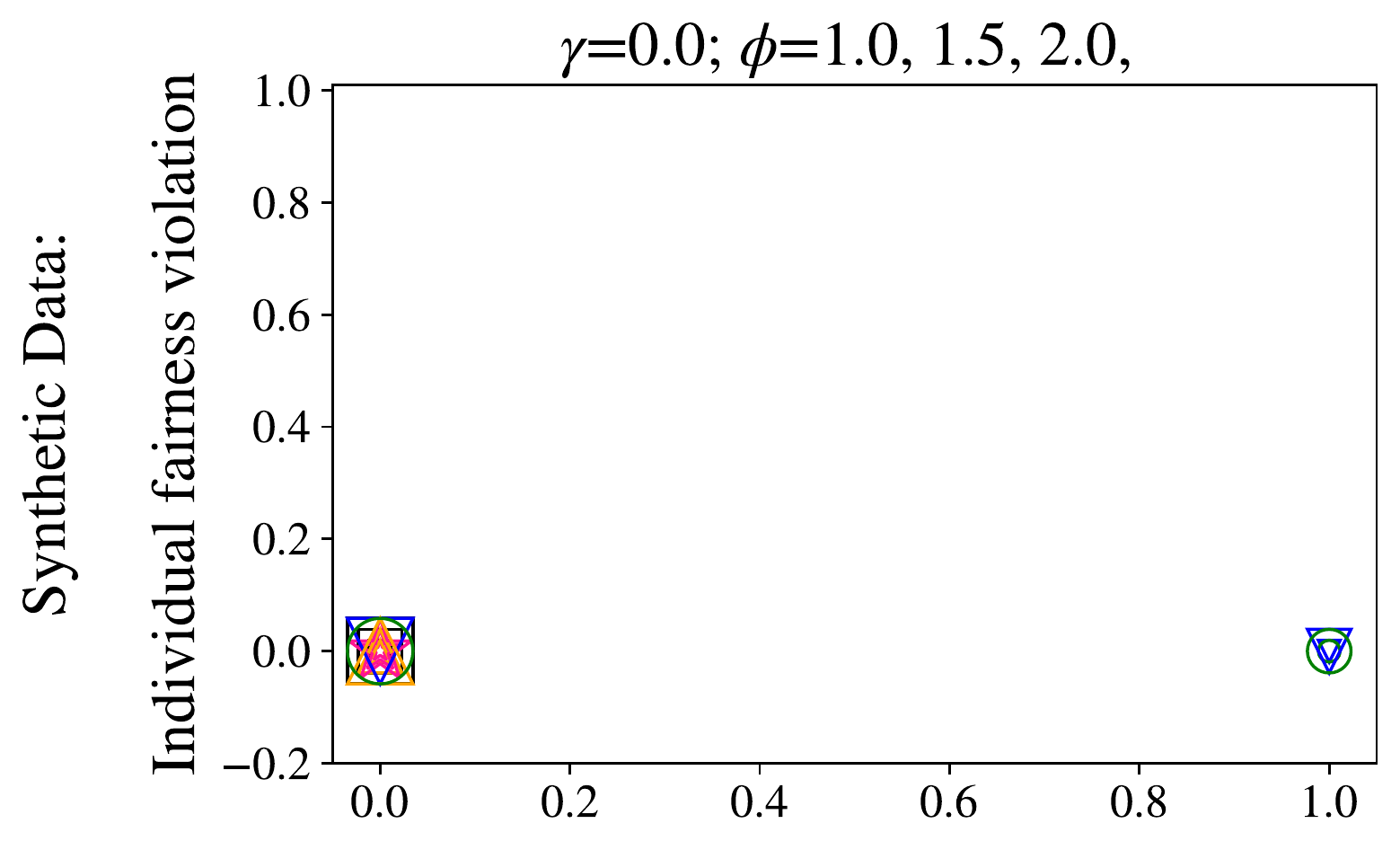}
        \end{subfigure}
        \begin{subfigure}
            \centering
            \includegraphics[scale=0.35]{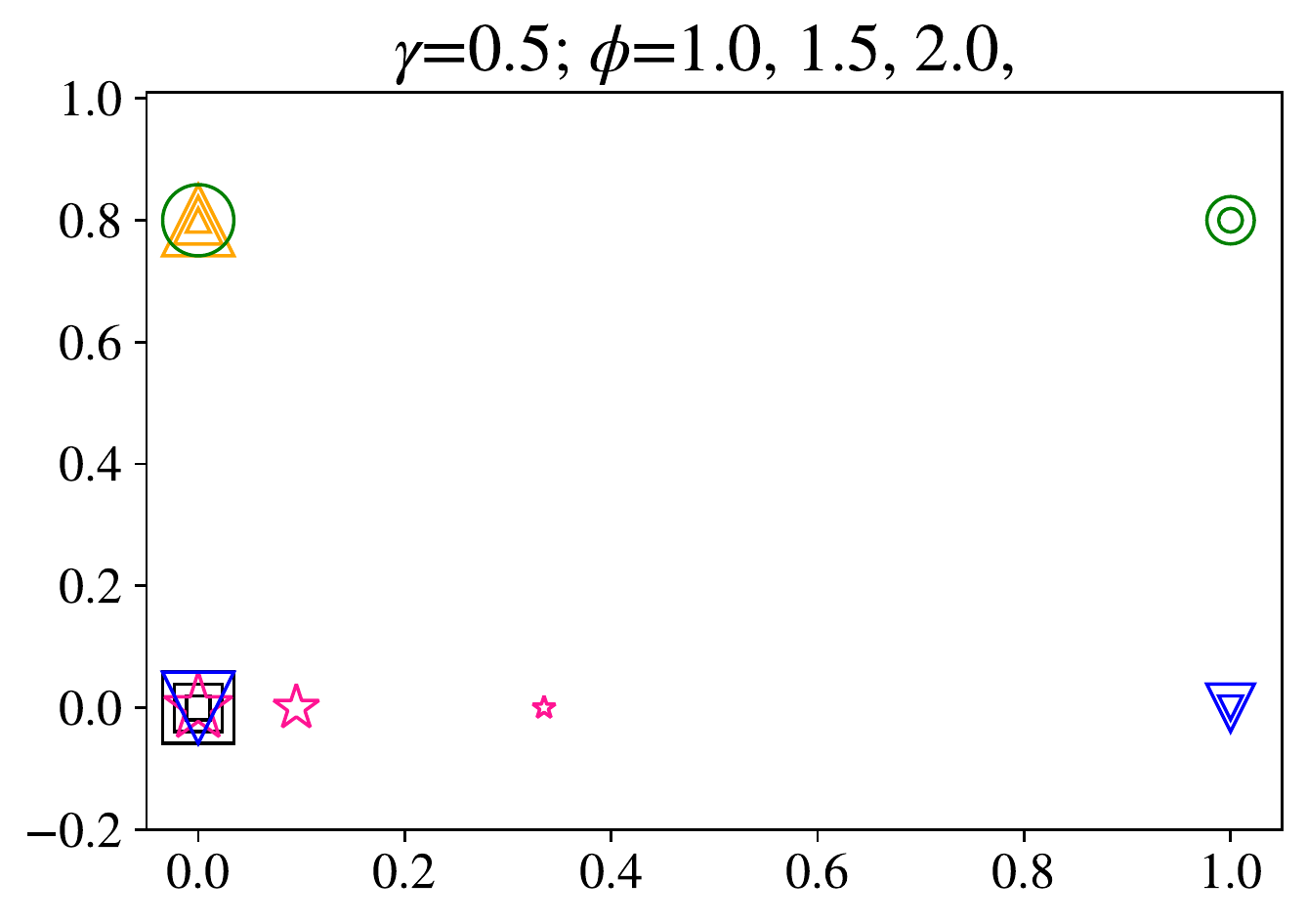}
        \end{subfigure}
        \begin{subfigure}
            \centering
            \includegraphics[scale=0.35]{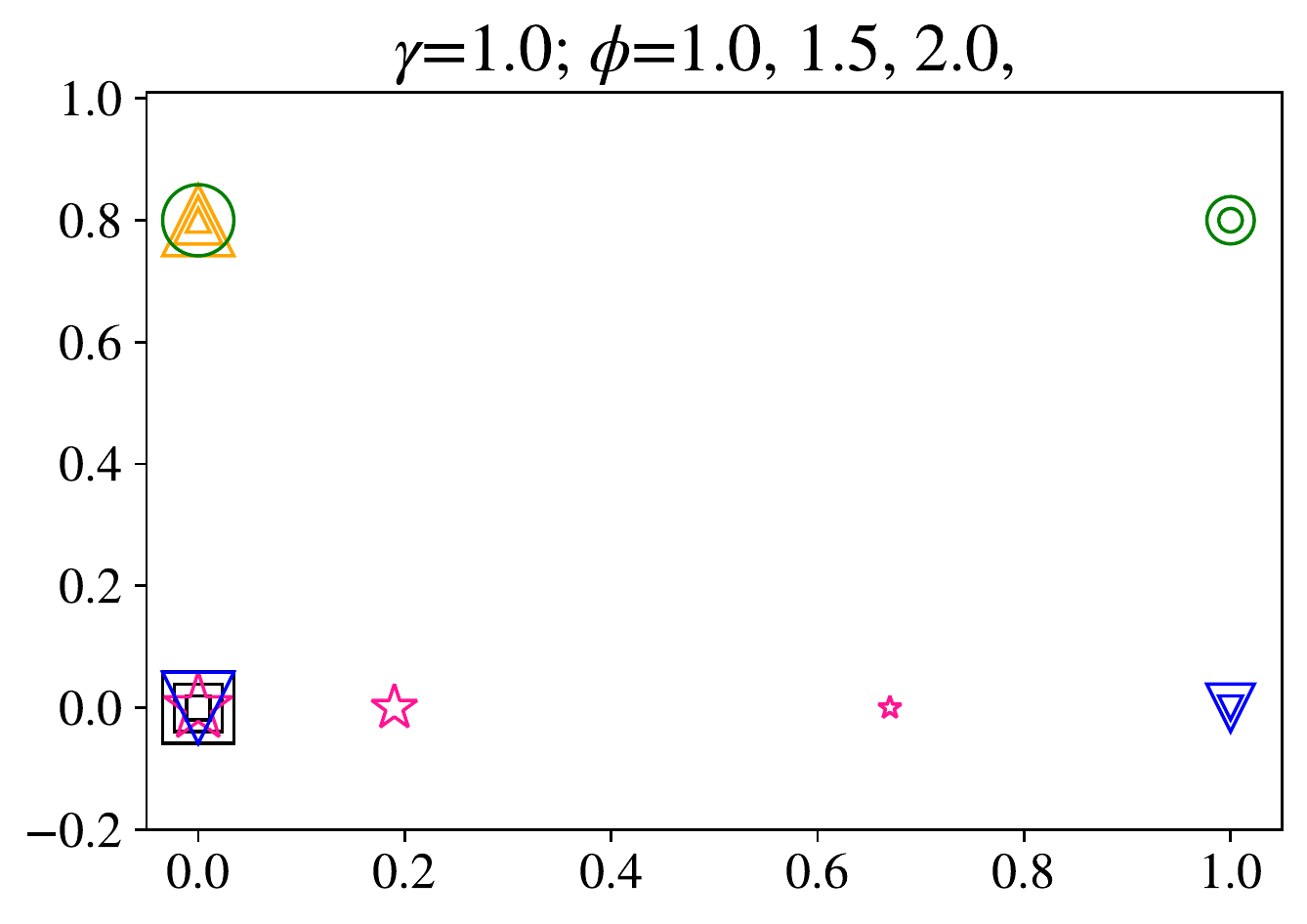}
        \end{subfigure}

        \begin{subfigure}
            \centering
            \includegraphics[scale=0.35]{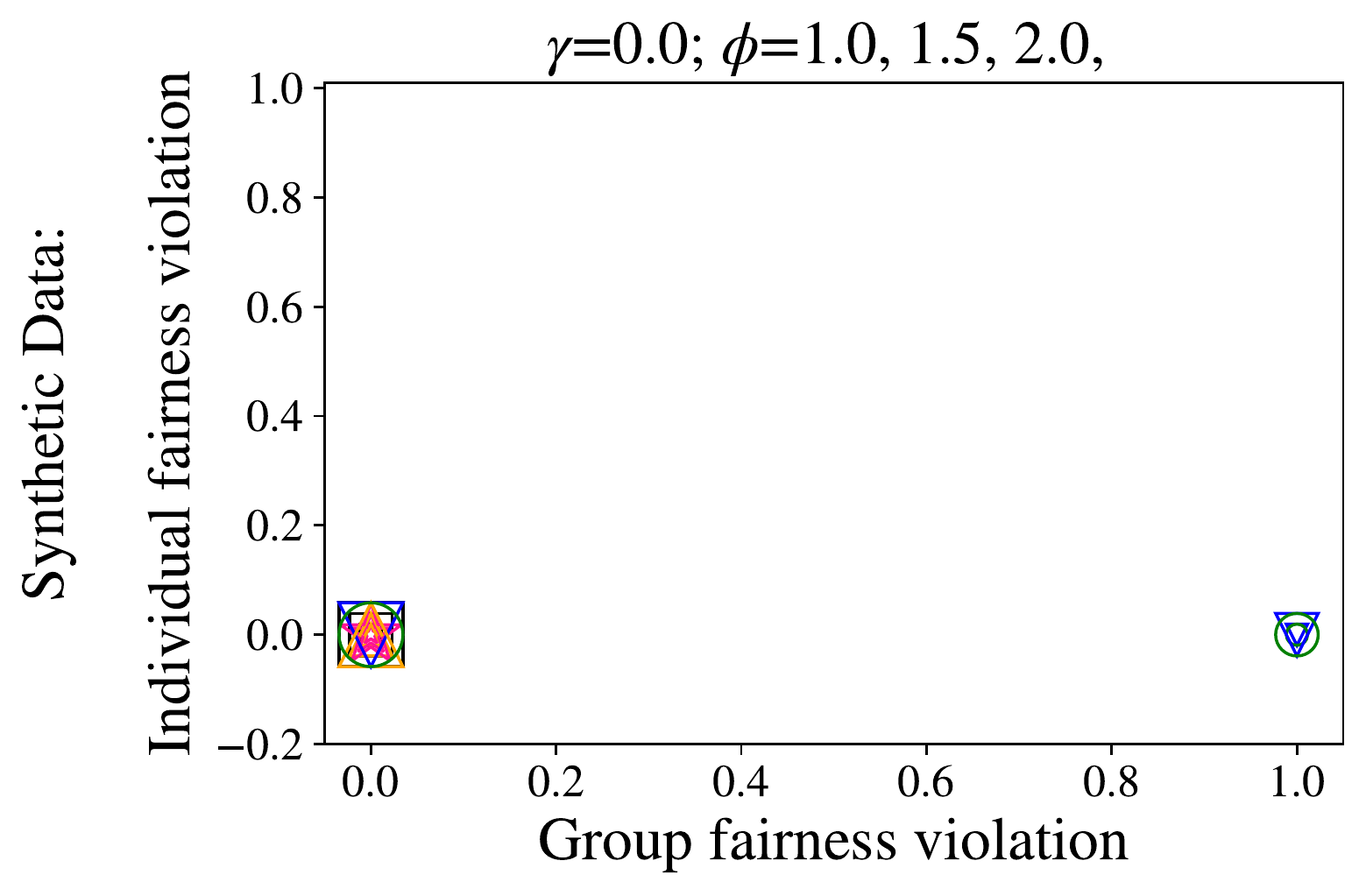}
        \end{subfigure}
        \begin{subfigure}
            \centering
            \includegraphics[scale=0.35]{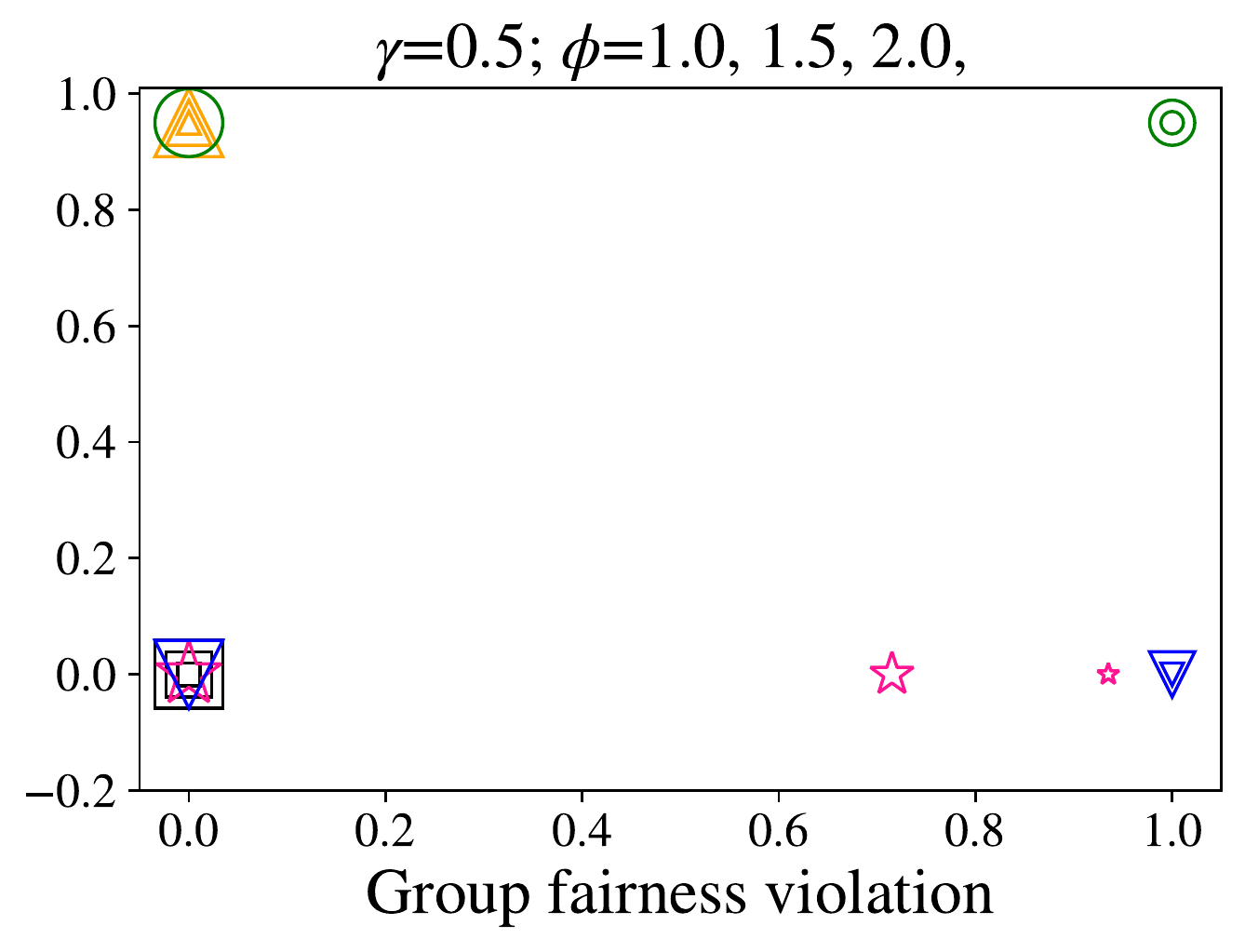}
        \end{subfigure}
        \begin{subfigure}
            \centering
            \includegraphics[scale=0.35]{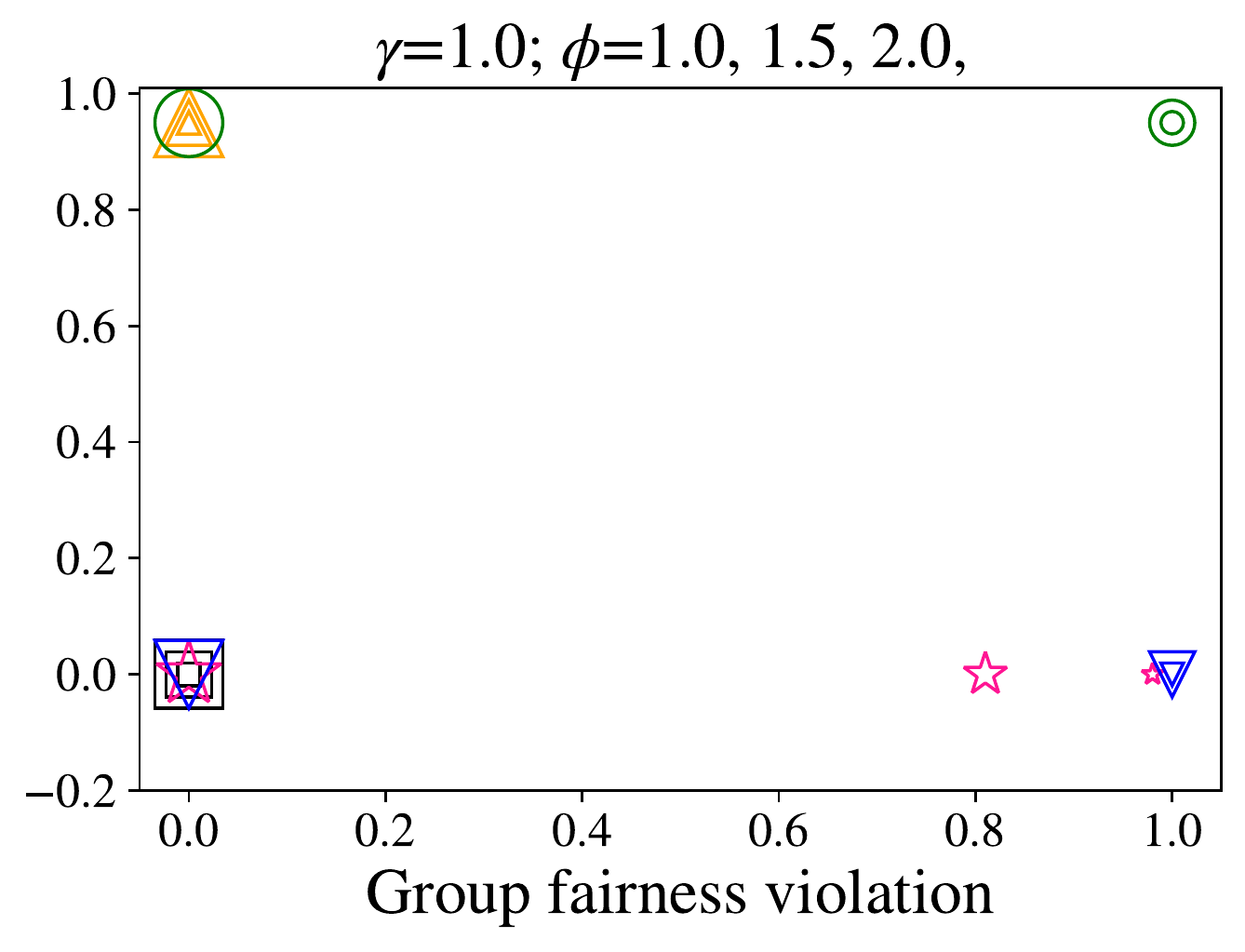}
        \end{subfigure}

        \caption{\textbf{Simulations with higher number of blocks:} Individual fairness violation vs. Group fairness violation on synthetic data where the first row is with $m = 100$, $k=20$,  and $n = 60$, that is top-$3$ blocks and the second row is with $m = 400$, $k=20$, and $n = 80$, that is top-$4$ blocks: In the plots, the parameter $\gamma$ controls individual fairness constraints and the parameter $\phi$ defines block-wise representation constraints. The size of the marker for each algorithm in each plot is proportional to the value of $\phi$. Lower the value of $\phi$, the stronger the group fairness constraints. In contrast, the lower the value of $\gamma$, the weaker the individual fairness constraints.}
        \label{fig:syn_supp}
    \end{figure*}

\section{Theoretical Overview}\label{sec:overview}
    In this section, we explain the key ideas in the proof of our main theoretical result, \cref{thm:algo_main}.
    Recall that \cref{thm:algo_main} proves that there is an algorithm, namely \cref{alg:main}, that given matrices $L, U$ specifying the group-fairness constraints and matrices $A,C$ specifying the individual fairness constraints along with vector $\rho$ specifying item utilities, outputs a ranking $R$ sampled from distribution $\cD$ such that:
    (1) $R$ \textit{always} satisfies the group fairness constraints and (2) $\cD$ satisfies the individual fairness constraints.
    Moreover, the expected utility of $R$ is at least $\alpha$ times the optimal--where $\alpha$ satisfies the lower bound in \cref{eq:lowerbound_on_alpha} (see {\cref{sec:theoretical_results} for tighter bounds on $\alpha$ under additional assumptions).}

    The algorithms in prior work \cite{fairExposureAshudeep,AshudeepUncertainty2021} roughly have the following structure:
        they first solve a linear program (e.g., \cref{prog:from_ashudeep}) to compute a marginal $\hD\in [0,1]^{n\times m}$ of distribution $\wh{\cD}$ from which they want to sample the ranking, and then they decompose $\hD$ into a convex combination of at most $\poly(n,m)$ rankings: $\hD=\sum_t \alpha_t R_t$.
        Let's represent this pictorially by the following two-step process:

        \vspace{-0mm}

        \makebox[\linewidth]{
            \centering
            \white{.}\hspace{5mm}
            \includegraphics[width=0.8\textwidth]{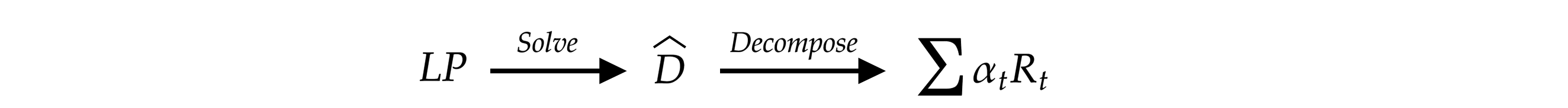}
        }

        \vspace{-0mm}

        \noindent In general, there exist infinitely many decompositions of any matrix $\hD$ and any valid decomposition is suitable for prior work \cite{fairExposureAshudeep,AshudeepUncertainty2021}:
        this is because, for any valid decomposition $\sum_t \alpha_t R_t$, sampling ranking $R_t$ with probability $\propto \alpha_t$ (for all $t$) results in optimal utility and satisfies the fairness constraints (both individual and group) in expectation.
        We, however, require a decomposition where \textit{each} ranking $R_t$ in the decomposition satisfies the group fairness constraints.
        Computing such a decomposition is the key technical difficulty in proving \cref{thm:algo_main}.

        In fact, such a decomposition may not even exist (see \cref{lem:existence_of_extreme_point} for an example).
        Moreover, instances, where the decomposition does not exist, are also not ``isolated'' or avoidable instances.
        Rather, they arise due to a fundamental property of group fairness constraints: that the set of matrices $R$ which satisfy the group fairness constraint and the constraint \cref{eq:prog_from_ashudeep_eq_2} form a polytope that has fractional vertices, which are vertices that (in matrix representation) have fractional entries.
        This is also related to the computational complexity of the problem:
        the Birkhoff von Neumann algorithm is able to efficiently compute a decomposition (when output all rankings are not required to satisfy group fairness constraints) precisely because such fractional vertices do not arise.
        Indeed, a deep result in Combinatorial Optimization is that the set of doubly-stochastic matrices -- which is the set of matrices that satisfy \cref{eq:prog_from_ashudeep_eq_2} (but do not necessarily satisfy the group fairness constraints) -- form a polytope that does not have fractional vertices \cite{schrijver2003combinatorial}.
        Such polytopes are said to be \textit{integral}; see \cite{schrijver2003combinatorial}.
        If such an integrality result was true in our setting, then we could have used a straightforward analog of the Birkhoff von Neumann algorithm.
        However, this is not the case as shown in \cref{sec:fractional_vertex}.

        Our idea is to first compute a ``coarse'' version of the decomposition and then ``refine'' it.
        Pictorially, our algorithm follows the following four-step process.

        \vspace{0mm}

        \makebox[\linewidth]{
            \white{.}\hspace{-10mm}
            \includegraphics[width=0.8\textwidth]{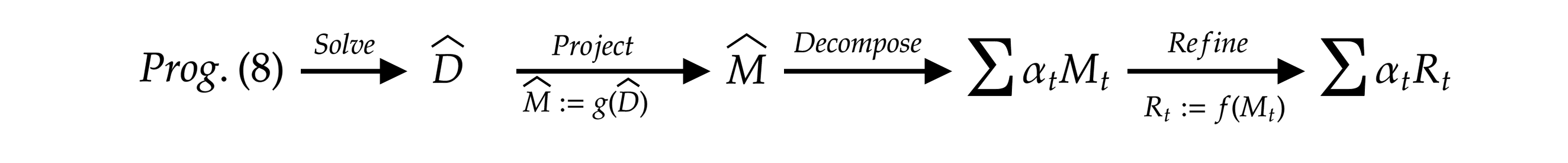}
        }

        \vspace{-02mm}

        \noindent Intuitively, a ``coarse ranking'' or a matching is an assignment of $m$ items to $q$ blocks $B_1,B_2,\dots,B_q$.
        Each item is matched to exactly 1 block and the $j$ the block has exactly $\abs{B_j}$ items.
        We encode a matching by an $m\times q$ matrix $M\in \inbrace{0,1}^{m\times q}$ where $M_{ij}=1$ item $i$ is in the block $B_j$ and $M_{ij}=0$ otherwise.

        Before discussing how to efficiently compute the decomposition $\sum_t\alpha_tM_t$ where each matching $M_t$ satisfies ``group fairness constraints,'' we need to define notions of fairness for matchings.
        \begin{definition}
            Given matrices $L,U\in \Z^{q\times p}$ and $A,C\in [0,1]^{m\times q}$, define the following definitions of fairness:
             \begin{enumerate}
                \item A distribution $\cD^{(\cM)}$ over the set $\cM$ of all matchings satisfies $(C,A)$-individual fairness constraints if for each $i$ and $j$, $C_{ij}\leq \Pr\nolimits_{M\sim \cD^{(\cM)}}\insquare{M_{ij} = 1} \leq A_{ij}.$
                \item A matching $M$ satisfies the $(L,U)$-group fairness constraints if for each $j$ and $\ell$, $L_{j\ell}\leq \sum\nolimits_{i\in G_\ell} M_{ij} \leq U_{j\ell}.$
            \end{enumerate}
        \end{definition}
        \noindent The fairness guarantee of \cref{alg:main} follows because of the following invariance.
        \begin{lemma}\label{lem:invariance}
            Let $f(D)\in [0,1]^{m\times q}$ be the projection of a matrix $D\in [0,1]^{m\times n}$ to the space of matchings.
            For any matrices $L,U\in \Z^{q\times p}$ and $A,C\in [0,1]^{m\times q}$,
            the following holds
            \begin{enumerate}
                \item A matrix $D$ satisfies $(C,A)$-individual fairness constraints if and only if $f(D)$ satisfies $(C,A)$-individual fairness constraints; and
                \item A matrix $D$ satisfies $(L,U)$-group fairness constraints if and only if $f(D)$ satisfies $(L,U)$-group fairness constraints.
            \end{enumerate}
        \end{lemma}
        \noindent Let the ``refinement'' of a matching $M$ be $g(M)$.
        One can show that the projection of the refinement of $M$, i.e., $g(f(M))$ is $M$ itself. %
        This and \cref{lem:invariance} imply that the refinement $f(M)$ of a matching $M$ satisfies the $(L,U)$-group fairness constraints if and only if $M$ satisfies $(L,U)$-group fairness constraints.
        Since the marginal $\hD$ computed by our algorithm satisfies group fairness and individual fairness constraints, chaining the above invariance results over the four steps presented above implies that each ranking $R_t$ output by our algorithm also satisfies individual fairness and group fairness constraints.

        It remains to show that \cref{alg:main} is efficient and to establish its utility guarantee.
        A lower bound on the utility follows straightforwardly due to the following:
        First, one can show that both $\hD$ (which has the optimal utility) and $\sum_t \alpha_t R_t$ project to the same point in the matching space $\hM$.
        Second, one can lower bound the ratio of the utility of any two points $D_1,D_2\in [0,1]^{m\times n}$, in the ranking space, that have the same projection in the matching space (\cref{lem:utility_guarantee_2}).

        As for the efficiency of \cref{alg:main}, it follows because the set of matchings satisfies the ``integrality'' property we discussed above.
        This is why we introduce matching into our algorithm.
        Consider the set of points in the matching space that satisfy the group fairness constraint and the following analog of \cref{eq:prog_from_ashudeep_eq_2}:
        $$\forall 1\leq j\leq q, \ \ \sum\nolimits_{i} D_{ij} = \abs{B_j}\quad\text{and}\quad\forall 1\leq i\leq m, \ \ \sum\nolimits_{j} D_{ij}\leq 1.$$
        Formally, this set of points forms a polytope $\mathscr{M}$ that is integral.
        Somewhat surprisingly, this connects our work to a recent work on fair matchings that also makes this observation  \cite{pln_2022}.
        The polytope $\mathscr{M}$ is a part of the problem statement of \citet{pln_2022}.
        Our key insight is to make the connection between rankings and matchings discussed in this section, which may be of independent interest to works designing fair ranking algorithms.

        \paragraph{Extension to laminar families of protected groups.} The fact that $\mathscr{M}$ is integral is the only part of the proof where we use the fact that the protected groups $G_1,\dots,G_p$ are disjoint.
        All the remaining steps in the proof hold for arbitrary group structures.
        \citet{pln_2022} observe that this property continues to hold when the groups $G_1,\dots,G_p$  form a laminar family, i.e., for any set of groups such that either $G_\ell\subseteq G_k$ or $G_k\subseteq G_\ell$ for all $1\leq \ell,k\leq p$.
        Hence, our proof extends to laminar families of protected groups.
        This provides another example where efficient algorithms continue to exist when the underlying set structure is related from disjoint to laminar; as has also been observed by earlier works in algorithmic fairness \cite{celis2019polarization} and, more broadly, in the Combinatorial Optimization literature \cite{lau2011iterative}.
        {This allows our algorithm to incorporate constraints on certain intersectional groups, as has been argued by seminal works on intersectionality \cite{king1988multiple} and, more recently, by works analyzing mathematical models of intersectional bias \cite{mehrotra2022intersectional}.
        A concrete example is as follows: since the collection of the sets of all women, all Hispanic women, and all Black (non-Hispanic) women is laminar, one can ensure sufficient representation of women in the output ranking and, within women, also ensure the representation of Hispanic and Black women.}

\section{Proof of Theoretical Results}
    \subsection{Proof of \cref{thm:algo_main}}\label{sec:proofof:thm:algo_main}
        In this section, we prove \cref{thm:algo_main}. For the reader's convenience, we restate the theorem below.

        \thmAlgoMain*

        \paragraph{Outline of the proof of \cref{thm:algo_main}.}
        The key idea is to consider a family of ``coarse rankings'' or \textit{matchings}: Each matching places $\abs{B_j}$ items in block $B_j$ (for each $j\in [q]$) but does not specify which positions in $B_j$ are these items are placed at.
        We encode a matching by an $m\times q$ matrix $M\in \inbrace{0,1}^{m\times q}$ where $M_{ij}=1$ item $i$ is in the block $B_j$ and $M_{ij}=0$ otherwise.
        We give natural functions $g$ (respectively $f$) to convert a ranking to a matching (respectively a matching to a ranking).
        We also give natural analogs of the group fairness constraints and the individual fairness constraints for matchings (\cref{def:group_constraints_matchings,def:individual_constraints_matchings}).
        The algorithm promised in \cref{thm:algo_main} (\cref{alg:main}) proceeds as follows:
        \begin{itemize}[leftmargin=14pt]
            \item First, it computes the (ranking) marginal $\hD\in [0,1]^{m\times n}$ that is  an  optimal solution \prog{prog:mod_of_ashudeep}.
            \item Then, it computes the corresponding (matching) marginal $\hM\in [0,1]^{m\times q}$ defined as $\hM\coloneqq \mu(\hD)$
            \item Then, it computes the decomposition $\hM=\sum_t \alpha_t M_t$ such that each $M_t$ satisfies the group fairness constraints
            \item Finally, for each $t$, it computes $R_t\coloneqq \sigma(M_t)$ and outputs $R_t$ with probability $\propto \alpha_t$.
        \end{itemize}
        Crucially, we are able to efficiently solve the matching variant of \cref{prob:main} because the constraints analogous to \cref{eq:group_constraints} (see \cref{eq:group_constraints_matchings}) form a polytope all of whose vertices are integral, i.e., an integral polytope.
        Whereas, with rankings, the constraints in \cref{eq:group_constraints} form a fractional polytope (see \cref{sec:fractional_vertex} for an example).
        The proof of \cref{thm:algo_main} is divided into three parts.
        \begin{itemize} %
            \item \cref{sec:fairness_guarantee} proves the fairness guarantees of \cref{thm:algo_main},
            \item \cref{sec:utility_guarantee} proves the utility guarantee of \cref{thm:algo_main},
            \item \cref{sec:efficiency_guarantee} proves that \cref{thm:algo_main} runs in polynomial time.
        \end{itemize}

        \paragraph{Notation.}
            Let $\cM$ be the set of all $m\times q$ matrices denoting a matching, i.e.,
            \begin{align*}
                \cM\coloneqq \inbrace{X\negsp{}\in\negsp{} \inbrace{0,1}^{m\times q} \colon
                \begin{array}{l}
                    \forall_{i\in [m]},\ \sum\nolimits_j X_{ij}\leq 1,\\
                    \forall_{j\in [q]},\ \sum\nolimits_i X_{ij} = \abs{B_j}
                \end{array}
                }
                \yesnum\label{def:set_of_matchings}
            \end{align*}
            {Throughout this section, we assume that $n=m$.
            One can ensure this by adding $m-n$ ``dummy'' additional positions $\inbrace{n+1,n+2,\dots,m}$ and one additional block $B_{q+1}$ consisting of all dummy positions. For $B_{q+1}$, we set vacuous constraints: $L_{q+1,\ell}=0$ and $U_{q+1,\ell}=\abs{B_{q+1}}$  for each $1\leq \ell\leq p$.
            Note that this implies that any $X\in \cM$ satisfies the constraint $\sum_j X_{ij}\leq 1$ with equality (for any $1\leq i\leq m$).}
            Define $g\colon \cR\to \cM$ to be the following function:
            For each $R\in \cR$, $1\leq i\leq m$, and $1\leq j\leq q$,
            \begin{align*}
                g(R)_{ij} \coloneqq \sum\nolimits_{t\in B_j} R_{it}.
                \yesnum\label{def:mu}
            \end{align*}
            Intuitively, $g(R)$ is the unique matching that matches item $i$ to block $B_j$ if and only if item $i$ appears in $B_j$ in $R$.
            $g$ is a many-to-one function: For any two rankings $R$ and $R'$, if the set of items that appear in block $B_j$ in $R$ is the same as the set of items in block $B_j$ in $R'$ (for each $j$), then $g(R)=g(R')$.
            Next, $f\colon \cM\to \cR$ as follows: $f(M)$ is the unique ranking that satisfies
            \begin{itemize}
                \item $g(f(M))=M$ and
                \item for each $j$, items in block $B_j$ appear in non-increasing order of their utility in $f(M)$.
            \end{itemize}
            Given any $M\in \cM$ and $R\in \cR$, $f(M)$ and $g(R)$ can be computed in $O(mn+n\log{n})$ time.

        \subsubsection{Fairness Guarantee of Algorithm~\ref{alg:main}}\label{sec:fairness_guarantee}
            We will use the following definitions of group fairness constraints and individual fairness constraints for matchings.
            \begin{definition}[\textbf{Group fairness constraints for matchings}]\label{def:group_constraints_matchings}
                Given matrices $L,U\in \Z^{q\times p}$ a matching $M\in \cM$ satisfies the $(L,U)$-group fairness constraints if for each $j\in [q]$ and $\ell\in [p]$
                \begin{align*}
                    L_{j\ell}\leq \sum\nolimits_{i\in G_\ell} M_{ij} \leq U_{j\ell}.
                    \yesnum\label{eq:group_constraints_matchings}
                \end{align*}
                A matrix $M\in [0,1]^{m\times q}$ (encoding a marginal) satisfies $(L,U)$-group fairness constraints if for each $j$ and $\ell$, $M$ satisfies \cref{eq:group_constraints_matchings}.
            \end{definition}
            \begin{definition}[\textbf{Individual fairness constraints for matchings}]\label{def:individual_constraints_matchings}
                Given $A,C\in [0,1]^{m\times q}$, a distribution $\cD^{(\cM)}$ over the set $\cM$ of all matchings satisfies $(C,A)$-individual fairness constraints if for each $i\in [m]$ and $j\in [q]$
                \begin{align*}
                    C_{ij}\leq \Pr\nolimits_{M\sim \cD^{(\cM)}}\insquare{M_{ij} = 1} \leq A_{ij}.
                    \yesnum\label{eq:equality_const_indv_fairness_matchings}
                \end{align*}
                A matrix $M\in [0,1]^{m\times q}$ (encoding a marginal) satisfies $(C,A)$-individual fairness constraints if for each $i$ and $j$, $M$ satisfies \cref{eq:equality_const_indv_fairness_matchings}.
            \end{definition}

            \noindent With some abuse in notation, given $L$ and $U$, we use $(L, U)$-group fairness constraints to refer to both the constraints for rankings (\cref{def:group_constraints}) and for matchings (\cref{def:group_constraints_matchings}).
            Similarly, given $A$ and $C$, we use $(C,A)$-individual fairness constraints to refer to both the constraints for rankings (\cref{def:individual_constraints}) and for matchings (\cref{def:individual_constraints_matchings}).
            The meaning will be clear from the context.

            The fairness guarantee of \cref{alg:main} relies on the following two lemmas. We prove the lemmas here. We complete the proof of \cref{alg:main}'s fairness guarantee in \cref{sec:complete_proof_algo_main}.

            \begin{lemma}\label{lem:equivalence_of_gf}
                For any matrices $L,U\in \Z^{m\times q}$ and matrix $D\in [0,1]^{m\times n}$ encoding the marginals of some distribution over rankings,
                the matrix $D^{(\cM)}\coloneqq g(D)$ satisfies $(L,U)$-group fairness constraints if and only if $D$ satisfies \cref{eq:mod_of_ashudeep_eq_2}.
            \end{lemma}
            \begin{proof}
                Fix any $1\leq \ell\leq p$ and $1\leq j\leq q$.
                \begin{align*}
                    \sum\nolimits_{i\in G_\ell} D^{(\cM)}_{ij}
                    \ \ \Stackrel{\eqref{def:mu}}{=} \ \
                    \sum\nolimits_{i\in G_\ell} \sum\nolimits_{t\in B_j} D_{it}
                \end{align*}
                Hence, $L_{j\ell}\leq \sum\nolimits_{i\in G_\ell} D^{(\cM)}_{ij} \leq U_{j\ell}$ if and only if $L_{j\ell}\leq \sum\nolimits_{i\in G_\ell} \sum\nolimits_{t\in B_j} D_{it}$ $\leq U_{j\ell}$.
                The result follows as this holds for all $1\leq i\leq m$ and $1\leq j\leq q$.
            \end{proof}

            \begin{lemma}\label{lem:equivalence_of_if}
                For any matrix $A,C\in [0,1]^{m\times q}$ and matrix $D\in [0,1]^{m\times n}$ encoding the marginals of some distribution over rankings,
                the matrix $D^{(\cM)}\coloneqq g(D)$ satisfies $(C,A)$-individual fairness constraints if and only if $D$ satisfies \cref{eq:prog_from_ashudeep_eq_1}.
            \end{lemma}
            \begin{proof}
                Fix any $1\leq i\leq m$ and $1\leq j\leq q$.
                \begin{align*}
                    D^{(\cM)}_{ij}
                    \ \ \Stackrel{\eqref{def:mu}}{=} \ \
                    \sum\nolimits_{t\in B_j} D_{it}
                \end{align*}
                Hence, $\sum\nolimits_{t\in B_j} D_{it}\geq C_{ij}$ if and only if $D^{(\cM)}_{ij} \geq C_{ij}$ and $\sum\nolimits_{t\in B_j} D_{it}\leq A_{ij}$ if and only if $D^{(\cM)}_{ij} \leq A_{ij}$.
                The result follows as this holds for all $1\leq i\leq m$ and $1\leq j\leq q$.
            \end{proof}

        \subsubsection{Utility Guarantee of Algorithm~\ref{alg:main}}\label{sec:utility_guarantee}
            Let $\cD^\star$ be  an  optimal solution of \cref{prob:main}.
            Let $D^\star\in [0,1]^{m\times n}$ be the corresponding matrix of marginals.
            Since $\cD^\star$ is $(C,A)$-individually fair and is supported over the set of $(L,U)$-group fair rankings, it follows that $D^\star$ is feasible for \prog{prog:mod_of_ashudeep}.
            Hence, $$\rho^\top (D^\star) v\leq \rho^\top \hD v$$ where $\hD$ is  an  optimal solution of \prog{prog:mod_of_ashudeep}.
            Since $\rho^\top (D^\star) v$ is the expected utility of $\cD^\star$, the expected utility of $\cD^\star$ is at most $\rho^\top \hD v$.
            \begin{lemma}\label{lem:utility_guarantee_1}
                The expected utility of $\cD^\star$ is at most $\rho^\top \hD v$, where  $\cD^\star$ is  an  optimal solution of \cref{prob:main} and $\hD$ is  an  optimal solution of \prog{prog:mod_of_ashudeep}.
            \end{lemma}
            \noindent Thus, it suffices to show that $\rho^\top Dv\geq \alpha (\rho^\top \hD v)$, where $D$ is the marginal of the distribution $\cD$ in \cref{thm:algo_main}.

            \begin{lemma}\label{lem:utility_guarantee_2}
                If $v\in \R^n$ is such that, {for all $\Delta\geq 0$, $\frac{v_{s+\Delta}}{v_s}$ is a non-decreasing function of $1\leq s \leq n-\Delta$,} then
                \begin{align*}
                    \rho^\top D v \geq \inparen{\frac{\sum_{s\in B_1} v_s}{\abs{B_1} \cdot v_{1}}} \cdot \rho^\top \hD v.
                \end{align*}
            \end{lemma}
            \begin{proof} The proof is divided into two steps, in the first step we show that $D$ and $\hD$ have identical marginals.
            In the second step, we use this to lower-bound the ratio of $D$'s expected utility to $\hD$ expected utility.

                \paragraph{Step 1 (Block-wise sums of $D$ and $\hD$ are identical):}
                    By design in \cref{alg:main}, $\Pr_{R'\sim \cD}[R'=R_t]=\alpha_t$ for each $t$.
                    Since $D$ is the marginal of $\cD$, it follows that $D = \sum_t \alpha_t R_t$.
                    Further, since $g(R_t)=M_t$, $g$ is linear, and $\hM=\sum_t \alpha_t M_t$, it follows that $g(D)=\hM$.
                    By construction in {Step 2} of \cref{alg:main}, $g(\hD)=\hM$.
                    Hence, $g(\cD)=g(\hD)$.
                    Consequently, by the definition of $g$ (\cref{def:mu}), for each $1\leq i\leq m$ and $1\leq j\leq q$, it holds that
                    \begin{align*}
                        \sum\nolimits_{t\in B_j} D_{it}=\sum\nolimits_{t\in B_j} \hD_{it}.
                        \yesnum\label{eq:equality_of_block_wise_sums}
                    \end{align*}

                \paragraph{Step 2 (Completing the proof):}
                    To complete the proof, we will upper bound $\rho^\top \hD v$ and lower bound $\rho^\top D v$.
                    For each $1\leq j\leq q$, let $s(j)$ be the smallest position in $B_j$.
                    In other words, $B_j=\inbrace{s(j), s(j)+1,\dots,s(j+1)-1}$.
                    \begin{align*}
                        \rho^\top \hD v
                        &= \sum\nolimits_{j=1}^q\sum\nolimits_{i=1}^m \sum\nolimits_{s\in B_j} \rho_i v_s \hD_{is} \\
                        &\leq \sum\nolimits_{j=1}^q v_{s(j)}\sum\nolimits_{i=1}^m \rho_i\sum\nolimits_{s\in B_j}\hD_{is}.\yesnum\label{eq:ub_util_guarantee}
                    \end{align*}
                    \noindent Fix any $1\leq j\leq q$ and $t$.
                    Suppose items $i_1, i_2, \dots, i_{\abs{B_j}}$ appear in block $B_j$ in $R_t$ in that order.
                    Since $R_t=f(M_t)$ and (by definition) items each block $B$ of $f(M)$ appear in non-increasing order of utility, it follows that
                    \begin{align*}
                      \rho_{i_1}\geq \rho_{i_2}\geq \dots \geq \rho_{i_{\abs{B_j}}}.  \yesnum\label{eq:increasing_rho}
                    \end{align*}
                    It holds that
                    \begin{align*}
                        \sum\nolimits_{i=1}^m \sum\nolimits_{s\in B_j} \rho_i (R_t)_{is} v_s
                        &= \sum\nolimits_{s\in B_j} \rho_{i_s} v_s\\
                        &\geq \frac{1}{\abs{B_j}} \inparen{\sum\nolimits_{s\in B_j} \rho_{i_s}}\inparen{\sum\nolimits_{s\in B_j} v_{s}}\\
                        &= \frac{1}{\abs{B_j}} \inparen{\sum\nolimits_{i=1}^m\sum\nolimits_{s\in B_j} \rho_i (R_t)_{is} }\inparen{\sum\nolimits_{s\in B_j} v_{s}}
                        \yesnum\label{eq:lb_on_util_of_one_rank}
                    \end{align*}
                    Where the inequality holds because $v_1\geq v_2\geq \dots \geq v_n$, \cref{eq:increasing_rho}, and the following standard inequality (which is also known as Chebyshev's order inequality): for all $z\geq 1$ and $x,y\in \R^z$, if $x_1\geq x_2\geq \dots\geq x_z$ and $y_1\geq y_2\geq \dots\geq y_z$, then
                    \[
                        z \sum_{i=1}^z x_i y_i - \inparen{\sum_{i=1}^z x_i}\inparen{\sum_{i=1}^z y_i}=\sum_{i,j\in [z]\colon i<j}\inparen{x_i-x_j}\inparen{y_i-y_j}\geq 0.
                    \]
                    The following lower bound holds for $\rho^\top D v$
                    \begin{align*}
                        \rho^\top D v\
                        &=\ \sum\nolimits_{i=1}^m\sum\nolimits_{j=1}^q\sum\nolimits_{s\in B_j}\sum\nolimits_{t} \rho_i v_s  (R_t)_{is}\\
                        &\Stackrel{\eqref{eq:lb_on_util_of_one_rank}}{\geq}\ \sum\nolimits_{j=1}^q \sum\nolimits_{t} \frac{\inparen{\sum\nolimits_{i=1}^m\sum\nolimits_{s\in B_j} \rho_i (R_t)_{is} }\inparen{\sum\nolimits_{s\in B_j} v_{s}}}{\abs{B_j}}\\
                        &\Stackrel{}{=}\ \sum\nolimits_{j=1}^q \frac{\inparen{\sum\nolimits_{i=1}^m\sum\nolimits_{s\in B_j} \rho_i D_{is} }\inparen{\sum\nolimits_{s\in B_j} v_{s}}}{\abs{B_j}}\\
                        &\Stackrel{}{=}\ \sum\nolimits_{j=1}^q \frac{\sum\nolimits_{s\in B_j} v_{s}}{\abs{B_j}} \sum\nolimits_{i=1}^m\rho_i \sum\nolimits_{s\in B_j}  D_{is}\\
                        &\Stackrel{\eqref{eq:equality_of_block_wise_sums}}{=}\ \sum\nolimits_{j=1}^q \frac{\sum\nolimits_{s\in B_j} v_{s}}{\abs{B_j}} \sum\nolimits_{i=1}^m\rho_i \sum\nolimits_{s\in B_j}  \hD_{is}.
                            \yesnum\label{eq:lb_util_guarantee}
                    \end{align*}
                    Taking the ratio of \cref{eq:ub_util_guarantee,eq:lb_util_guarantee}, it follows
                    \begin{align*}
                        \frac{\rho^\top D v}{\rho^\top \hD v}
                        &\geq \min_{1\leq j\leq q} \frac{\sum_{s\in B_j}  v_s}{\abs{B_j} \cdot v_{s(j)}}\\
                        &= \min_{1\leq j\leq q} \frac{\sum_{s\in B_j}  v_s}{k\cdot v_{s(j)}}
                        \tag{Using that each block $B_1,B_2,\dots,B_q$ has size $k$}\\
                        &\geq \frac{\sum_{s\in B_1}  v_s}{k\cdot v_{s(j)}}\\
                        &= \frac{v_1+v_2+\dots+v_k}{k\cdot v_1}.
                        \yesnum\label{eq:verifying_upperbound}
                    \end{align*}
                    Where the last inequality holds due to (1) {the assumption that, for all $\Delta\geq 0$, $\frac{v_{s+\Delta}}{v_s}$ is a non-decreasing function of $s$} and (2) that $B_1$, $B_2,\cdots,B_q$ are blocks of positions such that all positions in $B_j$ appear before all positions in $B_{j+1}$ (for each $1\leq j < q$). %
            \end{proof}

        \subsubsection{Efficiency Guarantee of Algorithm~\ref{alg:main}}\label{sec:efficiency_guarantee}
            In this section, we prove \cref{lem:decomposition}

            \begin{lemma}\label{lem:decomposition}
                There is a polynomial time algorithm, that given any $m\times q$ matrix $M\in [0,1]^{m\times q}$ satisfying the following
                \begin{itemize}
                    \item \cref{eq:group_constraints_matchings},
                    \item for each $i\in [m]$, $\sum_{j} M_{ij}\leq 1$,
                    \item for each $j\in [q]$, $\sum_{i} M_{ij} = \abs{B_j}$,
                \end{itemize}
                outputs numbers $\alpha_1,\alpha_2,\dots,\alpha_{O((m+p)n)}\geq 0$ and corresponding matchings $M_1,M_2,\dots,M_{O((m+p)n)}\in \cM$ such that
                \begin{align*}
                    M &= \sum\nolimits_{t} \alpha_t M_t \quad\text{and}\quad \sum\nolimits_t\alpha_t = 1,
                    \yesnum\label{eq:decomposition}
                \end{align*}
                and each $M_t$ satisfies $(L,U)$-group fairness constraints.
            \end{lemma}
            \noindent Thus, \cref{lem:decomposition} shows that Step 4 of \cref{alg:main} can be implemented in polynomial time if $\hM$ satisfies the requirements in \cref{lem:decomposition}.
            In \cref{sec:complete_proof_algo_main}, we show that $\hM$ does satisfy the requirements in \cref{lem:decomposition}.

            The key property used in the proof of \cref{lem:decomposition} is that all vertices of the following polytope are integral, i.e., the polytope is integral.
            \begin{align*}
                \mathscr{M}_{\rm GF}\coloneqq \inbrace{M\in [0,1]^{m\times q} \ \bigg| \begin{array}{cc}
                     \text{$M$ satisfies \cref{eq:group_constraints_matchings},}\\
                     \forall i\in [m],\quad \sum_{j} M_{ij} = 1,\\
                     \white{.}\ \ \forall j\in [q],\quad\  \sum_{i} M_{ij} = \abs{B_j}
                \end{array}
                }
            \end{align*}
            (Note that the constraint $\sum_{j} M_{ij} = 1$ has an equality and not an inequality like in the definition of $\cM$.)

            \paragraph{Each vertex of $\mathscr{M}_{\rm GF}$ is a $(L,U)$-group fair matching.}
            Integrality of $\mathscr{M}_{\rm GF}$ implies that each vertex of $\mathscr{M}_{\rm GF}$ is a matching satisfying $(L,U)$-group fairness constraints.
            To see this, fix any vertex $V$ of $\mathscr{M}_{\rm GF}$.
            Since $\mathscr{M}_{\rm GF}$ is integral, $V\in \inbrace{0,1}^{m\times q}$.
            This, combined with the fact that $V$ satisfies the constraints
            \begin{align*}
                \forall i\in [m],\quad \sum\nolimits_{j} M_{ij} = 1
                    \quad\text{and}\quad
                \forall j\in [q],\quad\  \sum\nolimits_{i} M_{ij} = \abs{B_j}
            \end{align*}
            implies that $V\in \cM$.
            Since $V$ is in $\cM$ and $V$ satisfies \cref{eq:group_constraints_matchings}, it follows that $V$ is a matching that satisfies $(L,U)$-group fairness constraints.
            Thus, each vertex of $\mathscr{M}_{\rm GF}$ is a $(L,U)$-group fair matching.

            \paragraph{$(M_t,\alpha_t)_t$ can be efficiently computed.}
                The constraints on $M$ in \cref{lem:decomposition} ensure that $M\in \mathscr{M}_{\rm GF}$.
                Hence, Carathéodory's theorem \cite{grotschel_lovasz_schrijver} guarantees that $M$ can be expressed as in \cref{eq:decomposition} where each of $M_1,M_2,\dots$ is a vertex of $\mathscr{M}_{\rm GF}$.
                Hence, by our previous argument, each of $M_1, M_2, \dots$ is a matching satisfying $(L, U)$-group fairness constraint.
                $M_1, M_2, \dots$ and $\alpha_1,\alpha_2,\dots$ can be computed efficiently using, e.g., the algorithm in \cite[Theorem 6.5.11]{grotschel_lovasz_schrijver}.

            \paragraph{$\mathscr{M}_{\rm GF}$ is integral.}
                It remains to prove the claim that $\mathscr{M}_{\rm GF}$ is integral.
                To prove this, we use \cite[Lemma 2.4]{pln_2022} (also implicit in \citet{meghana2019classified}).
                We use $\circ$ in the superscript to denote the variables in \cite{pln_2022}.
                \begin{lemma}[\protect{Rephrasing \cite[Lemma 2.4]{pln_2022}}]
                    For any $m,p,q\geq 1$, any disjoint groups $G^\circ_1,G^\circ_2,\dots,G^\circ_p\subseteq [m]$, and any integer assignments of the following values
                    \begin{itemize}
                        \item $L^\circ_{i S}$ and $U^\circ_{iS}$ for each $1\leq i\leq m$ and $S\subseteq [q]$,
                        \item $L^\circ_{j}$ and $U^\circ_{j}$ for each $1\leq j\leq q$,
                        \item $L^\circ_{j\ell}$ and $U^\circ_{j\ell}$ for each $1\leq j\leq q$ and $1\leq \ell\leq p$,
                    \end{itemize}
                    the following polytope is integral.
                    \begin{align*}
                        \cP \coloneqq \inbrace{P\in [0,1]^{m\times q} \ \mid \begin{array}{cc}
                             \forall i\in [m],\forall S\subseteq [q], & L^\circ_{iS}\leq \sum_{j\in S} P_{ij}\leq U^\circ_{iS},\\
                             \forall j\in [q], & L^\circ_{j}\leq \sum_{i\in [m]} P_{ij}\leq U^\circ_{j},\\
                             \forall j\in [q],\forall \ell\subseteq [p], & L^\circ_{j\ell}\leq \sum_{i\in G^\circ_\ell} P_{ij}\leq U^\circ_{j\ell}
                        \end{array}
                        }.
                    \end{align*}
                \end{lemma}
                \noindent The integrality of $\mathscr{M}_{\rm GF}$ follows as $P$ reduces to $\mathscr{M}_{\rm GF}$ for the following setting of parameters.
                \begin{itemize}
                    \item For each $1\leq \ell\leq p$, $G^\circ_\ell=G_\ell$,
                    \item For each $1\leq i\leq m$ and $S\subseteq [q]$, $L^\circ_{i S}=0$,
                    \item For each $1\leq i\leq m$ and $S\subseteq [q]$, where $S\neq [q]$, $U^\circ_{iS}=m$,
                    \item For each $1\leq i\leq m$, $L^\circ_{i[q]}=1$ and $U^\circ_{i[q]}=1$,
                    \item For each $1\leq j\leq q$, $L^\circ_{j}=U^\circ_{j}=\abs{B_j}$,
                    \item For each $1\leq j\leq q$ and $1\leq \ell\leq p$, $L^\circ_{j\ell}=L_{j\ell}$ and $U^\circ_{j\ell}=U_{j\ell}$.
                \end{itemize}
                To be precise, the equivalence follows after dropping the vacuous constraints.
                The vacuous constraints correspond to the parameters: $L^\circ_{iS}$ (for all $i$ and $S\neq [q]$.

                \paragraph{Additional discussion.}
                    There are several other ways of proving the integrality of $\mathscr{M}_{\rm GF}$.
                    For instance, one method is to modify the standard proof that the perfect matching polytope is integral (see \cite[Theorem 18.1]{schrijver2003combinatorial}) to account for the group fairness constraints.

        \subsubsection{Completing the Proof of Theorem~\ref{thm:algo_main}}\label{sec:complete_proof_algo_main}
            \begin{proof}[Proof of \cref{thm:algo_main}]
                Since $\hD$ is feasible for \prog{prog:mod_of_ashudeep}, it satisfies \cref{eq:prog_from_ashudeep_eq_1,eq:prog_from_ashudeep_eq_2}.
                Hence, \cref{lem:equivalence_of_gf,lem:equivalence_of_if} imply that $\hM$ satisfies $(L,U)$-group fairness constraints and $C$-individual fairness constraints.
                The former shows that \cref{lem:decomposition} is applicable for $\hM$.

                \paragraph{Each output ranking is group fair.}
                    \cref{lem:decomposition} implies that for each $t$, $M_t\in \cM$ computed in {Step 4} of \cref{alg:main} satisfies $(L,U)$-group fairness constraints.
                    Since $R_t\coloneqq f(M_t)$, by definition of $f$, $g(R_t)=M_t$ and, hence, because \cref{lem:equivalence_of_gf} and the fact that $M_t$ satisfies $(L,U)$-group fairness constraints, it follows that $R_t$ satisfies $(L,U)$-group fairness constraints (for each $t$).
                    This proves that each ranking output by \cref{alg:main} satisfies $(L, U)$-group fairness constraints.

                \paragraph{Output rankings are individually fair.}
                    Let $\cD$ be the distribution mentioned in \cref{thm:algo_main}.
                    By design in \cref{alg:main}, $\Pr_{R'\sim \cD}[R'=R_t]=\alpha_t$ for each $t$.
                    Let $D$ be the marginal of $\cD$.
                    It follows that $D = \sum_t \alpha_t R_t$.
                    Since $g(R_t)=M_t$, $g$ is linear, and $\hM=\sum_t \alpha_t M_t$, it follows that $g(D)=\hM$.
                    Hence, from \cref{lem:equivalence_of_if} and the earlier observation that $\hM$ is $(C,A)$-individually fair, it follows that $D$ is $(C,A)$-individually fair.
                    As $D$ is a marginal of $\cD$, it also follows that $\cD$ is $(C,A)$-individually fair.

                \paragraph{Expected utility of output rankings is $\alpha$-approximation of the optimal.}
                    Follows from \cref{lem:utility_guarantee_1,lem:utility_guarantee_2}.

                \paragraph{\cref{alg:main} terminates in polynomial time.}
                    Step 1 of \cref{alg:main} can be finished in polynomial time using any standard polynomial time linear programming solver.
                    Step 2 only requires computing $g(\hD)$, which can be done in $O(mn)$ time.
                    Step 3 takes $O(1)$ time.
                    \cref{lem:decomposition} proves that Step 4 can be implemented in polynomial time.
                    Step 5 has $T=O(m^2n^2)$ substeps, each of which requires computing $f(M)$ for some $M\in\cM$. $f(M)$ can be computed in $O(mn\log{n})$ time. Hence, Step 5 takes $O(m^3n^3\log{n})$ time.
            \end{proof}

    \subsubsection{Proof of \cref{eq:strongerInequality}}\label{sec:verify}
        In this section, we prove \cref{eq:strongerInequality}.
        Namely, let $\Delta$ be defined as follows
        \[
            \Delta \coloneqq \frac{\max_{i} \rho_i}{\min_{i} \rho_i} - 1.
            \yesnum\label{def:delta_app}
        \]
        We prove that the ranking $R$ output by \cref{alg:main} has an expected utility at least $\alpha$ times the optimal where $\alpha$ satisfies the following lower bound
        \begin{align*}
            \alpha \geq
            \frac{
                1 + \frac{v_k \Delta }{v_1+v_2+\dots+v_k}
            }{
                1 + \frac{  v_1\Delta}{  v_1+v_2+\dots+v_k}
            }
            \yesnum\label{eq:toProve}
        \end{align*}
        Recall that \cref{thm:algo_main} proves the following lower bound on $\alpha$.
        \begin{align*}
            \alpha \geq \frac{v_1 + v_2 + \dots + v_k}{kv_1}.
        \end{align*}
        To prove \cref{eq:toProve}, we modify the proof of  \cref{thm:algo_main}.
        Specifically, we prove the following variant of \cref{lem:utility_guarantee_2}.
        \begin{lemma}\label{lem:utility_guarantee_2_variant}
            If $v\in \R^n$ is such that, {for all $\Delta\geq 0$, $\frac{v_{s+\Delta}}{v_s}$ is a non-decreasing function of $1\leq s \leq n-\Delta$,} then
            \begin{align*}
                \rho^\top D v \geq
                \frac{
                    1 + \frac{v_k \Delta }{v_1+v_2+\dots+v_k}
                }{
                    1 + \frac{  v_1\Delta}{  v_1+v_2+\dots+v_k}
                }
                \cdot \rho^\top \hD v.
            \end{align*}
        \end{lemma}
        \noindent This variant combined with \cref{lem:utility_guarantee_1} implies \cref{eq:toProve}.
        \begin{proof}[Proof of \cref{lem:utility_guarantee_2_variant}]
            Like the proof of \cref{lem:utility_guarantee_2}, the proof of \cref{lem:utility_guarantee_2_variant} is divided into two steps, in the first step we show that $D$ and $\hD$ have identical marginals.
            In the second step, we use this to lower-bound the ratio of $D$'s expected utility to $\hD$ expected utility.

            \paragraph{Step 1 (``Marginals'' of $D$ and $\hD$ are identical):}
            The first step is identical to the first step in the proof of \cref{lem:utility_guarantee_2}.
            It concludes with the following equalities: for each $1\leq i\leq m$ and $1\leq j\leq q$, it holds that
            \begin{align*}
                \sum\nolimits_{t\in B_j} D_{it}=\sum\nolimits_{t\in B_j} \hD_{it}.
                \yesnum\label{eq:equality_of_block_wise_sums_2}
           \end{align*}

                \paragraph{Step 2 (Completing the proof):}
                    Without loss of generality assume that
                    \[
                        \rho_1\geq \rho_2\geq \cdots \geq \rho_m.
                    \]
                    For each $1\leq j\leq q$, let $s(j)$ be the smallest position in $B_j$, i.e., $B_j=\inbrace{s(j), s(j)+1,\dots, s(j+1)-1}$.
                    \begin{align*}
                        \frac{\rho^\top D v}{\rho^\top \hD v}\
                            &= \frac{
                                 \rho_m 1^\top D v
                                 +
                                 \inparen{\rho - \rho_m 1}^\top D v
                            }{
                                \rho_m 1^\top \hD v
                                 +
                                 \inparen{\rho - \rho_m 1}^\top \hD v
                            }\\
                            &= \frac{
                                 \rho_m  1^\top v
                                 +
                                 \inparen{\rho - \rho_m 1}^\top D v
                            }{
                                \rho_m  1^\top v
                                 +
                                 \inparen{\rho - \rho_m 1}^\top \hD v
                            }
                            \tag{Using that for all $s$, $\sum\nolimits_{i=1}^m D_{is}=\sum\nolimits_{i=1}^m \hD_{is}=1$}\\
                            &= \frac{
                                 1^\top v
                                 +
                                 \sum\nolimits_{i=1}^m
                                    \inparen{\frac{\rho_i}{\rho_m} - 1}
                                        \sum\nolimits_{s=1}^n  v_s  D_{is}
                            }{
                                1^\top v
                                 +
                                 \sum\nolimits_{i=1}^m
                                    \inparen{\frac{\rho_i}{\rho_m} - 1}
                                        \sum\nolimits_{s=1}^n  v_s  \hD_{is}
                            }\\
                            &\geq \frac{
                                 1^\top v
                                 +
                                 \sum\nolimits_{i=1}^m
                                    \inparen{\frac{\rho_i}{\rho_m} - 1} \cdot {
                                        \sum\nolimits_{s=1}^n  v_s  D_{is}
                                    }
                            }{
                                1^\top v
                                 +
                                 \sum\nolimits_{i=1}^m
                                    \inparen{\frac{\rho_i}{\rho_m} - 1} \cdot {
                                        \max\inbrace{
                                            \sum\nolimits_{s=1}^n  v_s  D_{is},
                                            \sum\nolimits_{s=1}^n  v_s  \hD_{is}
                                        }
                                    }
                            }\\
                            &\geq \frac{
                                 1^\top v
                                 +
                                 \Delta\cdot \sum\nolimits_{s=1}^n  v_s  \sum\nolimits_{i=1}^m  D_{is}
                                }
                            {
                                1^\top v
                                 +
                                 \Delta\cdot
                                    {
                                        \sum\nolimits_{i=1}^m   \max\inbrace{
                                            \sum\nolimits_{s=1}^n  v_s  D_{is},
                                            \sum\nolimits_{s=1}^n  v_s  \hD_{is}
                                        }
                                    }
                            }
                            \tag{Using that for all $b > a\geq 0$, $\frac{1+ax}{1+bx}$ is a decreasing function of $x$}\\
                            &\geq \frac{
                                 1^\top v
                                 +
                                 \Delta\cdot 1^\top v
                                }
                            {
                                1^\top v
                                 +
                                 \Delta\cdot
                                    {
                                        \sum\nolimits_{i=1}^m   \max\inbrace{
                                            \sum\nolimits_{j=1}^q  v_{s(j)} \sum_{s\in B_j} D_{is},
                                            \sum\nolimits_{j=1}^q  v_{s(j)} \sum_{s\in B_j}  \hD_{is}
                                        }
                                    }
                            }
                            \tag{Using that $v_1\geq v_2\geq \dots\geq v_n$ and $s(j)$ is the smallest value in $B_{j}$ and $\sum_{i=1}^m D_{is}=1$ for all $s$}
                    \end{align*}
                    \begin{align*}
                            &\geq \frac{
                                 1^\top v
                                 +
                                 \Delta\cdot 1^\top v
                                }
                            {
                                1^\top v
                                 +
                                 \Delta\cdot
                                    {
                                        \sum\nolimits_{i=1}^m \max\inbrace{
                                            \sum\nolimits_{j=1}^q  v_{s(j)}\Pr_D\insquare{i\in B_j}  ,
                                            \sum\nolimits_{j=1}^q  v_{s(j)}\Pr_{\hD}\insquare{i\in B_j}
                                        }
                                    }
                            }
                            \tag{Using that for all $j$, $\sum_{s\in B_j} D_{is}=\Pr_D\insquare{i\in B_j}$ and $\sum_{s\in B_j} \hD_{is}=\Pr_{\hD}\insquare{i\in B_j}$}\\
                            &\geq \frac{
                                 1^\top v\inparen{1+\Delta}
                                }
                            {
                                1^\top v
                                 +
                                 \Delta\cdot
                                        \sum\nolimits_{j=1}^q \abs{B_j} v_{s(j)}
                            }
                            \tag{Using that for all $j$, $\sum_{i=1}^m\Pr_{D}\insquare{i\in B_j}=\sum_{i=1}^m\Pr_{\hD}\insquare{i\in B_j}=\abs{B_j}$}\\
                            &= \frac{
                                 1 + \Delta
                                }
                            {
                                1 + \Delta \cdot
                                        \frac{
                                            \sum\nolimits_{j=1}^q  \abs{B_j} v_{s(j)}
                                        }{
                                            1^\top v
                                        }
                            }\\
                            &\geq \frac{
                                 1 + \Delta
                                }
                            {
                                1 + \Delta \cdot
                                        \max_{1\leq j\leq q}
                                        \frac{
                                            \abs{B_j} v_{s(j)}
                                        }{
                                            \sum_{s\in B_j} v_s
                                        }
                            }\\
                            &\geq \frac{
                                 1 + \Delta
                                }
                            {
                                1 + \frac{
                                           \Delta  k v_{1}
                                        }{
                                            v_1+v_2+\dots+v_k
                                        }
                            }.
                            \yesnum\label{eq:verifying_upperbound_new}
                    \end{align*}
                    Where the last inequality holds due to (1) {the assumption that, for all $\Delta\geq 0$, $\frac{v_{s+\Delta}}{v_s}$ is a non-decreasing function of $s$} and (2) that $B_1$, $B_2,\cdots,B_q$ are blocks of positions of size $k$, and (2) all positions in $B_j$ appear before all positions in $B_{j+1}$ (for each $1\leq j < q$). %
        \end{proof}

    \subsubsection{Proof That \cref{eq:strongerInequality} Is Strictly Larger Than the Bound in \cref{thm:algo_main}}\label{sec:verify2}
        In this section, we prove that for any $\Delta\geq 0$, \cref{eq:strongerInequality} is strictly larger than the RHS in \cref{eq:lowerbound_on_alpha}, i.e., for any $\Delta\geq 0$
        \[
            \inparen{ 1 + \Delta }\inparen{  1+\frac{ kv_1\Delta}{  v_1+v_2+\dots+v_k}}^{-1} \geq \frac{  v_1+v_2+\dots+v_k}{kv_1}
        \]
        This holds because of the following
        \begin{align*}
            \inparen{ 1+ \Delta }\inparen{  1+\frac{ kv_1\Delta}{  \sum_{i=1}^k v_i}}^{-1}
            &= \inparen{\sum_{i=1}^k v_i }\cdot \frac{1 + \Delta}{\sum_{i=1}^k v_i + k v_1\Delta }\\
            &= \frac{\sum_{i=1}^k v_i }{v_1}\cdot \frac{1 + \Delta}{\sum_{i=1}^k \frac{v_i}{v_1}+k\Delta}\\
            &\geq \frac{\sum_{i=1}^k v_i }{v_1}\cdot \frac{1+\Delta}{k\inparen{1+\Delta}}\tag{Using that $v_1\geq v_2\geq \dots\geq v_n$}\\
            &= \frac{\sum_{i=1}^k v_i }{k v_1}.\yesnum\label{eq:bound_tmp}
        \end{align*}

    \begin{remark}
        {The following example shows that the bound on utility in \cref{eq:lowerbound_on_alpha} is tight.
        Suppose $n=m=4$, $\rho=(1,1,0,0)^\top$, $B_1=\inbrace{1,2}$, and $B_{2}=\inbrace{3,4}$.
        Suppose $v=(1,0,0,0)^\top$ and, hence, the RHS in \cref{eq:lowerbound_on_alpha} is $\frac{1}{2}$.
        Suppose the marginal $D^\star$ of $\cD^\star$ is as follows
        \[D^\star = \begin{bmatrix}
            0.5 & 0.5 & 0 & 0\\
            0 & 0 & 0.5 & 0.5\\
            0.5 & 0.5 & 0 & 0\\
            0 & 0 & 0.5 & 0.5
        \end{bmatrix}.\]
        In this example, \cref{alg:main} can output a ranking sampled from the uniform distribution over $R_1$ and $R_2$ where
        \[
            R_1 = \begin{bmatrix}
            1 & 0 & 0 & 0\\
            0 & 1 & 0 & 0\\
            0 & 0 & 1 & 0\\
            0 & 0 & 0 & 1
            \end{bmatrix}
            \quad\text{and}\quad
            R_2 = \begin{bmatrix}
            0 & 0 & 1 & 0\\
            0 & 0 & 0 & 1\\
            1 & 0 & 0 & 0\\
            0 & 1 & 0 & 0
            \end{bmatrix}.
        \]
        Here, the optimal utility is $1$ (as $\rho^\top D^\star v = 1$) and the expected utility achieved by \cref{alg:main} is $\frac{1}{2}$ (as  $\frac{1}{2}\rho^\top (R_1+R_2) v=\frac{1}{2}$).
        Hence, \cref{alg:main} achieves $\inparen{\frac{\sum_{i=1}^k v_i }{v_1}\cdot \frac{1+\Delta}{k\inparen{1+\Delta}}}$-fraction of the optimal utility.}
    \end{remark}

    \subsection{Proof of \cref{thm:approxStochasUtil}}\label{sec:proofof:thm:approxStochasUtil}
        In this section, we prove \cref{thm:approxStochasUtil}.
        To ease readability, we restate \cref{thm:approxStochasUtil} below.

        \thmApproxStochasUtil*

        \noindent The proof of \cref{thm:approxStochasUtil} relies on the following standard lemmas.
        \begin{lemma}\label{lem:unifConc}
            If $m$ values $\inbrace{x_1,x_2,\dots,x_m}$ are drawn i.i.d. from the uniform distribution on $[0,S]$ and sorted to get $x_{(1)}\geq x_{(2)}\geq \dots \geq x_{(m)}$, then
            with probability at least $1-4m^{-\frac14}$ the following holds:
            \[
                \text{for all $1\leq i\leq m$,}\quad
                \abs{
                    x_{(i)} - S\inparen{ 1 - \frac{i}{m}}
                }
                \leq \frac{2S}{m^{1/4}}.
                \yesnum\label{def:eventETmp}
            \]
        \end{lemma}

        \newcommand{\errBnd}{\ensuremath{\sigma_{\max}\sqrt{2\log\frac{qm^2}{\delta}}}}

        \begin{lemma}[\protect{\textbf{\cite[Lemma A.1]{charikar2006nearoptimal} and union bound}}]\label{lem:gauss}
            If, for each $1\leq i\leq m$, $\rho_i$ is i.i.d. from $\mathcal{N}(\mu_i,\sigma_i^2)$, then with probability at least $1-\frac{\delta}{mq}$ (for any $\delta>0$), it holds that
            \[
                \text{for all $1\leq i\leq m$,}\quad
                \abs{\rho_i-\mu_i}\leq \errBnd{},
                \yesnum\label{def:eventF}
            \]
            where $\sigma_{\max}\coloneqq \inbrace{\sigma_1,\sigma_2,\dots,\sigma_m}$.
        \end{lemma}
        \noindent We begin by defining three high-probability events which will be used in the remainder of the proof.

        \paragraph{High probability events ($\evE$, $\evF$, and $\evG$).}
        Recall $\mu_1,\mu_2,\dots,\mu_m$ are drawn from the uniform distribution on $[0,S]$.
        Without loss of generality rearrange the items to ensure that
        \[
            \mu_1\geq \mu_2\geq \dots\geq \mu_m.
        \]
        Let $\evE$  be the following event:
        \[
            \text{for all $1\leq i\leq m$,}\quad
                \abs{
                    \mu_i - S\inparen{ 1 - \frac{i}{m}}
                }
            \leq \frac{2S}{m^{1/4}}.
            \yesnum\label{def:eventE}
        \]
        \cref{lem:unifConc} implies that
        \[
            \Pr\insquare{\evE} \geq  1 - 4m^{-\frac{1}{4}}.
        \]
        Let $\evF$ be the event \cref{def:eventF} holds.
        \cref{lem:gauss} implies that $\Pr\insquare{\evF}\geq 1-\frac{\delta}{mq}$.
        Let $P_{\rm small}\subseteq [m]\times [q]$ be the set of item-block pairs such that the corresponding individual fairness constraint is at most $\frac{\delta}{mq}$, i.e.,
        \[
            P_{\rm small}
                \coloneqq
                \inbrace{
                    (i,t)\in [m]\times [q] \mid
                    C_{it}\leq \frac{\delta}{mq}
                }.
        \]
        Let $\cD_{\rm alg}$ be the distribution that  \cref{alg:main} samples a ranking from.
        \cref{thm:algo_main} implies that $\cD_{\rm alg}$ satisfies $C$-individual fairness constraints.
        Let $\cD_{\rm oth}$ be any other distribution over rankings that satisfies $C$-individual fairness constraints.
        Let $\evG$ be the event that there exists no pairs $(i,t)\in P_{\rm small}$ such that item $i$ appears in $B_t$ in rankings $R_1\sim \cD_{\rm alg}$ and $R_{2}\sim \cD_{\rm oth}$. %
        Since for each $t\in [q]$, $C_{1t}+C_{2t}+\dots+C_{mt}=1$ any distribution $\cD$ that satisfies the $C$-individual fairness constraints must place item $i$ in $B_t$ with probability \textit{exactly} $C_{it}$.
        Hence, for any element $(i,t)\in P_{\rm small}$, the probability that  item $i$ appears in $B_t$ in the ranking output by \cref{alg:main} is at most $\frac{\delta}{mq}$.
        Since $\abs{P_{\rm small}}\leq mq$, the union bound implies that
        \[
            \Pr\insquare{\evG} \geq 1 - 2\delta.
        \]
        Finally, by another union bound and earlier lower bounds on $\Pr\insquare{\evE}$ and $\Pr\insquare{\evF}$, it follows that
        \[
            \Pr\insquare{\evE\text{, } \evF\text{, and } \evG} \geq 1-4m^{-\frac{1}{4}}-3\delta.
            \yesnum\label{eq:unionBoundWHP}
        \]

        \paragraph{Step 1 (Analysing utility for $(i,t)\not\in P_{\rm small}$):}
            Since event $\evG$ occurs with high probability, we would be able to restrict our attention to pairs $(i,t)\not\in P_{\rm small}$.
            In this step, we bound $\rho_i$ as a function of $t$ for any  $(i,t)\not\in P_{\rm small}$.
            Fix any  $(i,t)\not\in P_{\rm small}$.
            Let $R(\rho)$ be the ranking that sorts items in decreasing order by $\rho$.
            By definition of $P_{\rm small}$, $C_{it}>\frac{\delta}{mq}$.
            Hence, by the construction of $C$ (\cref{eq:example_indv_const}), it follows that
            \begin{align*}
                \frac{\delta}{mq}
                &< C_{it}\\
                &= \Pr_{}\insquare{i\in B_t \text{ in the ranking } R(\rho)}\\
                &=
                \Pr_{}\insquare{i\in B_t \text{ in the ranking } R(\rho) \mid \evF}\cdot \Pr\insquare{\evF} +\Pr_{}\insquare{i\in B_t \text{ in the ranking } R(\rho) \mid \lnot \evF}\cdot \Pr\insquare{\lnot \evF}\\
                &\leq
                \Pr_{}\insquare{i\in B_t \text{ in the ranking } R(\rho) \mid \evF} +\Pr_{}\insquare{i\in B_t \text{ in the ranking } R(\rho) \mid \lnot \evF} \cdot \frac{\delta}{mq}
                    \tag{Using that $1-\frac{\delta}{mq}\leq \Pr[\evF]\leq 1$; see \cref{lem:gauss}}\\
                &\leq
                \frac{\delta}{mq} + \Pr_{}\insquare{i\in B_t \text{ in the ranking } R(\rho) \mid \evF}.
            \end{align*}

            Rearranging the above inequality implies that
            \begin{align*}
                 \Pr_{}\insquare{i\in B_t \text{ in the ranking } R(\rho) \mid \evF} > 0.
            \end{align*}
            Let $\rho_{t,\min }$ and $\rho_{t,\max }$ let the minimum and maximum utilities of an item appearing in $B_t$ in $R(\rho)$, i.e.,
            \[
                {\rho_{t, \min} = \min_{i\in B_t \text{ in } R(\rho)} \rho_{i}
                \quad\text{and}\quad
                \rho_{t, \max} = \max_{i\in B_t \text{ in } R(\rho)} \rho_{i}.}
            \]
            Recall that conditioned on $\evF$
            \[
                \text{for all $1\leq i\leq m$},\quad
                \abs{\rho_i - \mu_i} \leq \errBnd{}.\yesnum\label{eq:tmptmp}
            \]
            Since conditioned on $\evF$, $i$ appears in $B_t$ in $R(\rho)$ with positive probability, it must hold that
            \begin{align*}
                \mu_i - \errBnd{} \leq \rho_{t,\max}
                \quad\text{and}\quad
                \mu_i + \errBnd{} \geq \rho_{t,\min}.
            \end{align*}
            Consequently, conditioned on $\evF$ the following hold
            \[
                \mu_i \in \insquare{
                    \rho_{t,\min} - \errBnd{},
                    \rho_{t,\max} + \errBnd{}
                }.
                \yesnum\label{eq:boundOnMuI}
            \]
            Moreover, since $\mu_1\geq \mu_2\geq \dots\geq \mu_m$ and all blocks have size $k$, \cref{eq:tmptmp} implies the following bounds on
            $\rho_{t,\min }$ and $\rho_{t,\max }$
            \begin{align*}
                \rho_{t,\min } \in \mu_{tk} \pm \errBnd{}
                \quad\text{and}\quad
                \rho_{t,\max } \in \mu_{t(k-1)+1} \pm \errBnd{}.
                \yesnum\label{eq:boundsOnRhoMaxMin}
            \end{align*}
            To see the first containment above, note that since $\rho_1,\rho_2,\dots,\rho_{tk}\geq \mu_{tk} - \errBnd{}$ at least $tk$ items have values above $\mu_{tk} - \errBnd{}$ and, since $\rho_i < \mu_{tk} + \errBnd{}$ for any $i > tk$, at most $tk$ items have value larger than $\mu_{tk} + \errBnd{}$.
            The second containment follows analogously by replacing $tk$ by $t(k-1)+1$ in the previous argument.
            \cref{eq:boundOnMuI,eq:boundsOnRhoMaxMin} imply that
            \[
                \mu_i\in \insquare{
                    \mu_{tk} - 2\errBnd{},
                    \mu_{tk} + 2\errBnd{}
                }.
                \yesnum\label{eq:boundOnMuIWithMuTK}
            \]
            Thus, we have shown that conditioned on event $\evF$ \cref{eq:boundOnMuIWithMuTK} holds for any $(i,t)\not\in P_{\rm small}$.

        \paragraph{Step 2 (Completing the proof):}
            Suppose events $\evE$, $\evF$, and $\evG$ hold.
            Recall that this event happens with high probability (at least $1-3\delta-4m^{-\frac{1}{4}}$ and, hence, it will turn out that the expected utility of the ranking $R_{\rm alg}$ output by \cref{alg:main} conditioned on $\evE$, $\evF$, and $\evG$ is similar to the expected utility without this conditioning.
            Fix any $1\leq j\leq q$ and any ranking $R$ sampled from a distribution that satisfies $C$-individual fairness constraints. %
            Consider any item $i$ appearing in block $B_j$.
            Since $\evG$ holds, $(i,j)\not\in P_{\rm small}$ and further since $\evF$ holds, \cref{eq:boundOnMuIWithMuTK} implies that
            \[
                \mu_i \in \mu_{jk} \pm 2\errBnd{}.
            \]
            Using the fact that $\evF$ holds again, implies that
            \[
                \rho_i \in \mu_{jk} \pm 3\errBnd{}.
                \yesnum\label{eq:noChangeTillHere}
            \]
            Moreover, since $\evE$ holds, it follows that
            \[
                \rho_i \in  S\inparen{1-\frac{jk}{m}} \pm \frac{2S}{m^{\frac{1}{4}}} \pm 3\errBnd{}.
                \yesnum\label{eq:boundOnRhoIConditional}
            \]
            Consider the ranking $R_{\rm alg}\sim \cD_{\rm alg}$ output by \cref{alg:main}.
            Since $\cD_{\rm alg}$ satisfies $C$-individual fairness constraints, the utility of $R_{\rm alg}$ conditioned on $\evE$, $\evF$, and $\evG$ is lower bounded as follows
            \begin{align*}
                \rho^\top R_{\rm alg} v
                &= \sum_{j=1}^q \sum_{s\in B_j} v_s \sum_{i=1}^m \rho_i \inparen{R_{\rm alg}}_{ij}\\
                &\geq \sum_{j=1}^q \sum_{s\in B_j} v_s \abs{B_j} \inparen{ S\inparen{1-\frac{jk}{m}} - \frac{2S}{m^{\frac{1}{4}}} - 3\errBnd{}}.
                \tag{Using that \cref{eq:boundOnRhoIConditional} holds for all $i\in B_j$}
            \end{align*}
            Let $\cD^\star$ be  a  distribution with the highest expected utility subject to satisfying $C$-individual fairness constraints.
            Consider $R^\star\sim \cD^\star$.
            Since $\cD^\star$ satisfies $C$-individual fairness constraints, the utility of $R^\star$ conditioned on $\evE$, $\evF$, and $\evG$ is upper bounded as follows
            \begin{align*}
                \rho^\top R^\star v
                &= \sum_{j=1}^q \sum_{s\in B_j} v_s \sum_{i=1}^m \rho_i \inparen{R^\star}_{ij}\\
                &\leq \sum_{j=1}^q \sum_{s\in B_j} v_s \abs{B_j} \inparen{ S\inparen{1-\frac{jk}{m}} + \frac{2S}{m^{\frac{1}{4}}} + 3\errBnd{}}.
                \tag{Using that \cref{eq:boundOnRhoIConditional} holds for all $i\in B_j$}
            \end{align*}
            Combining the above two inequalities implies that
            \begin{align*}
                \Ex_{R_{\rm alg}, \ R^\star, \rho, \mu}\insquare{
                    \frac{\rho^\top R_{\rm alg} v }{\rho^\top R^\star v } \mid \evE, \evF, \evG
                }
                &\geq \frac{
                    \sum_{j=1}^q \sum_{s\in B_j} v_s \abs{B_j} \inparen{
                        1-m^{-1}jk - 2m^{-\frac{1}{4}} - 3S^{-1}\errBnd{}
                    }
                }{
                    \sum_{j=1}^q \sum_{s\in B_j} v_s \abs{B_j} \inparen{
                        1-m^{-1}jk + 2m^{-\frac{1}{4}} + 3S^{-1}\errBnd{}
                    }
                }\\
                &\geq \min_{1\leq j\leq q}\min_{s\in B_j} \frac{
                    1-m^{-1}jk - 2m^{-\frac{1}{4}} - 3S^{-1}\errBnd{}
                }{
                    1-m^{-1}jk + 2m^{-\frac{1}{4}} + 3S^{-1}\errBnd{}
                }\\
                &\geq \frac{
                    1-nm^{-1} - 2m^{-\frac{1}{4}} - 3S^{-1}\errBnd{}
                }{
                    1-nm^{-1} + 2m^{-\frac{1}{4}} + 3S^{-1}\errBnd{}
                }\tag{Using that for any $b>a>0$, $\frac{a-x}{b-x}$ is a decreasing function of $x$ on the interval $[0,b)$, $qk\leq n$, and $n\leq m$}
            \end{align*}
            \begin{align*}
                \qquad \qquad  &\geq \frac{
                    1-nm^{-1} - 2m^{-\frac{1}{4}} - 3S^{-1}\errBnd{}
                }{
                    1-nm^{-1} + 2m^{-\frac{1}{4}} + 3S^{-1}\errBnd{}
                }\\
                &= \inparen{
                    1 - \frac{
                            2m^{-\frac{1}{4}} + 3S^{-1}\errBnd{}
                        }{
                            1-nm^{-1}
                        }
                }\inparen{
                    1 + \frac{
                            2m^{-\frac{1}{4}} + 3S^{-1}\errBnd{}
                        }{
                            1-nm^{-1}
                        }
                }^{-1}\\
                &\geq \inparen{
                    1 - \frac{
                            2m^{-\frac{1}{4}} + 3S^{-1}\errBnd{}
                        }{
                            1-nm^{-1}
                        }
                }^2\tag{Using that $\frac{1}{1+x}\geq 1-x$ for all $x\in \R$}\\
                &\geq {
                    1 - \frac{
                            4m^{-\frac{1}{4}} + 6S^{-1}\errBnd{}
                        }{
                            1-nm^{-1}
                        }
                }.\tag{Using that $\inparen{1-x}^2\geq 1-2x$ for all $x\in \R$}
            \end{align*}
            An unconditional lower bound on the utility ratio follows because of the following equality
            \begin{align*}
                \Ex\insquare{
                    \frac{\rho^\top R_{\rm alg} v }{\rho^\top R^\star v }
                }
                &= \Ex\insquare{
                    \frac{\rho^\top R_{\rm alg} v }{\rho^\top R^\star v } \mid \evE, \evF, \evG
                } \Pr\insquare{\evE, \evF, \evG}
                + \Ex\insquare{
                    \frac{\rho^\top R_{\rm alg} v }{\rho^\top R^\star v } \mid \lnot (\evE, \evF, \evG)
                } \Pr\insquare{\lnot (\evE, \evF, \evG)}.
            \end{align*}
            Where all expectations and probabilities are with respect to the randomness in $R_{\rm alg}, R^\star, \rho,$ and $\mu$.
            Finally, we can lower bound $\Ex\insquare{
                    \frac{\rho^\top R_{\rm alg} v }{\rho^\top R^\star v }
                }$ as follows
            \begin{align*}
                \Ex\insquare{
                    \frac{\rho^\top R_{\rm alg} v }{\rho^\top R^\star v }
                }
                &= \Ex\insquare{
                    \frac{\rho^\top R_{\rm alg} v }{\rho^\top R^\star v } \mid \evE, \evF, \evG
                } \Pr\insquare{\evE, \evF, \evG}
                + \Ex\insquare{
                    \frac{\rho^\top R_{\rm alg} v }{\rho^\top R^\star v } \mid \lnot (\evE, \evF, \evG)
                } \Pr\insquare{\lnot (\evE, \evF, \evG)}\\
                &\geq \Ex\insquare{
                    \frac{\rho^\top R_{\rm alg} v }{\rho^\top R^\star v } \mid \evE, \evF, \evG
                } \Pr\insquare{\evE, \evF, \evG}\\
                &\geq \Ex\insquare{
                    \frac{\rho^\top R_{\rm alg} v }{\rho^\top R^\star v } \mid \evE, \evF, \evG
                } \inparen{1-3\delta-4m^{-\frac{1}{4}}}\\
                &\geq   \inparen{
                            1 - \frac
                                { 4m^{-\frac{1}{4}} + 6S^{-1}\errBnd{} }
                                { 1-nm^{-1} }
                        }
                    \cdot \inparen{1-3\delta-4m^{-\frac{1}{4}}}\\
                &\geq   1-3\delta - \frac
                                { 8m^{-\frac{1}{4}} + 6S^{-1}\errBnd{} }
                                { 1-nm^{-1} }.
                                \yesnum
            \end{align*}
        The above inequality holds for all $\delta\in (0,1]$.
        To get the required bound it suffices to set $\delta=\frac{\sigma_{\max}}{S}$ (this value ensures that both terms involving $\delta$ are of the same order -- up to logarithmic factors in parameters -- although, one may be able to improve the lower bound by choosing a different $\delta$).
        To see this, observe that subsituting $\delta=\frac{\sigma_{\max}}{S}$ in the above inequality and using the fact that $1-nm^{-1}$ is bounded away from 1, implies that
        \begin{align*}
            \Ex\insquare{
                    \frac{\rho^\top R_{\rm alg} v }{\rho^\top R^\star v }
            }
            &\geq   1
                    - \frac{3\sigma_{\max}}{S}
                    - \frac{8m^{-\frac{1}{4}}}{1-nm^{-1}}
                    - \frac
                        { 6\sigma_{\max} \sqrt{2\log\frac{mS}{\sigma_{\max}}} }
                        { S\inparen{1-nm^{-1}} }\\
            &=   1
                    - O\inparen{\frac{\sigma_{\max}}{S}}
                    - O\inparen{m^{-\frac{1}{4}}}
                    - O\inparen{
                        \frac{ \sigma_{\max}  }{ S}
                        \sqrt{\log\frac{mS}{\sigma_{\max}}}
                        }
                    \tag{Using that $1-nm^{-1}$ is bounded away from 1}\\
            &\geq   1
                    - \wt{O}\inparen{
                            \frac{\sigma_{\max}}{S}
                            \cdot
                            \sqrt{\log{m}}
                        }
                    - O\inparen{m^{-\frac{1}{4}}}.
        \end{align*}

        \paragraph{Proof of bound when $\mu_1,\mu_2,\dots,\mu_m$ are arbitrary deterministic values.}
            The above proof also implicitly bounds $\Ex\insquare{
                    \frac{\rho^\top R_{\rm alg} v }{\rho^\top R^\star v }
            }$ when $\mu_1,\mu_2,\dots,\mu_m$ are arbitrary deterministic values.
            As before, without loss of generality assume that
            \[
                \mu_1\geq \mu_2 \geq \dots \geq \mu_m.
            \]
            Redefine $\evE$ to be the event that is always true.
            Continue with the above proof till \cref{eq:noChangeTillHere}.
            This shows that conditioned on $\evE$, $\evF$, and $\evG$,
            for both the ranking $R_{\rm alg}$ output by \cref{alg:main} and the ranking $R^\star\sim \cD^\star$, it holds that
            for any $t\in [q]$ and any $i$ appearing in block $B_t$ of either  $R_{\rm alg}$ or $R^\star$ \cref{eq:noChangeTillHere} holds.
            This implies the following lower and upper bounds
            \begin{align*}
                \rho^\top R_{\rm alg} v
                &\geq \sum_{j=1}^q \sum_{s\in B_j} v_s \abs{B_j} \inparen{\mu_{jk} - 3\errBnd{}},\\
                \rho^\top R^\star v
                &\leq \sum_{j=1}^q \sum_{s\in B_j} v_s \abs{B_j} \inparen{\mu_{jk} + 3\errBnd{}}.
            \end{align*}
            These, in turn, imply the following lower bound on $\Ex\insquare{
                    \frac{\rho^\top R_{\rm alg} v }{\rho^\top R^\star v }
            }$
            \begin{align*}
                \Ex\insquare{
                    \frac{\rho^\top R_{\rm alg} v }{\rho^\top R^\star v }
                }
                &= \Ex\insquare{
                    \frac{\rho^\top R_{\rm alg} v }{\rho^\top R^\star v } \mid \evE, \evF, \evG
                } \Pr\insquare{\evE, \evF, \evG}
                + \Ex\insquare{
                    \frac{\rho^\top R_{\rm alg} v }{\rho^\top R^\star v } \mid \lnot (\evE, \evF, \evG)
                } \Pr\insquare{\lnot (\evE, \evF, \evG)}\\
                &\geq \Ex\insquare{
                    \frac{\rho^\top R_{\rm alg} v }{\rho^\top R^\star v } \mid \evE, \evF, \evG
                } \Pr\insquare{\evE, \evF, \evG}\\
                &\geq \Ex\insquare{
                    \frac{\rho^\top R_{\rm alg} v }{\rho^\top R^\star v } \mid \evE, \evF, \evG
                } \cdot \inparen{1-3\delta-4m^{-\frac{1}{4}}}\\
                &\geq \frac{
                        \sum_{j=1}^q \sum_{s\in B_j} v_s \abs{B_j} \inparen{\mu_{jk} - 3\errBnd{}}
                    }{
                        \sum_{j=1}^q \sum_{s\in B_j} v_s \abs{B_j} \inparen{\mu_{jk} + 3\errBnd{}}
                    }
                    \cdot \inparen{1-3\delta-4m^{-\frac{1}{4}}}\\
                &\geq \min_{1\leq j\leq q}\min_{s\in B_j}
                    \frac{
                        \mu_{jk} - 3\errBnd{}
                    }{
                        \mu_{jk} + 3\errBnd{}
                    }
                    \cdot \inparen{1-3\delta-4m^{-\frac{1}{4}}}\\
                &\geq
                    \frac{
                        \mu_{n} - 3\errBnd{}
                    }{
                        \mu_{n} + 3\errBnd{}
                    }\cdot \inparen{1-3\delta-4m^{-\frac{1}{4}}}
                \intertext{{Where we used the facts that for any $a>0$, $\frac{x-a}{x+a}$ is an increasing function of $x$ on the interval $[0,\infty)$ with $x=\mu_{jk}$ and $a=3\errBnd{}$.
                To use this fact we recall: $\mu_1\geq \mu_2\geq \dots \geq \mu_m$, $qk\leq n$, and $n\leq m$.}
                Proceeding with the above chain of inequalities, we get}
                \Ex\insquare{
                    \frac{\rho^\top R_{\rm alg} v }{\rho^\top R^\star v }
                }
                &\geq
                    \inparen{1 - \frac{6\errBnd{}}{\mu_{n}}}
                    \cdot \inparen{1-3\delta-4m^{-\frac{1}{4}}}
                    \tag{Using that $\mu_n>0$, $\frac{1-x}{1+x}\geq 1-2x$ for all $x\in \R$}\\
                &\geq
                    1 - 3\delta - 4m^{-\frac{1}{4}} - \frac{6\errBnd{}}{\mu_{n}}.
            \end{align*}
            As before, this inequality holds for all $\delta\in (0,1]$, and substituting $\delta=\frac{\sigma_{\max}}{\mu_n}$, implies that
            \begin{align*}
                \Ex\insquare{
                    \frac{\rho^\top R_{\rm alg} v }{\rho^\top R^\star v }
                }
                \geq 1
                    - \wt{O}\inparen{
                            \frac{\sigma_{\max}}{\mu_n}
                            \cdot
                            \sqrt{\log{m}}
                        }
                    - O\inparen{m^{-\frac{1}{4}}}.
                    \yesnum\label{eq:proofForArbitraryMu}
            \end{align*}

        \newcommand{\unif}{\cU}

    \subsubsection{Proof of \cref{lem:unifConc}}
        For all $k\in [m]$, let $U_{(k:m)}$ be the $k$-th order statistic from $m$ independent draws from $\unif$, i.e., $U_{(k:m)}$ is the $k$-th smallest value from $m$ independent draws.
        The proof of \cref{lem:unifConc} uses the following standard facts.
        \begin{fact}[{\protect{\textbf{\cite[Eqs. 8.2 and 8.8]{ahsanullah2013introduction}}}}]\label{fact:os}
          For all $j\in [m]$, it holds that
          \begin{align*}
              \Ex\insquare{U_{(j:m)}} &= \frac{j}{m+1},\\
              \mathrm{Var}\insquare{U_{(j:m)}} &= \frac{j(m-j+1)}{(m+1)^2(m+2)}\leq \frac{1}{m+2}.
          \end{align*}
        \end{fact}

        \begin{fact}[\protect{\textbf{\cite[Example 5.1]{ahsanullah2013introduction}}}]\label{fact:conditional_os}
          Let $U_{(1:m)},\dots,U_{(m:m)}$ and $V_{(1:m)},\dots,V_{(m:m)}$ be two sets of order statistics of the uniform distribution.
          For any fixed $j\in [m]$ and threshold $t\in (0,1)$,
          conditioned on the event $U_{(j:m)}\negsp =\negsp t$,
          \begin{itemize}
            \item for all $\ell=j+1,j+2,\dots,m$, the distribution of $U_{(j+\ell:m)}$ is the same as the distribution of $t + (1-t) V_{(\ell:m-j)}$,
            \item for all $\ell=1,2,\dots,j-1$, the distribution of $U_{(\ell:m)}$ is the same as the distribution of $t V_{(\ell:j-1)}.$
          \end{itemize}
        \end{fact}

        \begin{proof}[Proof of \cref{lem:unifConc}]
            With some abuse of notation, in this proof, we use $\mu_i$ to refer to $\frac{\mu_i}{S}$ for each $1\leq i\leq m$.
            Here, $\mu_i$ is the $\inparen{m-i}$-th order statistic of $m$ draws from $\unif$.
            Hence, it is distributed as $U_{(m-i:m)}$
            From \cref{fact:os}, we have that for all $i\in [m]$,
            $$\mathrm{Var}\insquare{\mu_i} = \mathrm{Var}\insquare{U_{(m-i:m)}} \leq \frac{1}{m}.$$
            Using Chebyshev's inequality~\cite{motwani1995randomized}, we get:
            \begin{align*}
              \forall\ i\in [m],\quad \Pr\insquare{  \abs{\mu_i - \Ex[\mu_i]}
              \geq m^{-\frac14}  } \leq m^{-\frac12}.
              \yesnum\label{eq:vi_conc:new}
            \end{align*}
            \noindent For each $i\in \N$, define
            \begin{align*}
              s(i)\coloneqq (i-1)\cdot m^{\frac34} + 1.
            \end{align*}
            Starting from the first order statistic, consider order statistics at intervals of $m^{\frac34}$:
            $$\mu_{s(1)},\ \ \mu_{s(2)},\ \dots.$$
            Let $\evG$ be the event that all of these order statistics are $m^{-\frac14}$-close to their mean:
            \begin{align*}
              \forall\ i\in \inbrace{1,2,\dots,\floor{\smash{m^{-\frac{3}{4}}\cdot m}}},\quad
              \abs{\mu_{s(i)} - \Ex[\mu_{s(i)}]} < m^{-\frac14}.
              \yesnum\label{eq:def_evG}
            \end{align*}
            Note that $\evG$ implies that \cref{def:eventE} holds because the $\ell$-th order statistic is necessarily between the $s(\sfloor{\ell\cdot m^{-\frac{3}{4}}})$-th and $s(\sfloor{\ell\cdot m^{-\frac{3}{4}}}+1)$-th order statistics.
            Taking the union bound over all $i\in [m]$, by using Equation~\eqref{eq:vi_conc:new}, we get that
            \begin{align*}
              \Pr[\evG]\geq 1-4 m^{\frac14}\cdot  m^{-\frac12}  = 1-4m^{-\frac14}.\yesnum\label{eq:whpcg}
            \end{align*}

        \end{proof}

    \subsection{Proofs Omitted From \cref{sec:model}}\label{sec:proofof:thm:hardness_main}

            In this section, we prove the following hardness result. %
            \begin{restatable}[]{theorem}{thmHardnessMain}
                \label{thm:hardness_main}
                Given matrix $D\in [0,1]^{m\times n}$, number $\eta \in [0, 1]$, and an instance of \cref{prob:main}, it is \np-hard to check if there is a distribution $\cD$ such that $D$ is the marginal of $\cD$ and $\Pr_{R\sim \cD}[R\in \rgf] > \eta.$
            \end{restatable}

            \noindent Our result is based on a reduction from the following matching problem that is known to be \np-hard \cite{Thomas2016}. %

            \paragraph{Couple constrained matching problem on bipartite cycles.}
            \begin{itemize}
                \item \textbf{Input:} A bipartite cycle graph $G(X\cup Y, E)$, a set of pairs or \textit{couples} $C\subseteq E\times E$, and an integer $k$
                \item \textbf{Output:} YES if there is a matching $M$ of size $k$ in graph $G$ such that for every couple $(e_1,e_2)\in C$, $e_1\in M$ if and only if $e_2\in M$.
                NO otherwise.
            \end{itemize}

            \newcommand{\pori}{\ensuremath{P_{\rm Original}}}
            \newcommand{\pdum}{\ensuremath{P_{\rm Dummy}}}
            \newcommand{\pgad}{\ensuremath{P_{\rm Gadget}}}

            \newcommand{\iori}{\ensuremath{I_{\rm Original}}}
            \newcommand{\idum}{\ensuremath{I_{\rm Dummy}}}
            \newcommand{\igad}{\ensuremath{I_{\rm Gadget}}}

            \paragraph{Construction.}
                Consider an instance $(X, Y, E, C, k)$ of the above problem.
                Since the underlying graph $G(X\cup Y, E)$ is a bipartite cycle, $\abs{X}=\abs{Y}$.
                If $\abs{X} < k$, then the problem is clearly a NO instance.
                Hence, we assume $\abs{X}\geq k$.
                We set
                \begin{align*}
                    n &\coloneqq m\coloneqq \abs{X}+2\abs{C}+\abs{X}-k,\\
                    q &\coloneqq 2\abs{C}+1,\\
                    p &\coloneqq q. %
                \end{align*}
                We divide the $m$ items into three disjoint sets \iori{}, \idum{}, and \igad{} of sizes $\abs{X}$, $\abs{X}-k$, and $2\abs{C}$ respectively.
                Let $\iori{}=X$.
                Similarly, we divide the positions into three disjoint sets \pori{}, \pdum{}, and \pgad{} of sizes $\abs{X}$, $\abs{X}-k$, and $2\abs{C}$ respectively.
                Let $\pori{}=Y$.
                It remains to construct the matrix $D$, blocks $B_1,\dots,B_p$, groups $G_1,\dots,G_p$, and matrices $L$ and $U$ specifying the fairness constraints.
                Define $$\phi\coloneqq {\frac{n-2k}{2(n-k)}}.$$
                We construct $D$, the blocks, and the groups as follows.
                \begin{itemize}[leftmargin=14pt,itemsep=1pt]

                    \item For each $\ell$, consider the $\ell$-th couple consisting of edges $(g_\ell,h_\ell)$ and $(i_\ell,j_\ell)$. %
                    select two items $\beta_\ell,\delta_\ell\in \igad$ and two positions $\alpha_\ell,\gamma_\ell\in \pgad{}$.
                    For each pair $x\in\inbrace{g_\ell,i_\ell,\beta_\ell,\delta_\ell}$ and $y\in\inbrace{h_\ell,j_\ell,\alpha_\ell,\gamma_\ell}$, set $D_{x,y}$ as specified in \cref{fig:ranking_instance}.
                    Where any edge absent from the figure has a weight of 0.
                    Construct groups $G_{2\ell-1}$ and $G_{2\ell}$ as in \cref{fig:groups} and blocks $B_{2\ell-1}$ and $B_{2\ell}$ as in \cref{fig:blocks}.

                    \begin{figure}
                        \centering
                        \subfigure[]{
                            \includegraphics[width=0.4\linewidth,trim={24cm 0cm 24cm 8cm},clip]{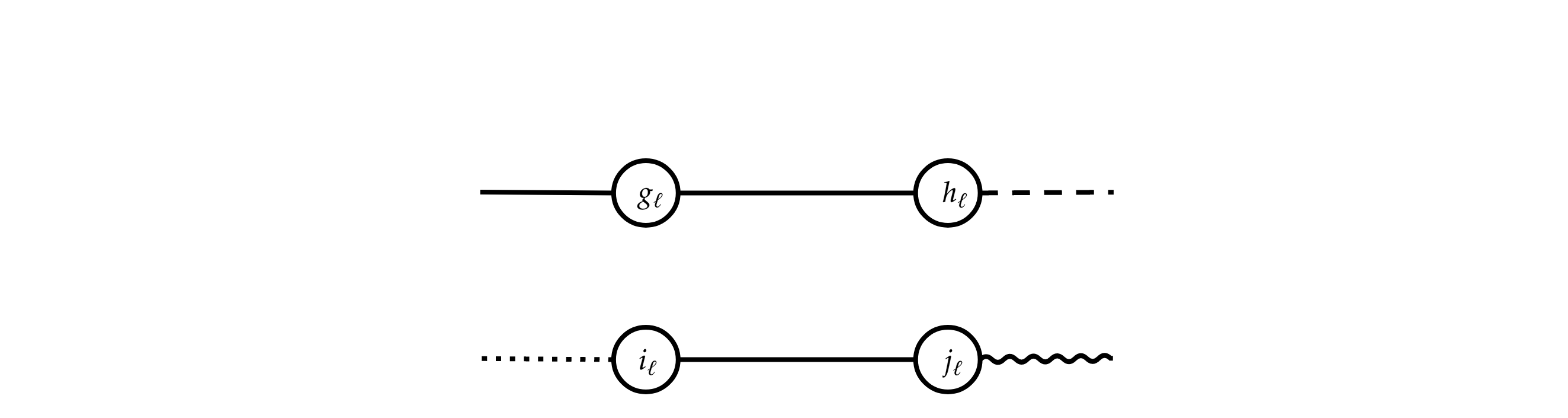}
                        }
                        \subfigure[\label{fig:ranking_instance_eg}]{
                            \includegraphics[width=0.4\linewidth,trim={28cm 0cm 20cm 0cm},clip]{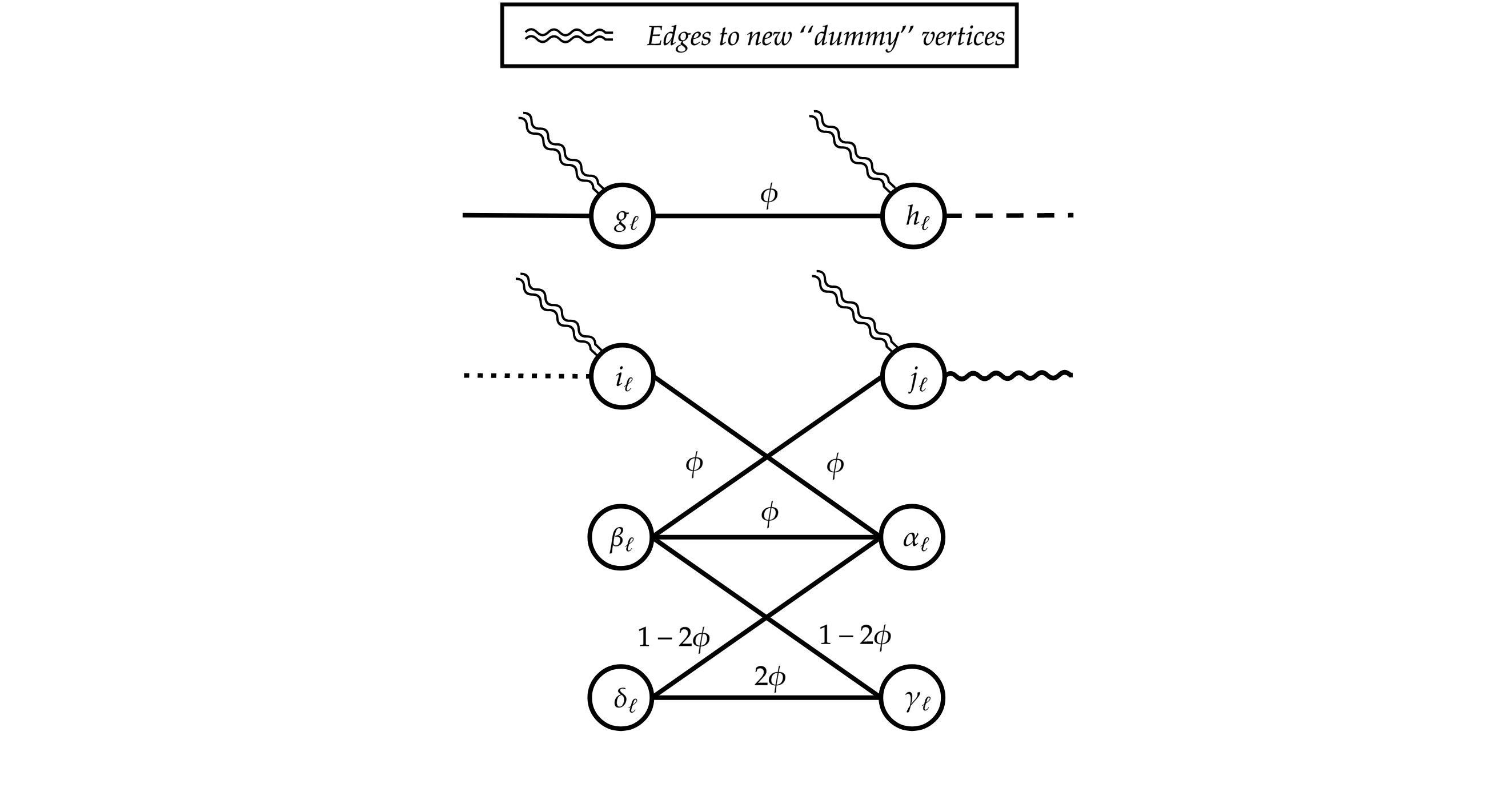}
                        }
                        \vspace{-0.1in}
                        \caption{The construction for the reduction from the couple-constrained matching problem to group-fair ranking.
                        For each couple consisting of edges $(g_\ell,h_\ell)$ and edges $(i_\ell,j_\ell)$, we introduce four additional vertices $\alpha_\ell,\beta_\ell,\gamma_\ell,\delta_\ell$ as shown in \cref{fig:ranking_instance_eg}
                        The numbers on the edge $(x,y)$ denotes the value of $D_{x,y}$; where $D_{x,y}=0$ for all edges absent from the figure.
                        }
                        \label{fig:ranking_instance}
                    \end{figure}

                    \begin{figure}
                        \centering
                        \subfigure[\label{fig:groups}]{
                            \includegraphics[width=0.45\linewidth,trim={28cm 0cm 10cm 0cm},clip]{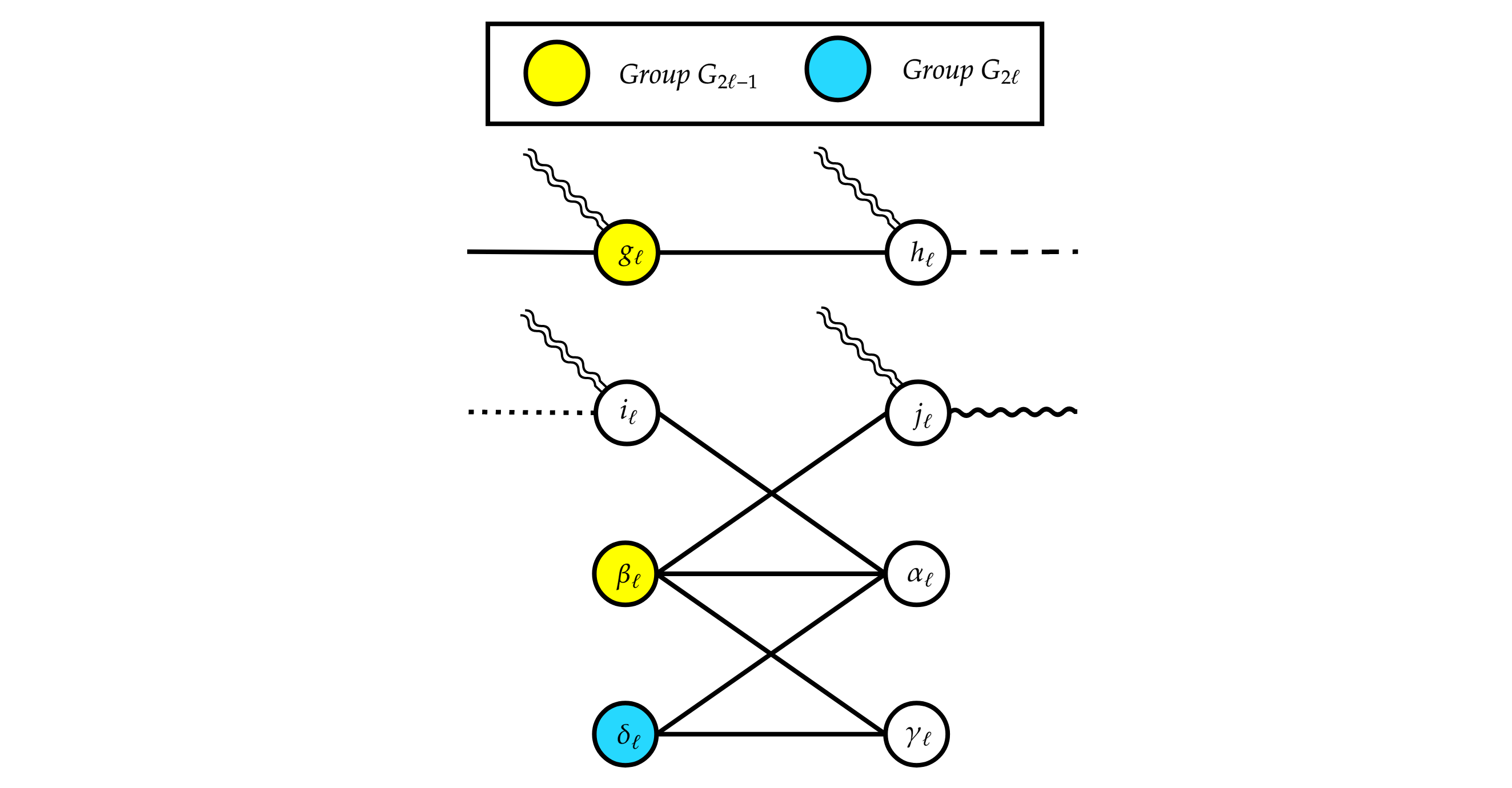}
                        }
                        \subfigure[\label{fig:blocks}]{
                            \includegraphics[width=0.45\linewidth,trim={28cm 0cm 10cm 0cm},clip]{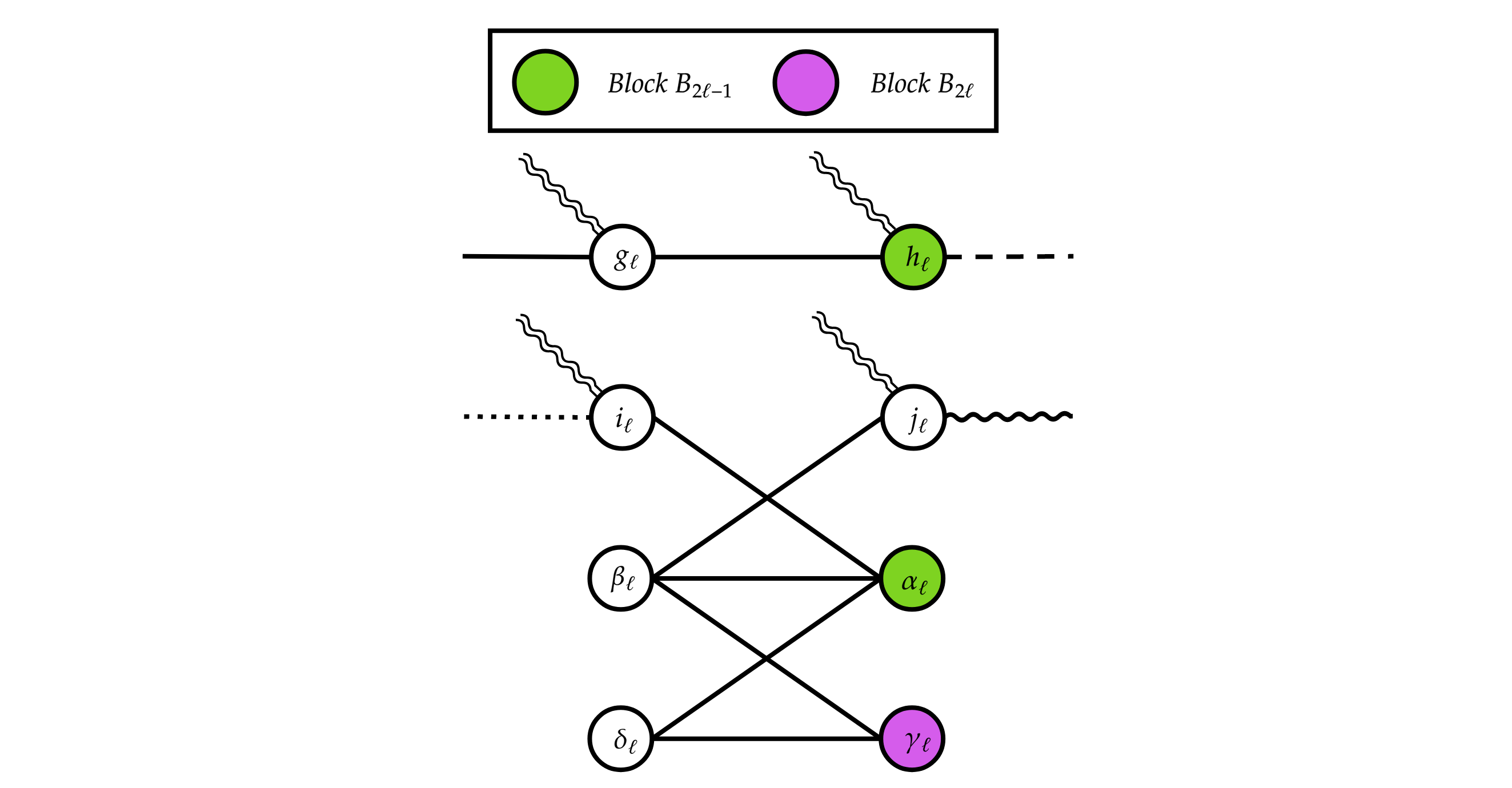}
                        }
                        \caption{Groups $G_{2\ell-1}$ and $G_{2\ell}$ and blocks $B_{2\ell-1}$ and $B_{2\ell}$ for each $\ell\in [p/2]$ appearing in the proof of \cref{thm:hardness_main}. Where for simplicity $p/2$ denotes $\floor{p/2}$.}
                        \label{fig:groups_and_blocks}
                    \end{figure}

                    \item For each edge $(i,j)\in \idum\times\pori$, set $D_{ij}=\frac{1-2\phi}{\abs{\idum}} = \frac{1-2\phi}{\abs{X}-k}$.
                    \item For each edge $(i,j)\in \iori{}\times \pdum$, set  $D_{ij}=\frac{1-2\phi}{\abs{\pdum}} = \frac{1-2\phi}{\abs{X}-k}$.
                    \item For each edge $(i,j)\in \iori{}\times \pori{}$ that is not a part of any couple constraint, set $D_{ij}=\phi$.
                    \item For any remaining edges $(i,j)$ set $D_{i,j}=0$.
                    Let $G_{\rm dummy}$ be the set of items that have been not included in groups $G_{2\ell-1}$ and $G_{2\ell}$ for any $\ell\in [p/2]$. (Where, for simplicity, we use $p/2$ to denote $\floor{p/2}$.)
                    Let $B_{\rm dummy}$ be the set of positions that have been not included in blocks $B_{2\ell-1}$ and $B_{2\ell}$ for any $\ell\in [p/2]$.
                \end{itemize}
                Finally, we define $L$ and $U$ to be the following $p\times q$ matrices.
                \begin{align*}
                    L \coloneqq \begin{bmatrix}
                        1 & 0 & \dots & 0 & 0\\
                        0 & 1 & \dots & 0 & 0\\
                        \vdots & \vdots & \ddots & \vdots & \vdots\\
                        0 & 0 & \dots & 1 & 0\\
                        0 & 0 & \dots & 0 & 0
                    \end{bmatrix}
                    \quad \text{and}\quad
                    U \coloneqq \begin{bmatrix}
                        1 & n & \dots & n & n\\
                        n & 1 & \dots & n & n\\
                        \vdots & \vdots & \ddots & \vdots & \vdots\\
                        n & n & \dots & 1 & n\\
                        n & n & \dots & n & n
                    \end{bmatrix}.
                \end{align*}
                In other words, the constraints specified by $L$ and $U$ require that there is exactly 1 item from $G_{k}$ placed in $B_k$ for each $k\in [p-1]$ and places no constraint on other group-block pairs.

        \paragraph{Polynomial time.}
            Clearly, the above construction can be completed in polynomial time in the bit-complexity of the input (i.e., $(X, Y, E, C,k)$).

        \paragraph{Main lemmas.}
            We will prove the following lemma.
            \begin{lemma}\label{lem:main_reduction}
                An instance $(X, Y, E, C, k)$ of the ``Couple constrained matching problem on bipartite cycles'' is a YES instance if and only if
                there is a ranking $R$ in the above construction such that (1) $R\in \rgf{}$ and (2) for any $i$ and $j$, if $R_{ij}=1$, then $D_{ij}>0$.
            \end{lemma}
            \noindent The reduction then follows from the following lemma.
            \begin{lemma}
                    Given a doubly stochastic matrix $D\in [0,1]^{m\times n}$,
                    there is a ranking $R$ in the above construction such that $R$ and $D$ satisfy the condition in \cref{lem:main_reduction}
                    if and only if
                    there is a distribution $\cD$ such that $D$ is the marginal of $\cD$ and $\Pr_{R\sim \cD}[R\in \rgf{}]>0$.
            \end{lemma}
            \begin{proof} The proof is divided into two parts.

                \paragraph{(Existence of $R$ implies existence of $\cD$).}
                Suppose there is a ranking $R$ satisfying the above properties.
                Then, there is an $\alpha>0$ such that $D-\alpha R\in [0,1]^{m\times n}$.
                Since $D$ is doubly stochastic $m$ must be equal to $n$.
                Hence, any ranking and, in particular, $R$ is also doubly stochastic
                Using that $R$ and $D$ are doubly stochastic, it follows that
                \begin{align*}
                    \forall i\in [m],\quad  &\quad \sum\nolimits_{j} (D_{ij}-\alpha R_{ij}) = 1 - \alpha,\\
                    \forall j\in [n],\quad  &\quad \sum\nolimits_{i} (D_{ij}-\alpha R_{ij}) = 1 - \alpha.
                \end{align*}
                Consequently, $\frac{D-\alpha R}{1-\alpha}$ is doubly stochastic.
                Thus, by the Carathéodory's theorem \cite{grotschel_lovasz_schrijver} there exists some distribution $\cD'$ over $\cR$ such that $\frac{D-\alpha R}{1-\alpha}$ is the marginal of $\cD'$--where $\cR$ is the set of all rankings.
                Consequently, $D$ is the marginal of the distribution $\cD$ defined as $$\Pr_{R'\sim \cD}[R'=S] = (1-\alpha) \Pr_{R'\sim \cD'}[R'=S]+\alpha\mathbb{I}[R=S].$$
                As $R\in \rgf{}$, $\cD$ satisfies that $\Pr_{S\sim \cD}[S\in \rgf{}]\geq \alpha > 0$.

                \paragraph{(Existence of $\cD$ implies existence of $R$).}
                    Fix any $R\in \rgf$ such that $\Pr_{S\sim \cD}[S=R]=\alpha>0$.
                    Since $D$ is the marginal of $\cD$, it follows that for all $i$ and $j$
                    \begin{align*}
                        D_{ij} = \Pr_{T\sim \cD}[T_{ij}=1]
                        \geq  \Pr_{T\sim \cD}[T=R]\cdot \mathbb{I}[R_{ij}=1]
                         > \alpha\mathbb{I}[R_{ij}=1].
                         \yesnum\label{}
                    \end{align*}
                    Since $\alpha>0$, for any $i$ and $j$, if $R_{ij}=1$, then $D_{ij}>0$.
            \end{proof}

        \noindent It remains to prove \cref{lem:main_reduction}, which we present below.

         \paragraph{Couple constrained matching $\to$ ranking.}
         Suppose we are given a YES instance for the couple-constrained matching problem.
         Let the corresponding matching be $M$.
         We will construct a ranking $R$ satisfying the properties in \cref{lem:main_reduction}.
         First, we will construct a partial ranking $R$ that leaves some positions unassigned.
         Then we will complete this ranking using the dummy items and positions.

         \smallskip
        \noindent \textit{Step 1: Construct partial ranking $R$.}
         The construction of the partial ranking is as follows.
         \begin{itemize}[leftmargin=14pt]
             \item For every $i\in \iori$ that is not part of any couple, set $R_{ij}=1$ for all $j\in \pori$, if $(i,j)\in M$ and $D_{ij}>0$
             \item For every $j\in \pori$ that is not part of any couple, set $R_{ij}=1$ for all $i\in \iori$, if $(i,j)\in M$ and $D_{ij}>0$
             \item For each $\ell$, consider the $\ell$-th couple.
             Suppose this couple consists of the edges $(g_\ell,h_\ell)$ and $(i_\ell,j_\ell)$.
             \begin{itemize}[leftmargin=14pt]
                \item If $(g_\ell,h_\ell)\in M$ (and, hence, $(i_\ell,j_\ell)\in M$), then for each $x\in\inbrace{g_\ell,i_\ell,\beta_\ell,\delta_\ell}$ and $y\in\inbrace{h_\ell,j_\ell,\alpha_\ell,\gamma_\ell}$, set $R_{x,y}=1$ if $(x,y)$ is green in \cref{fig:ranking_instance_1} and $R_{x,y}=0$ otherwise.
                \item If $(g_\ell,h_\ell)\not\in M$ (and, hence, $(i_\ell,j_\ell)\in M$), then for each $x\in\inbrace{g_\ell,i_\ell,\beta_\ell,\delta_\ell}$ and $y\in\inbrace{h_\ell,j_\ell,\alpha_\ell,\gamma_\ell}$, set $R_{x,y}=1$ if $(x,y)$ is green in \cref{fig:ranking_instance_2} and $R_{x,y}=0$ otherwise.
            \end{itemize}
            One can verify that in both cases the group fairness constraints are satisfied for both groups $G_{2\ell-1}$ and $G_{2\ell}$.
         \end{itemize}
         \begin{figure}
             \centering
             \vspace{-0.3in}
             \subfigure[\label{fig:ranking_instance_1}]{
                {\includegraphics[width=0.35\linewidth,trim={28cm 0cm 24cm 0cm},clip]{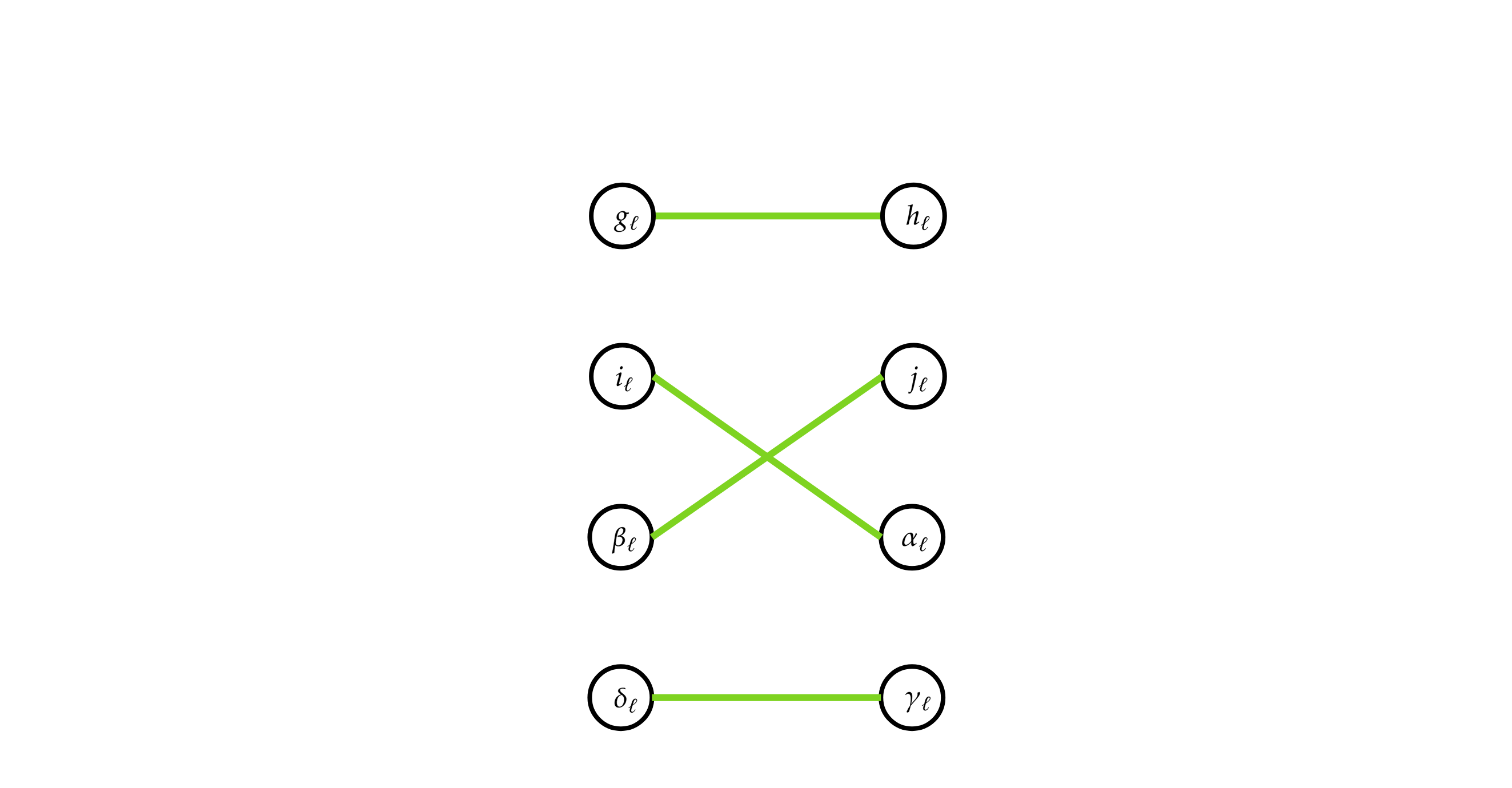}}
             }
             \subfigure[\label{fig:ranking_instance_2}]{
                {\includegraphics[width=0.35\linewidth,trim={28cm 0cm 24cm 0cm},clip]{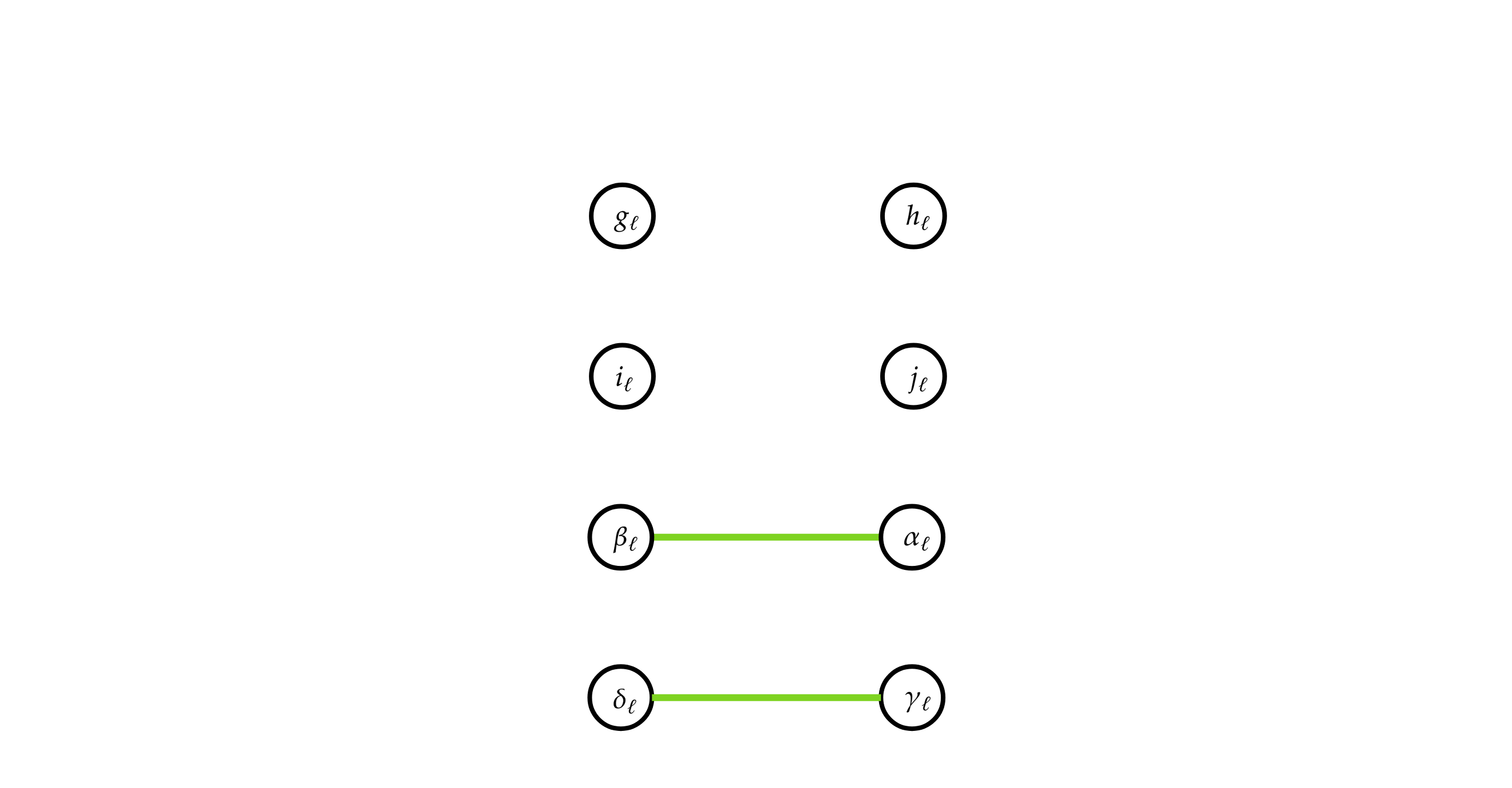}}
             }
             \vspace{-0.1in}
             \caption{Illustration of the rankings appearing in the  proof of \cref{lem:main_reduction}.}
             \label{fig:enter-label}
         \end{figure}
         One can verify that the above construction satisfies the following invariants.
         \begin{enumerate}
             \item For each $i\in \iori$, $\sum_j R_{ij}$ is equal to the number of edges incident on $i$ in $M$
             \item For each $j\in \pori$, $\sum_i R_{ij}$ is equal to the number of edges incident on $j$ in $M$
             \item The partial ranking satisfies all group fairness constraints
             \item All items $i\in \igad$ (respectively positions $j\in \pgad{}$) are matched to a unique item (respectively position).
             \item No dummy item $i\in \idum$ is placed in the ranking $R$ and all dummy positions $j\in \pdum$ are empty in the ranking $R$.
         \end{enumerate}
        The first two invariants above and the fact that $M$ matches at least $2k$ vertices imply that, (1) at most $\abs{X}-k$ items are not assigned to a position in the partial ranking and (2) at most $\abs{X}-k$ positions empty in the partial ranking.

        \smallskip
        \noindent \textit{Completing the ranking $R$.} For each such position $j\in \pori$ (respectively item $i\in \iori$), select a unique dummy item $i\in \idum$ (respectively dummy position $j\in \pori$)  such that $D_{ij}>0$ and set $R_{ij}=1$.
        Such a dummy item (respectively position) exists because of invariant 5, the fact that $\abs{\idum}\geq \abs{X}-k$ (respectively $\abs{\pdum}\geq \abs{X}-k$), and that $D_{ij}>0$ for any dummy item $i$ and any $j\in [n]$ (respectively any dummy position $j$ and any $i\in [m]$).
        Finally, for any empty dummy positions $j\in\pdum$ select a unique dummy item $i\in \idum$ such that $D_{ij}>0$ and set $R_{ij}=1.$

        \smallskip
        \noindent \textit{Step 2: $R$ satisfies the required properties.}
            By construction, it follows that $R_{ij}=1$ then $D_{ij}>0$.
            From invariant 3 and the fact that no block or group contains dummy items, it follows that $R$ satisfies the group fairness constraint.
            Finally, invariant 4 and the construction in Step 2 imply that $R$ is a valid ranking.

         \paragraph{Ranking $\to$ couple constrained matching.}
         Suppose there exists a ranking $R$ satisfying the properties in \cref{lem:main_reduction}.
         We will construct a corresponding matching $M$ which matches at least $k$ items and satisfies the couple constraints.
         The construction is as follows:
         \begin{enumerate}[itemsep=-1pt]
             \item For each $(i,j)\in \iori{} \times \pori{}$, then match $i$ and $j$ in $M$
             \item For each $\ell$, consider the $\ell$-th couple.
             Suppose this couple consists
             $(g_\ell,h_\ell)$ and $(i_\ell,j_\ell)$.
             \begin{enumerate}
                 \item If $R_{i_\ell,\alpha_\ell}=1$, then include edges $(g_\ell,h_\ell)$ and $(i_\ell,j_\ell)$ in $M$.
                 \item If $R_{\beta_\ell,j_\ell}=1$, then include edges $(g_\ell,h_\ell)$ and $(i_\ell,j_\ell)$ in $M$.
             \end{enumerate}
         \end{enumerate}
         By construction, it follows, that (1) for item $i\in \iori$, the number of edges incident on $i$ is at most $\sum_j R_{ij}=1$ and
         (2) for item $j\in \pori$, the number of edges incident on $j$ is at most $\sum_i R_{ij}=1$.
         Hence, $M$ is a valid matching.
         To see that $M$ satisfies couple-constraints, consider any $c\in C$.
         Suppose $c$ consists of the edges $(i_c,j_c)$ and $(g_c,h_c)$.
         By construction, in $D_{g_c,h_c}=0$ and, hence, $R_{g_c,h_c}=0$.
         Thus, $(g_c,h_c)$ must have been added in Step 2 of the above construction.
         By construction in Step (2), it follows that $(i_c,j_c)$ is also present in $M$.
         A symmetric argument shows that if $(i_c,j_c)$ is in $M$ then so is $(g_c,h_c)$.
         Since this holds for all couples $c\in C$, it follows that $M$ satisfies the couple constraints.
         It remains to show that $M$ has at least $k$ edges.
         This is true because by construction every item $i\in \iori{}$ that is placed in a non-dummy position in $R$, is matched in $M$.
         Since there are at most $\abs{X}-k$ dummy positions and each item is placed at some position in $R$, it follows that at least $k$ items are placed in a non-dummy position in $R$.
         Thus, $M$ has at least $k$ edges.

\section{Limitations and Conclusion}\label{sec:limitation}\label{sec:conclusion}
    We present an algorithm (\cref{alg:main}) that works with a general class of group fairness constraints and individual fairness constraints, it outputs rankings sampled from a distribution that satisfies the specified individual fairness constraints and, moreover, each output ranking satisfies the specified group fairness constraints (\cref{thm:algo_main}).
    Further, the algorithm guarantees a constant fraction approximation of the optimal (expected) utility subject to satisfying these constraints (\cref{thm:algo_main,thm:approxStochasUtil}).
    This algorithm works with families of disjoint protected groups as well as certain families of overlapping protected groups (namely, collections of laminar sets) (\cref{sec:overview}).
    Empirically, we observe that our algorithm is able to satisfy the specified fairness criteria while losing at most 6\% loss of the utility compared to the unconstrained baseline (\cref{sec:empirical}).

    Our work raises several questions.
    We consider the setting where the utility of a ranking of multiple items is a linear function of the utilities of individual items.
    While this captures a broad spectrum of applications \cite{fair_ranking_survey1,fair_ranking_survey2,overviewFairRanking}, in some applications, the utility of a ranking may be a non-linear function of the items present in the ranking; this is particularly, the case where the diversity of the items in a ranking has an effect on its utility \cite{microsoft_diverse,asadpour2022sequential,kleinberg2022ordered}.
    Extending our approach to this (more complicated) setting is an interesting direction for future work.
    Moreover, while our algorithm works for certain families of overlapping protected groups, extending it to arbitrary families of overlapping protected groups is an important question.
    Further, \cref{thm:hardness_main} demonstrates that solving a certain relaxation of our problem is \np-hard, exploring other potential relaxations may be fruitful to further improve the utility guarantee of our algorithm.

    \paragraph{Acknowledgments.}
        Part of this work was done when AM was an intern at Microsoft Research.
        AM was supported in part by NSF Awards CCF-2112665 and IIS-204595. SG was supported by a Google Ph.D. Fellowship award. AL was supported in part by a Pratiksha Trust Young Investigator Award. AL is also grateful to Microsoft Research for supporting this collaboration.

\newpage

\printbibliography

\clearpage

\appendix

\section{An Example With a Unique, Optimal, and Fractional Vertex}\label{sec:fractional_vertex}
        In this section, we present an example of group fairness constraints, individual fairness constraints, and utilities, where the (unique) optimal solution of \prog{prog:mod_of_ashudeep} has fractional entries; this example is inspired by a similar fact (without individual fairness constraints) in \citet{celis2018ranking}.
        \begin{fact}\label{lem:existence_of_extreme_point}
            There exists an instance of \prog{prog:mod_of_ashudeep}, such that, the unique optimal solution $\Pi$ is fractional.
            Moreover, $\Pi$ cannot be represented as a convex combination of rankings that satisfy the corresponding group-fairness constraints.
        \end{fact}
        \begin{proof}
            Let $m=n=4$. Suppose there is one protected group $G_1\coloneqq \inbrace{1,2}$ and two blocks $B_1\coloneqq \inbrace{1,2}$ and $B_2\coloneqq \inbrace{3}$. %
            Let $U_{11}=1$ and, for any $j\neq 1$, $U_{1j}=\infty$.
            That is, any feasible distribution over rankings must place at most 1 item from $G_1$ in the first two positions. %
            Let $k\coloneqq 2$ and, for each item $i$, $C_{i1}=\frac{1}{2}$,
            and $C_{32}=1$, and $A_{ij}=1$ for all $j$.
            That is, in any feasible distribution over rankings, each item should appear in one of the first two positions with probability at least $\frac{1}{2}$ (in fact, with probability exactly $\frac{1}{2}$) and item $i=3$ must appear in the first three positions.
            Let $\rho_1 > \rho_4 > \rho_3 > \rho_2.$
            Consider the following matrix $\Pi\in [0,1]^{m\times n}$ denoting the marginal of some distribution rankings:
            \begin{align*}
                \Pi = \begin{bmatrix}
                    \frac{1}{2} & 0 & \frac{1}{2} & 0\\
                    0 & \frac{1}{2} & 0 & \frac{1}{2}\\
                    0 & \frac{1}{2} & \frac{1}{2} & 0\\
                    \frac{1}{2} & 0 & 0 & \frac{1}{2}
                \end{bmatrix}.
            \end{align*}
            \paragraph{A. $\Pi$ is a vertex solution.}
            It is easy to verify that $\Pi$ is feasible for \prog{prog:mod_of_ashudeep}.
            We claim that $\Pi$ is also a vertex of polytope formed by the set of feasible solutions of \prog{prog:mod_of_ashudeep}.
            Indeed $\Pi$ is supported by $nm$ ($=16$) linearly independent inequalities.
            \begin{itemize}
                \item Eight inequalities are $\Pi_{12}=0$, $\Pi_{14}=0$, $\Pi_{21}=0$, $\Pi_{23}=0$, $\Pi_{31}=0$, $\Pi_{34}=0$, $\Pi_{42}=0$, and $\Pi_{43}=0$.
                \item There are $n+m$ inequalities that require row and column sums to be 1.
                Out of these, $n+m-1$ are linearly independent.
                \item The remaining inequality comes from the set of individual fairness constraints:
                $\Pi_{i1}+\Pi_{i2}=\frac{1}{2}$ for each $i\in [m]$.
            \end{itemize}
            (One can verify that 16 of the above inequalities are linearly independent.) %

            \paragraph{B. $\Pi$ cannot be represented as a convex combination of group-fair rankings.}
                $\Pi$ \textit{uniquely} decomposes as a convex combination of rankings as follows:
                \begin{align*}
                \Pi = \frac{1}{2}\underbrace{\begin{bmatrix}
                    0 & 0 & 1 & 0\\
                    0 & 0 & 0 & 1\\
                    0 & 1 & 0 & 0\\
                    1 & 0 & 0 & 0
                \end{bmatrix}}_{\Pi_a}
                + \frac{1}{2}\underbrace{\begin{bmatrix}
                    1 & 0 & 0 & 0\\
                    0 & 1 & 0 & 0\\
                    0 & 0 & 1 & 0\\
                    0 & 0 & 0 & 1
                \end{bmatrix}}_{\Pi_b}
            \end{align*}
            Here $\Pi_a$ satisfies group fairness constraints and $\Pi_b$ violates the group fairness constraints (it violates the upper bound for group $G_1$ on block 1).
            Since this is the unique decomposition of $\Pi$ into rankings and $\Pi_b$ is not group-fair, it follows that $\Pi$ cannot be decomposed as a convex sum of group-fair rankings.

            \paragraph{C. $\Pi$ is the unique optimal solution.}
                Due to the individual fairness constraints and the constraints that the sum of entries in columns one and two is 1, it follows that each item appears in the first 2 positions with probability exactly $\frac{1}{2}$.
                Since $\rho_1,\rho_4>\rho_2,\rho_3$, by a swapping argument it follows that it is optimal to have the first two columns exactly as in $\Pi$.
                Next, to satisfy the individual fairness constraint specified by $C_{32}$, it is necessary to satisfy that $\Pi_{33}=\frac{1}{2}$.
                To satisfy the third column sum, we require $\Pi_{13}+\Pi_{23}+\Pi_{43}=\frac{1}{2}$.
                As $\rho_1>\rho_4,\rho_2$, by a swapping argument, it follows that it is optimal to have $\Pi_{13}=\frac{1}{2}$.
                Hence, the entries of column 4 are determined by the other three columns as the rows must sum to 1.
        \end{proof}

\section{Approximating Prefix-Based Constraints by Block-Based Constraints}\label{sec:additional_remarks}
        In this section, we consider general families of group fairness and individual fairness constraints considered in the literature and show they can be approximated using group fairness constraints in \cref{def:group_constraints} and the individual fairness constraints in \cref{def:individual_constraints} can be used to approximate them.

        First, we consider the following family of group fairness constraints considered by \cite{celis2018ranking,zehlike2017topk}.
        \begin{definition}[\cite{celis2018ranking,zehlike2017topk}]\label{def:group_constraints2}
            Given matrices $L_{\rm pre},U_{\rm pre}\in \Z^{n\times p}$ a ranking $R$ satisfies the prefix-based $(L_{\rm pre},U_{\rm pre})$-group fairness constraints if for each $j\in [n]$ and $\ell\in [p]$
            \begin{align*}
                \inparen{L_{\rm pre}}_{j\ell}\leq \sum\nolimits_{i\in G_\ell}\sum\nolimits_{t=1}^j R_{ij} \leq \inparen{U_{\rm pre}}_{j\ell}.
                \yesnum\label{eq:group_constraints_prefix}
            \end{align*}
        \end{definition}
        \noindent The following lemma shows how the constraints in \cref{def:group_constraints} can approximate the above constraints.
        \begin{lemma}
            Let $k=2$ and $q=\frac{n}{2}$.
            For any matrices $L_{\rm pre}, U_{\rm pre}\in \Z^{n\times p}$ defining \textit{feasible} prefix-based $(L_{\rm pre}, U_{\rm pre})$-group fairness constraints, there exist matrices $L, U\in \Z^{\frac{n}{2}\times p}$ such that any ranking $R$ satisfying the $(L, U)$-group fairness constraints violates prefix-based $(L_{\rm pre}, U_{\rm pre})$-group fairness constraints at any position by at most an additive factor of $1$.
        \end{lemma}
        \noindent For several natural fairness constraints (such as equal representation or proportional representation), the lower and upper bound constraints at the $j$-th position (for any $1\leq j\leq n$) are of order $\Theta(j)$, and hence, the above result shows that the ranking satisfying the $(L,U)$-group fairness constraints violates the prefix based constraints by at most a multiplicative factor of $\Theta(n^{-1})$ for any $j=\Theta(n)$.
        \begin{proof}
            Fix $k\coloneqq 2$ {and $q=\frac{n}{2}$.}
            Let $R_{\rm pre}$ be any ranking that satisfies the prefix-based $(L_{\rm pre},U_{\rm pre})$-group fairness.
            $R_{\rm pre}$ exists as the prefix-based constraints are feasible.
            We claim that the following matrices $L$ and $U$ suffice: for each $j\in [\frac{n}{2}]$ and $\ell\in [p]$
            \begin{align*}
                L_{j\ell}
                &= \sum\nolimits_{i\in G_\ell} \inparen{ \inparen{R_{\rm pre}}_{ij} + \inparen{R_{\rm pre}}_{i(j+1)} },\yesnum\label{eq:def_1}\\
                U_{j\ell}
                &= \sum\nolimits_{i\in G_\ell} \inparen{ \inparen{R_{\rm pre}}_{ij} + \inparen{R_{\rm pre}}_{i(j+1)} }.\yesnum\label{eq:def_2}
            \end{align*}
            To complete the proof, consider any ranking $R$ that satisfies $(L,U)$-group fairness constraints.
            Since $R$ satisfies $(L,U)$-group fairness constraints, for any $j\in [n]$ and $\ell\in [p]$, the following holds:
            \begin{align*}
                \sum\nolimits_{i\in G_\ell}\sum\nolimits_{t=1}^j R_{ij}
                &\leq \sum\nolimits_{t=1}^{\ceil{j/2}} U_{t\ell}\\
                &=  \sum\nolimits_{t=1}^{\ceil{j/2}} \sum\nolimits_{i\in G_\ell} \inparen{ \inparen{R_{\rm pre}}_{ij} + \inparen{R_{\rm pre}}_{i(j+1)} }\tag{Using \cref{eq:def_2}}\\
                &\leq  \sum\nolimits_{t=1}^{j+1} \sum\nolimits_{i\in G_\ell} { \inparen{R_{\rm pre}}_{ij}}\\
                &\leq  1 + \sum\nolimits_{t=1}^{j} \sum\nolimits_{i\in G_\ell} { \inparen{R_{\rm pre}}_{ij}} \tag{Using that $\inparen{R_{\rm pre}}_{ij}\in \zo$ for all $i$ and $j$}\\
                &\leq  1 + \inparen{U_{\rm pre}}_{j\ell}
                \tag{Using that $R_{\rm pre}$ satisfies the prefix-based $(L_{\rm pre},U_{\rm pre})$-group fairness},
            \end{align*}
            \begin{align*}
                \sum\nolimits_{i\in G_\ell}\sum\nolimits_{t=1}^j R_{ij}
                &\geq \sum\nolimits_{t=1}^{\floor{j/2}} L_{t\ell}\\
                &=  \sum\nolimits_{t=1}^{\floor{j/2}} \sum\nolimits_{i\in G_\ell} \inparen{ \inparen{R_{\rm pre}}_{ij} + \inparen{R_{\rm pre}}_{i(j+1)} }\tag{Using \cref{eq:def_2}}\\
                &\geq  \sum\nolimits_{t=1}^{j-1} \sum\nolimits_{i\in G_\ell} { \inparen{R_{\rm pre}}_{ij}}\\
                &\geq  -1 +  \sum\nolimits_{t=1}^{j} \sum\nolimits_{i\in G_\ell} { \inparen{R_{\rm pre}}_{ij}} \tag{Using that $\inparen{R_{\rm pre}}_{ij}\in \zo$ for all $i$ and $j$}\\
                &\geq  -1 +  \inparen{L_{\rm pre}}_{j\ell}
                \tag{Using that $R_{\rm pre}$ satisfies the prefix-based $(L_{\rm pre},U_{\rm pre})$-group fairness}.
            \end{align*}
            Thus, the result follows.
        \end{proof}
        \noindent Next, we consider a family of individual fairness constraints that captures the constraints considered in \cite{AshudeepUncertainty2021}.
        \begin{definition}[\cite{AshudeepUncertainty2021}]
            Given $C_{\rm pre}\in [0,1]^{m\times n}$, a distribution $\cD$ over the set $\cR$ of all rankings satisfies prefix-based $C_{\rm pre}$-individual fairness constraints if for each $i\in [m]$ and $j\in [n]$
            \begin{align*}
                \Pr\nolimits_{R\sim \cD}\insquare{R_{it} = 1 \text{ for some $1\leq t\leq j$}} \geq \inparen{C_{\rm pre}}_{ij}.
                \yesnum\label{eq:equality_const_indv_fairness_2}
            \end{align*}
        \end{definition}
        \noindent The following lemma shows that the individual fairness constraints in \cref{def:individual_constraints} can approximate the above family of constraints.
        \begin{lemma}
            Let $k=2$ and $q = \frac{n}{2}$.
            For any matrix $C_{\rm pre}\in [0,1]^{m\times n}$ defining \textit{feasible} prefix-based $C_{\rm pre}$-individual fairness constraints, there exists a matrix $C \in [0,1]^{m\times\frac{n}{2}}$ such that any distribution $\cD$ satisfying the $(C,1)$-individual fairness constraints violates prefix-based $C_{\rm pre}$-individual fairness constraints for any item at any position by at most an additive factor of $\max_{1\leq i\leq m}$ $\max_{1\leq j\leq n-1}\inparen{C_{\rm pre}}_{ij}$.
        \end{lemma}
        \noindent Thus, as long as the relative change in the value of the prefix-based  individual fairness constraint at each position is small, any distribution satisfies the $(C,1)$-individual fairness violates the prefix-based individual fairness constraint by a small additive amount.
        \begin{proof}
            Fix $k\coloneqq 2$ and $q=\frac{n}{2}$.
            Let $\cD_{\rm pre}$ be any distribution over rankings that satisfies the prefix-based $C_{\rm pre}$-group fairness.
            $\cD_{\rm pre}$ exists as the prefix-based constraints are feasible.
            We claim that the following matrix $C$ suffices: for each $j\in [\frac{n}{2}]$ and $i\in [m]$
            \begin{align*}
                C_{ij}
                &= \sum\nolimits_{i\in G_\ell} \Pr\nolimits_{R\sim \cD}\insquare{R_{it} = 1 \text{ for some $t\in \inbrace{j,j+1}$}} .\yesnum\label{eq:def_3}
            \end{align*}
            To complete the proof, consider any distribution $\cD$ that satisfies $(C,1)$-individual fairness constraints.
            Since $\cD$ satisfies $(C,1)$-individual fairness constraints, for any $j\in [n]$ and $i\in [m]$, the following holds:
            \begin{align*}
                &\Pr\nolimits_{R\sim \cD}\insquare{R_{it} = 1 \text{ for some $1\leq t\leq j$}}\\
                &\quad\geq \sum\nolimits_{t=1}^{\floor{j/2}} C_{it}\\
                &\quad=  \sum\nolimits_{t=1}^{\floor{j/2}} \Pr\nolimits_{R\sim \cD}\insquare{R_{it} = 1 \text{ for some $t\in \inbrace{j,j+1}$}}
                \tag{Using \cref{eq:def_3}}\\
                &\quad\geq \Pr\nolimits_{R\sim \cD}\insquare{R_{it} = 1 \text{ for some $1\leq t\leq j-1$}}\\
                &\quad\geq \Pr\nolimits_{R\sim \cD}\insquare{R_{it} = 1 \text{ for some $1\leq t\leq j-1$}}  -\max_{1\leq r\leq m}\max_{1\leq s\leq n-1}\inparen{C_{\rm pre}}_{rs}\tag{Using that $\inparen{C_{\rm pre}}_{ij}\leq \max_{1\leq r\leq m;\ 1\leq s\leq n-1}\inparen{C_{\rm pre}}_{rs}$ for all $i$ and $j$}\\
                &\quad\geq \inparen{C_{\rm pre}}_{j\ell} - \max_{1\leq r\leq m}\max_{1\leq s \leq n-1}\inparen{C_{\rm pre}}_{rs}. \tag{Using that $R_{\rm pre}$ satisfies the prefix-based $C_{\rm pre}$-individual fairness}
            \end{align*}
            Thus, the result follows.
        \end{proof}
\end{document}